\documentclass[letterpaper,twocolumn,10pt, usenames,dvipsnames]{article}
\usepackage{usenix-2020-09}

\newif\iflong
\longtrue %

\iflong
\else
\fi

\usepackage{amsmath}
\usepackage{amssymb}
\usepackage{algorithm}
\usepackage[noend]{algpseudocode}
\usepackage{amsthm}
\usepackage{stmaryrd}
\usepackage{xspace}
\usepackage{tabulary}
\usepackage{graphicx}
\usepackage{setspace}
\usepackage{etoolbox}
\usepackage[shortlabels]{enumitem}
\usepackage{hyphenat}

\usepackage[labelfont=bf,skip=10pt,belowskip=0pt,font=small]{caption}
\makeatletter
\patchcmd{\ALG@doentity}{\item[]\nointerlineskip}{}{}{}
\makeatother

\setlist[enumerate]{topsep=2pt}
\setlist[itemize]{topsep=2pt}
\setlist{noitemsep}

\usepackage{tikz}
\usetikzlibrary{arrows}
\newcommand*\circled[1]{\tikz[baseline=(char.base)]{
            \node[shape=circle,fill, color=black,draw,inner sep=.5pt] (char) {\small
              {\bf \textcolor{white}{#1}}};}}

\algrenewcommand\algorithmicindent{1em}

\newcommand{\System}{\textsc{UniStore}\xspace}
\newcommand{\RedBlue}{\textsc{RedBlue}\xspace}
\newcommand{\CureFT}{\textsc{CureFT}\xspace}
\newcommand{\Strong}{\textsc{Strong}\xspace}

\newcommand{\Uniform}{\textsc{Uniform}\xspace}
\newcommand{\Causal}{\textsc{Causal}\xspace}

\newcommand{\SpaceHandler}{\vspace{5pt}}

\definecolor{Background}{gray}{0.85}
\newcommand{\BackColor}{Background}
\newcommand{\StrongColor}{Black}
\newcommand{\CausalTxColor}{NavyBlue}
\newcommand{\StrongTxColor}{BrickRed}

\newcommand{\algline}[2]{line~\ref{#1}:\ref{#2}}
\newcommand{\alglines}[3]{lines~\ref{#1}:\ref{#2}-\ref{#3}}
\newcommand{\algtwolines}[3]{lines~\ref{#1}:\ref{#2} and \ref{#1}:\ref{#3}}

\algblockdefx[upon]{When}{EndWhen}%
  [3][Unknown]{\textbf{when for every} #3 \textbf{received a quorum of} #1(#2)}%
  {\textbf{end}}

\algblockdefx[upon]{WhenRcv}{EndWhenRcv}%
  [1][Unknown]{\textbf{when received} #1}%
  {\textbf{end}}

\algblockdefx[upon]{Upon}{EndUpon}%
  [1][Unknown]{\textbf{upon} #1}%
  {\textbf{end upon}}

\makeatletter
\ifthenelse{\equal{\ALG@noend}{t}}%
  {\algtext*{EndWhenRcv}}
  {}%
\makeatother

\makeatletter
\ifthenelse{\equal{\ALG@noend}{t}}%
  {\algtext*{EndWhen}}
  {}%
\makeatother

\makeatletter
\ifthenelse{\equal{\ALG@noend}{t}}%
  {\algtext*{EndUpon}}
  {}%
\makeatother

\makeatletter
\newcommand\xleftrightarrow[2][]{%
  \ext@arrow 9999{\longleftrightarrowfill@}{#1}{#2}}
\newcommand\longleftrightarrowfill@{%
  \arrowfill@\leftarrow\relbar\rightarrow}
\makeatother

\newcommand{\DC}{\mathcal{D}}

\newcommand{\Partitions}{\mathcal{P}}

\newcommand{\COMMIT}{\textsc{commit}}

\newcommand{\BLUEPREPARED}{\textsc{prepare\_ack}}

\newcommand{\CERTIFY}{\textsc{certify}}

\newcommand{\newtxid}{\mathtt{generate\_tid}}
\newcommand{\partitionof}{\mathtt{partition}}

\newcommand{\MAKEUNIFORM}{\textsc{uniform\_barrier}}
\newcommand{\ATTACH}{\textsc{attach}}
\newcommand{\STARTTX}{\textsc{start\_tx}}
\newcommand{\UPDATETX}{\textsc{do\_op}}
\newcommand{\COMMITBLUE}{\textsc{commit\_causal}}
\newcommand{\EXECOP}{\textsc{get\_version}}

\newcommand{\OPRET}{\textsc{version}}
\newcommand{\BLUEPREPARE}{\textsc{prepare}}

\newcommand{\PROPAGATELOCAL}{\textsc{propagate\_local\_txs}}
\newcommand{\PROPAGATEREMOTE}{\textsc{forward\_remote\_txs}}
\newcommand{\REPLICATE}{\textsc{replicate}}
\newcommand{\BHEARTBEAT}{\textsc{heartbeat}}
\newcommand{\RHEARTBEAT}{\textsc{heartbeat\_strong}}

\newcommand{\UPDRCV}{\textsc{knownvec\_global}}
\newcommand{\UPDCSS}{\textsc{knownvec\_local}}
\newcommand{\UPDUSS}{\textsc{stablevec}}
\newcommand{\CASTVC}{\textsc{broadcast\_vecs}}
\newcommand{\COMMITRED}{\textsc{commit\_strong}}

\newcommand{\DELIVERUPD}{\textsc{deliver\_updates}}

\newcommand{\partition}{\ensuremath{p}}

\newcommand{\clock}{{\sf clock}}

\newcommand{\past}{{\sf pastVec}}

\newcommand{\red}{\textsl{strong}}

\newcommand{\replicavectorclock}{{\sf knownVec}}
\newcommand{\areplicavectorclock}{\ensuremath{\mathit{knownVec}}}
\newcommand{\avecsnapshottime}{\ensuremath{\mathit{snapVec}}}
\newcommand{\vecsnapshottime}{\ensuremath{\mathsf{snapVec}}}

\newcommand{\cstate}{\ensuremath{\mathit{state}}}

\newcommand{\effval}[2]{\ensuremath{{\sf retval}(#1, #2)}}

\newcommand{\apply}{\ensuremath{{\sf apply}}}
\newcommand{\op}{\ensuremath{\mathit{op}}}
\newcommand{\ts}{\ensuremath{\mathit{ts}}}

\newcommand{\rvc}{{\sf globalMatrix}}
\newcommand{\pmc}{{\sf localMatrix}}
\newcommand{\gmc}{{\sf stableMatrix}}
\newcommand{\commitvector}{\ensuremath{\mathit{commitVec}}}

\newcommand{\stablesnapshot}{{\sf stableVec}}
\newcommand{\astablesnapshot}{\ensuremath{\mathit{stableVec}}}
\newcommand{\uniformsnapshot}{{\sf uniformVec}}

\newcommand{\tx}{\ensuremath{\mathit{tid}}}

\newcommand{\preparedblue}{{\sf preparedCausal}}
\newcommand{\committedset}{{\sf committedCausal}}
\newcommand{\store}{{\sf opLog}}

\newcommand{\aws}{\ensuremath{\mathit{wbuff}}}

\newcommand{\readset}{{\sf rset}}
\newcommand{\writeset}{{\sf wbuff}}

\newcommand{\atxs}{\ensuremath{\mathit{txs}}}

\newcommand{\avc}{\ensuremath{V}}

\newcommand{\areplica}{l}
\newcommand{\asecondreplica}{n}
\newcommand{\athirddc}{h}

\let\ensuremathorig=\ensuremath
\renewcommand{\ensuremath}[1]{\ensuremathorig{#1}\xspace}

\renewcommand{\emptyset}{\varnothing}

\renewcommand{\implies}{\ensuremath{\Rightarrow}}

\renewcommand{\iff}{\Leftrightarrow}

\newtheoremstyle{mytheorem}
  {5pt} %
  {5pt} %
  {} %
  {} %
  {\bfseries} %
  {.} %
  {.4em} %
  {} %

\theoremstyle{mytheorem}

\newtheorem{definition}{\sc Definition}
\newtheorem{property}{\sc Property}

\newboolean{showcomments}
\setboolean{showcomments}{true}
\ifthenelse{\boolean{showcomments}}
{ \newcommand{\mynote}[2][fromWhom?]{%
    {\color{red}\small%
      $\blacktriangleright$%
      \fbox{\bfseries\sffamily\scriptsize#1}%
      \textsf{\em #2}%
      $\blacktriangleleft$%
}}}
{ \newcommand{\mynote}[2][fromWhom?]{}}

\newcommand{\liveness}{Eventual Visibility}

\makeatletter

\def\section{\@startsection {section}{1}{\z@}
  {-3ex plus-1ex minus -.2ex}{2ex plus.2ex}{\reset@font\large\bf}}

\renewcommand\subsection{\@startsection{subsection}{2}{\z@}%
  {-.75\baselineskip \@plus -2\p@ \@minus -.2\p@}%
  {.3\baselineskip}%
  {\reset@font\bf}}

\renewcommand\paragraph{\@startsection{subparagraph}{5}{\z@}%
  {-.5\baselineskip \@plus -1\p@ \@minus -.2\p@}%
  {-3.5\p@}%
  {\normalsize\bfseries}}

\makeatother

\newcommand{\tr}[2]{\iflong{}\S#1\else{}\cite[\S#2]{ext}\fi}

\newcommand{\nappfull}{A}
\newcommand{\nappconsistency}{B}
\newcommand{\napptcs}{C}
\newcommand{\nappproof}{D}

\theoremstyle{definition}
\newtheorem{apptheorem}{Theorem}
\newtheorem{applemma}[apptheorem]{Lemma}
\newtheorem{appdefinition}[apptheorem]{Definition}
\newtheorem{appassumption}{\textsc{Assumption}}

\newcommand{\set}[1]{\{#1\}}
\DeclareMathOperator*{\argmin}{argmin}

\newcommand{\hStatex}{\vspace{4pt}}

\newcommand{\causalentry}{\textsl{causal}}
\newcommand{\strongentry}{\textsl{strong}}

\newcommand{\IfThenElse}[3]{
  \State \algorithmicif\ #1\ \algorithmicthen\ #2\ \algorithmicelse\ #3}

\newcommand{\redcolor}[1]{\textcolor{red}{#1}}
\newcommand{\bluecolor}[1]{\textcolor{blue}{#1}}

\newcommand{\purple}[1]{\textcolor{purple}{#1}}

\newcommand{\lccolor}[1]{\bluecolor{#1}}
\newcommand{\tscolor}[1]{\purple{#1}}
\newcommand{\strongcolor}[1]{\redcolor{#1}}

\newcommand{\unistore}{\textsc{UniStore}}
\newcommand{\code}[2]{#1:#2}

\newcommand{\pastVC}{{\sf pastVec}}
\newcommand{\ok}{\mathrm{ok}}

\newcommand{\var}{{\bf var}\;}
\newcommand{\pre}{{\bf pre:}\;}
\newcommand{\rpc}{{\bf remote\; call}\;}
\newcommand{\at}{\;{\bf at}\;}
\newcommand{\send}{{\bf send}\;}
\newcommand{\sendto}{\;{\bf to}\;}
\newcommand{\from}{\;{\bf from}\;}

\newcommand{\notkw}{{\bf not}\;}
\newcommand{\timeout}{{\bf timeout}\;}
\newcommand{\wait}{{\bf wait}\;}
\newcommand{\asyncwait}{{\bf async wait}\;}
\newcommand{\until}{{\bf until}\;}
\newcommand{\receive}{{\bf receive}\;}
\newcommand{\received}{{\bf received}\;}

\newcommand{\upcall}{{\bf upcall}\;}
\newcommand{\prop}[1]{\textbf{\textsc{Property #1}}}

\newcommand{\N}{\mathbb{N}}
\newcommand{\D}{\mathcal{D}}
\newcommand{\C}{\mathcal{C}}
\newcommand{\OP}{O}
\newcommand{\cdrange}{c} %
\renewcommand{\P}{\mathcal{P}}
\newcommand{\opfunc}{\textsl{op}}

\newcommand{\p}{\mathit{p}}
\newcommand{\decvar}{\mathit{dec}}
\newcommand{\clientdc}{{\sf d}}
\newcommand{\cldc}{\textsl{cur\_dc}}
\newcommand{\clientcoord}{{\sf p}}
\newcommand{\starttx}{\textsc{start\_tx}}
\newcommand{\start}{\textsc{start}}
\renewcommand{\read}{\textsc{read}}
\newcommand{\updateproc}{\textsc{update}}
\newcommand{\commit}{\textsc{commit}}
\newcommand{\commitcausaltx}{\textsc{commit\_causal\_tx}}
\newcommand{\commitstrongtx}{\textsc{commit\_strong\_tx}}
\newcommand{\abort}{\textsc{abort}}
\newcommand{\commitcausal}{\textsc{commit\_causal}}
\newcommand{\commitstrong}{\textsc{commit\_strong}}

\newcommand{\vc}{V}
\newcommand{\vcvar}{\mathit{vc}}
\newcommand{\generatetid}{\texttt{generate\_tid}}
\newcommand{\partitionofproc}{\mathtt{partition}}
\newcommand{\doop}{\textsc{do\_op}}
\newcommand{\doread}{\textsc{do\_read}}
\newcommand{\doupdate}{\textsc{do\_update}}

\newcommand{\find}{\texttt{find}}
\newcommand{\readkey}{\textsc{get\_version}}
\newcommand{\getversion}{\textsc{get\_version}}
\newcommand{\versionproc}{\textsc{version}}
\newcommand{\prepare}{\textsc{prepare}}
\newcommand{\prepareack}{\textsc{prepare\_ack}}

\newcommand{\heartbeat}{\textsc{heartbeat}}
\newcommand{\replicate}{\textsc{replicate}}
\newcommand{\propagate}{\textsc{propagate\_local\_txs}}
\newcommand{\forward}{\textsc{forward\_remote\_txs}}

\newcommand{\bcast}{\textsc{broadcast\_vecs}}
\newcommand{\knownvclocal}{\textsc{knownvec\_local}}
\newcommand{\stablevcproc}{\textsc{stablevec}}
\newcommand{\knownvcglobal}{\textsc{knownvec\_global}}

\renewcommand{\log}{\textsl{log}}
\newcommand{\oplog}{{\sf opLog}}

\newcommand{\ep}{{\sf ep}}

\newcommand{\events}{\textsl{events}}
\newcommand{\ws}{\textsl{ws}}
\newcommand{\rs}{\textsl{rs}}
\newcommand{\rsvar}{\mathit{rs}}
\newcommand{\key}{\textsl{key}}
\newcommand{\client}{\textsl{client}}
\newcommand{\clients}{\mathbb{C}}
\newcommand{\dc}{\textsl{dc}}
\newcommand{\coord}{\textsl{coord}}
\newcommand{\startoftx}{\textsl{st}}
\newcommand{\commitoftx}{\textsl{ct}}
\newcommand{\partitionsfunc}{\textsl{partitions}}
\newcommand{\snapshotVC}{\textsl{snapshotVec}}
\newcommand{\snapshotproc}{\texttt{snapshot}}
\newcommand{\commitVC}{\textsl{commitVec}}

\newcommand{\certfunc}{F}
\newcommand{\fdec}{\certfunc_{\rm dec}}
\newcommand{\fvec}{\certfunc_{\rm vec}}
\newcommand{\flc}{\certfunc_{\rm lc}}

\newcommand{\ud}{\textsl{ud}} %
\newcommand{\W}{W}
\newcommand{\R}{R}
\newcommand{\payload}{W}

\newcommand{\txs}{T}
\newcommand{\txsvar}{\mathit{txs}}
\newcommand{\txsincertify}{T_{c}}
\newcommand{\causaltxs}{\txs_{\causalentry}}
\newcommand{\strongtxs}{\txs_{\strongentry}}
\newcommand{\allstrongtxs}{\txs_{\textsl{all-strong}}}

\newcommand{\conflict}{\bowtie}

\newcommand{\intcertify}{{\sf certify}}
\newcommand{\intdecide}{{\sf decide}}
\newcommand{\intdeliver}{{\sf deliver}}

\newcommand{\act}{{\sf act}}
\newcommand{\actvar}{\mathit{act}}

\newcommand{\lccertify}{C}

\newcommand{\Decision}{\mathbb{D}}
\newcommand{\Vector}{\mathbb{V}}
\newcommand{\Key}{\mathit{Key}}
\newcommand{\Val}{\mathit{Val}}

\newcommand{\certify}{\textsc{certify}}
\newcommand{\preparestrong}{\textsc{prepare\_strong}}
\newcommand{\heartbeatstrong}{\textsc{heartbeat\_strong}}
\newcommand{\alreadydecided}{\textsc{already\_decided}}
\newcommand{\acceptack}{\textsc{accept\_ack}}
\newcommand{\replicas}{\textsc{replicas}}
\newcommand{\leaderproc}{\textsc{leader}}
\newcommand{\leaderof}{{\sf leader}}
\newcommand{\decision}{\textsc{decision}}
\newcommand{\decisionvar}{\mathit{decision}}
\newcommand{\learndecision}{\textsc{learn\_decision}}
\newcommand{\follower}{\textsc{follower}}
\newcommand{\certification}{\textsc{certification\_check}}
\newcommand{\deliver}{\textsc{deliver}}
\newcommand{\deliverupdates}{\textsc{deliver\_updates}}
\newcommand{\newleader}{\textsc{new\_leader}}
\newcommand{\newleaderack}{\textsc{new\_leader\_ack}}
\newcommand{\newstate}{\textsc{new\_state}}
\newcommand{\newstateack}{\textsc{new\_state\_ack}}
\newcommand{\retry}{\textsc{retry}}

\newcommand{\votevar}{\mathit{vote}}

\newcommand{\reqIdVar}{\mathit{rid}}
\newcommand{\generateReqId}{\texttt{generate\_req\_id}}
\newcommand{\trustedVar}{{\sf trusted}}
\newcommand{\recover}{\textsc{recover}}
\newcommand{\nack}{\textsc{nack}}
\newcommand{\normalMode}{\textsc{normal}}
\newcommand{\recovering}{\textsc{recovering}}
\newcommand{\restoring}{\textsc{restoring}}
\newcommand{\doNotWaitFor}{{\sf doNotWaitFor}}
\newcommand{\unknownTx}{\textsc{unknown\_tx}}
\newcommand{\unknowntxAck}{\textsc{unknown\_tx\_ack}}
\newcommand{\unknowntx}{\textsc{unknown}}
\newcommand{\callerMode}{\mathit{callerMode}}
\newcommand{\senderMode}{\mathit{senderMode}}

\newcommand{\dvar}{\mathit{d}}
\newcommand{\mvar}{\mathit{m}}

\newcommand{\pvar}{\mathit{p}}

\newcommand{\cl}{\mathit{cl}}

\newcommand{\snapvc}{\mathit{snapVec}}
\newcommand{\commitvc}{\mathit{commitVec}}
\newcommand{\knownvc}{\mathit{knownVec}}
\newcommand{\stablevc}{\mathit{stableVec}}
\newcommand{\preparedstrongvar}{\mathit{preparedStrong}}
\newcommand{\decidedstrongvar}{\mathit{decidedStrong}}

\newcommand{\knownVC}{{\sf knownVec}}
\newcommand{\stableVC}{{\sf stableVec}}
\newcommand{\snapVC}{{\sf snapVec}}
\newcommand{\uniformVC}{{\sf uniformVec}}
\newcommand{\localmatrix}{{\sf localMatrix}}
\newcommand{\globalmatrix}{{\sf globalMatrix}}
\newcommand{\stablematrix}{{\sf stableMatrix}}

\newcommand{\realtime}{\tau}
\newcommand{\txfunc}{\textsl{tx}}

\newcommand{\tidselector}{\textsl{tid}}
\newcommand{\tidvar}{\mathit{tid}}
\newcommand{\tids}{{\sf TID}}
\newcommand{\tvar}{\mathit{t}}
\newcommand{\ctid}{{\sf tid}} %
\newcommand{\rset}{{\sf rset}}
\newcommand{\rsetvar}{\mathit{rset}}
\newcommand{\wbuff}{{\sf wbuff}}
\newcommand{\wbuffvar}{\mathit{wbuff}}
\newcommand{\clockVar}{{\sf clock}}
\newcommand{\statusVar}{{\sf status}}
\newcommand{\committedVar}{{\sf committed}}
\newcommand{\committedcausal}{{\sf committedCausal}}
\newcommand{\preparedcausal}{{\sf preparedCausal}}
\newcommand{\preparedstrong}{{\sf preparedStrong}}
\newcommand{\decidedstrong}{{\sf decidedStrong}}
\newcommand{\accept}{\textsc{accept}}
\newcommand{\ballotVar}{{\sf ballot}}
\newcommand{\ballotvar}{\mathit{b}}
\newcommand{\cballot}{{\sf cballot}}
\newcommand{\cballotvar}{\mathit{cballot}}
\newcommand{\lastdeliveredVar}{{\sf lastDelivered}}

\newcommand{\Fence}{Q}
\newcommand{\fencerange}{q}
\newcommand{\fence}{\textsc{cl\_uniform\_barrier}}
\newcommand{\uniformbarrier}{\textsc{uniform\_barrier}}

\newcommand{\Attach}{A}
\newcommand{\attachrange}{a}
\newcommand{\clattach}{\textsc{cl\_attach}}
\newcommand{\attach}{\textsc{attach}}

\newcommand{\maxPrep}{\mathit{maxPrep}}
\newcommand{\maxDec}{\mathit{maxDec}}

\newcommand{\intread}{R_{\textsc{int}}}
\newcommand{\extread}{R_{\textsc{ext}}}
\newcommand{\rocommit}{C_{\textsc{ro}}}
\newcommand{\updatecommit}{C_{\textsc{rw}}}

\newcommand{\tsfunc}{\textsl{ts}}
\newcommand{\tsvar}{\mathit{ts}}
\newcommand{\lc}{{\sf lc}}
\newcommand{\lcvar}{\mathit{lc}}
\newcommand{\lclock}{\textsl{lclock}}
\newcommand{\lcorder}{\textsl{lc}}
\newcommand{\timestamp}{\textsl{ts}}

\newcommand{\eok}{\textsl{eo}_{k}}
\newcommand{\po}{\textsl{po}}
\newcommand{\so}{\textsl{so}}
\newcommand{\vis}{\textsl{vis}}
\newcommand{\ar}{\textsl{ar}}
\newcommand{\rel}[1]{\xrightarrow{#1}} %
\newcommand{\relation}{\mathcal{R}}

\newcommand{\txevents}{s}
\newcommand{\txnevents}{Y}

\newcommand{\rval}{\textsl{rval}}
\newcommand{\uval}{\textsl{uval}}

\newcommand{\retval}{\textsc{RVal}}
\newcommand{\intretval}{\textsc{IntRVal}}
\newcommand{\extretval}{\textsc{ExtRVal}}
\newcommand{\cc}{\textsc{CausalConsistency}}
\newcommand{\cv}{\textsc{CausalVisibility}}
\newcommand{\ca}{\textsc{CausalArbitration}}
\newcommand{\ev}{\textsc{EventualVisibility}}
\newcommand{\por}{\textsc{PoR}}
\newcommand{\conflictaxiom}{\textsc{ConflictOrdering}}

\begin{document}

\title{\System: A fault-tolerant marriage of causal and strong consistency}

\author{
{\rm Manuel Bravo} \qquad
{\rm Alexey Gotsman} \qquad
{\rm Borja de R\'egil} \qquad
\\[2pt]
IMDEA Software Institute
\and
\and
{\rm Hengfeng Wei~\thanks{Also with the State Key Laboratory for Novel Software Technology, Software Institute.}}\\[2pt]
Nanjing University
}

\maketitle

\begin{abstract}
Modern online services rely on data stores that replicate their data across
geographically distributed data centers. %
Providing strong consistency in such data stores results in high latencies and
makes the system vulnerable to network partitions.
The alternative of relaxing consistency violates crucial correctness properties.
A compromise is to allow multiple consistency levels to
coexist in the data store. In this paper we present \System, the first fault-tolerant and
scalable data store that combines causal and strong consistency.
The key challenge we address in \System is to maintain liveness despite data
center failures: this could be compromised if a strong transaction takes a
dependency on a causal transaction that is later lost because of a failure.
\System ensures that such situations do not arise while paying the cost of
durability for causal transactions only when necessary.
We evaluate \System on Amazon EC2 using both microbenchmarks and a
sample application. Our results show that
\System effectively and scalably combines causal and strong
consistency.

\end{abstract}

\section{Introduction}
\label{sec:intro}

Many of today's Internet services rely on geo-distributed data stores, which
replicate data in different geographical locations. This improves user
experience by allowing accesses to the closest site and ensures
disaster-tolerance. However, geo-distribution also makes it more challenging to
keep the data consistent. The classical approach is to make replication
transparent to clients by providing strong consistency models, such as
linearizability~\cite{linearizability} or
serializability~\cite{wv}. The downside is that this approach requires
synchronization between data centers in the critical path. This significantly
increases latency~\cite{pacelc} and makes the system unavailable during network
partitionings~\cite{cap}. Thus, even though several commercial geo-distributed
systems follow this approach~\cite{spanner,cockroachdb,yugabytedb,faunadb,giza},
the associated cost has prevented it from being adopted more widely.

An alternative approach is to relax synchronization: the data store executes an
operation at a single data center, without any communication with others, and
propagates updates to other data centers in the
background~\cite{bayou,dynamo}. This minimizes the latency and makes the system
{\em highly available}, i.e., operational even during network partitionings. But
on the downside, the systems following this approach provide weaker consistency
models: e.g., eventual consistency~\cite{bayou,vogels} or {\em causal
  consistency}~\cite{causal-memory}. The latter is particularly appealing: it
guarantees that clients see updates in an order that respects the potential
causality between them.
For example, assume that in a banking application Alice deposits \$100 into
Bob's account ($u_1$) and then posts a notification about it into Bob's inbox
($u_2$). Under causal consistency, if Bob sees the notification ($u_3$), and
then checks his account balance ($u_4$), he will see the deposit. This is not
guaranteed under eventual consistency, which does not respect causality
relationships, such as those between $u_1$ and $u_2$. In some settings, causal
consistency has been shown to be the strongest model that
allows availability during network partitionings~\cite{hagit-cc,alvisi-cc}. It
has been a subject of active research in recent years, with scalable
implementations~\cite{cure,cops,occult} and some industrial
deployments~\cite{MongoDB,redis}.

However, even causal consistency is often too weak to preserve critical
application invariants. For example, consider a banking application that
disallows overdrafts and thus maintains an invariant that an account balance is
always non-negative. Assume that the %
balance of an account stored at two replicas is $100$, and clients concurrently
issue two ${\tt withdraw}(100)$ operations ($u_5$ and $u_6$) at different
replicas. Since causal consistency executes operations without synchronization,
both withdrawals will succeed, and once the replicas exchange the updates, the
balance will go negative. To ensure integrity invariants in examples such as
this, the programmer has to introduce synchronization between replicas, and,
since synchronization is expensive, it pays off to do this sparingly. To this
end, several research~\cite{red-blue,valter,pileus,por} and
commercial~\cite{cosmosdb,documentdb,google,dynamodb,cassandra} data stores
allow the programmer to choose whether to execute a particular operation under
weak or strong consistency. For example, to preserve the integrity invariant in
our banking application, only withdrawals need to use strong consistency, and
hence, synchronize; deposits may use weaker consistency and proceed without
synchronization.

Given the benefits of causal consistency, it is particularly appealing to marry
it with strong consistency in a geo-distributed data store. But like real-life
marriages, to be successful this one needs to hold together both in good times
and in bad -- when data centers fail due to catastrophic events or power
outages. Unfortunately, none of the existing data stores meant for
geo-replication combine causal and
strong consistency while providing such fault tolerance~\cite{valter,red-blue,por}.
In this paper we present \System{} -- the first fault-tolerant and scalable data
store that combines causal and strong consistency. More precisely, \System
implements a transactional variant of {\em Partial Order-Restrictions
  consistency (PoR consistency)}~\cite{por,cise-popl16}. This guarantees
transactional causal consistency by default~\cite{cure} and allows the
programmer to additionally specify which pairs of transactions \emph{conflict},
i.e., have to synchronize. For instance, to preserve the integrity invariant in
our previous example, the programmer should declare that withdrawals from the
same account conflict. Then one of the withdrawals $u_5$ and $u_6$ will observe
the other and will fail.

The key challenge we have to address in \System is to maintain liveness despite
data center failures. Just adding a Paxos-based commit protocol for strong
transactions~\cite{discpaper,spanner,mdcc} to an existing causally consistent
protocol does not yield a fault-tolerant data store.
In such a data store, a committed strong transaction $t_2$ may never become
visible to clients if a causal transaction $t_1$ on which it depends is lost due
to a failure of its origin data center. This compromises the liveness of the
system, because no transaction $t_3$ conflicting with $t_2$ can commit
from now on: according to the PoR model, one of the transactions $t_2$ and
$t_3$ has to observe the other, but $t_2$ will never be visible and
$t_2$ did not observe $t_3$.

\System addresses this problem by ensuring that, before a strong transaction
commits, all its causal dependencies are \emph{uniform}, i.e., will eventually
become visible at all correct data centers. This adapts the classical notion of
uniformity in distributed computing to causal
consistency~\cite{cachin-book}. \System does so without defeating the benefits
of causal consistency. Causal transactions remain
highly available at the cost of increasing the latency of strong
transactions: a strong transaction may have to wait for some of its
dependencies to become uniform before committing. To minimize this cost,
\System executes causal transactions on a snapshot that is slightly in the
past, such that a strong transaction will mostly depend on causal
transactions that are already uniform before committing.
Furthermore, \System reuses the mechanism for
tracking uniformity to let clients make causal transaction durable on
demand and to enable consistent client migration.

In addition to being fault tolerant, \System scales horizontally, i.e., with the
number of machines in each data center;
this also goes beyond previous proposals~\cite{valter,red-blue,por}. To this
end, \System builds on Cure~\cite{cure} -- a scalable implementation of
transactional causal consistency. Our protocol extends Cure with a novel
mechanism that distributes the task of tracking the set of uniform transactions
among the machines of a data center.
We also add the ability for data centers to
forward transactions they receive from others, so that a transaction can
propagate through the system even if its origin data center fails.
Finally, we carefully integrate an existing fault-tolerant atomic commit for
strong transactions~\cite{discpaper} into the protocol for causal consistency.

We have rigorously proved the correctness of the \System protocol
(\S\ref{sec:correctness} and \tr{\ref{section:correctness-proof}}{\nappproof}).
We have also evaluated it on Amazon EC2 using both microbenchmarks and a more
realistic RUBiS benchmark. Our evaluation demonstrates that \System scalably
combines causal and strong transactions, with the latter not affecting the
performance of the former. Under the RUBiS mix workload, causal transactions
exhibit a low latency (1.2ms on average), and the overall average latency is
3.7$\times$ lower than that of a strongly consistent system.

\section{System Model}
\label{sec:sysmodel}

We consider a geo-distributed system consisting of a set of data centers
$\DC=\{1, \dots, D\}$ that manage a large set of data items. A data item is
uniquely identified by its \emph{key}. For scalability, the key space is split
into a set of logical partitions $\Partitions = \{1, \ldots, N\}$.  Each data
center stores replicas of all partitions, scattered among its servers.
We let $\partition^m_d$ be the replica of partition $m$ at data center $d$, and
we refer to replicas of the same partition as \emph{sibling} replicas.
As is standard, we assume that $D = 2f+1$ and at most $f$ data centers may fail.
We call a data center that does not fail {\em correct}. If a data center fails,
all partition replicas it stores become unavailable. For simplicity, we do not
consider the failures of individual replicas within a data center: these can be
masked using standard state-machine replication protocols executing within a
data center~\cite{smr,paxos}.

Replicas have physical clocks, which are loosely synchronized by a
protocol such as NTP. The correctness of \System does not depend on the
precision of clock synchronization, but large %
drifts may negatively impact its performance.
Any two replicas are connected by a reliable FIFO channel, so that 
messages between correct data centers are guaranteed to be %
delivered.
As is standard, to implement strong consistency we require the network to be
{\em eventually synchronous}, so that message delays between sibling replicas in correct
data centers are eventually bounded by some constant~\cite{psync}.

\section{Consistency Model}
\label{sec:consistency}

A client interacts with \System by executing a stream of {\em transactions} at
the data center it is connected to. A transaction consists of a sequence of
operations, each on a single data item, and can be {\em interactive}: the data
items it accesses are not known a priori. A transaction that modifies at least
one data item is an {\em update} transaction; otherwise it is {\em read-only}.

A {\em consistency model} defines a contract between the data store and its
clients that specifies which values the data store is allowed to return in
response to client operations. \System implements a transactional variant of
{\em Partial Order-Restrictions consistency (PoR
  consistency)}~\cite{por,cise-popl16}, which we now define informally; we give
a formal definition in~\tr{\ref{section:spec}}{\nappconsistency}.
The PoR model enables the programmer to classify transactions as either {\em causal} or
{\em strong}. Causal transactions satisfy transactional {\em causal
  consistency}, which guarantees that clients see transactions in an order that
respects the potential causality between them~\cite{causal-memory,cure}.
However, clients can observe causally independent transactions in arbitrary
order. Strong transactions give the programmer more control over their
visibility. To this end, the programmer provides a symmetric {\em conflict
  relation} $\bowtie$ on operations
that is lifted to strong transactions as follows: two transactions conflict
if they perform conflicting operations on the same data item. Then the PoR model
guarantees that,
out of two conflicting strong transactions, one has to observe the
other.

More precisely, a transaction $t_1$ precedes a transaction $t_2$ in the {\em
  session order} if they are executed by the same client and $t_1$ is executed
before $t_2$. A set of transactions $T$ committed by the data store satisfies
PoR
consistency if there exists a {\em causal order}
relation $\prec$ on $T$ such that the following properties hold:
\begin{description}[noitemsep,topsep=3pt,parsep=3pt,leftmargin=10pt]
\item[Causality Preservation.] The relation $\prec$ is transitive, irreflexive,
  and includes the session order.
\item[Return Value Consistency.] Consider an operation $u$ on a data item $k$ in
  a transaction $t \in T$. The return value of $u$ can be computed from the
  state of $k$ obtained as follows: first execute
  all operations on $k$ by transactions preceding $t$ in $\prec$ in an order
  consistent with $\prec$; then execute
  all operations on $k$ that precede $u$ in $t$.
\item[Conflict Ordering.] For any distinct strong transactions $t_1, t_2 \in T$,
  if $t_1\bowtie t_2$, then either $t_1\prec t_2$ or $t_2\prec t_1$.
\item[Eventual Visibility.] A transaction $t \in T$ that is either strong or
  originates at a correct data center eventually becomes visible at all correct
  data centers: from some point on, $t$ precedes in $\prec$ all transactions
  issued at correct data centers.
\end{description}
If all transactions are causal, then the above definition specializes to
transactional causal consistency~\cite{cure,framework-concur15}. If all
transactions are strong and all pairs of operations conflict, then we obtain
(non-strict) serializability.

When $t_1 \prec t_2$, we say that $t_1$ is a {\em causal dependency} of $t_2$.
Return Value Consistency ensures that all operations in a transaction $t$
execute on a snapshot consisting of its causal dependencies (as well as prior
operations by $t$). Transactions are atomic, so that either all of their
operations are included into the snapshot or none at all. The transitivity of
$\prec$, mandated by Causality Preservation, ensures that the snapshot a
transaction executes on is \emph{causally consistent}: if a transaction $t_1$ is
included into the snapshot, then so is any other transaction $t_2$ on which
$t_1$ depends (i.e., $t_2\prec t_1$). The inclusion of the session order into
$\prec$, also mandated by Causality Preservation, ensures session guarantees
such as {\em read your writes}~\cite{session}. The consistency model disallows
the causality violation anomaly from \S\ref{sec:intro}. Indeed, since $\prec$
includes the session order, we have $u_1 \prec u_2$ and $u_3 \prec
u_4$. Moreover, Bob sees Alice's message, and by Return Value Consistency this
can only happen if $u_2 \prec u_3$. Then since $\prec$ is transitive,
$u_1 \prec u_4$, and by Return Value Consistency, Bob has to see Alice's
deposit.

Causal consistency nevertheless allows the overdraft anomaly from
\S\ref{sec:intro}: the withdrawals $u_5$ and $u_6$ may not be related by
$\prec$, and thus may both execute on the balance $100$ and succeed. The
Conflict Ordering property can be used to disallow this anomaly by declaring
that ${\tt withdraw}$ operations on the same account conflict and labeling
transactions containing these as strong.
Then one of the withdrawal transactions will be guaranteed to causally precede
the other. The latter will be executed on the account balance $0$ and will fail.

Finally, Eventual Visibility ensures that strong transactions and those causal
ones that originate at correct data centers are durable, i.e., will eventually
propagate through the system.

To facilitate the use of causal transactions, \System includes {\em replicated
  data types} (aka CRDTs), which implement policies for merging concurrent
updates to the same data item~\cite{crdts}.
Each data item in the store is associated with a type (e.g., counter, set),
which is backed by a CRDT implementation managing updates to it. For example,
the programmer can use a counter CRDT to represent an account balance. Then if
two clients concurrently deposit $100$ and $200$ into an empty account using
causal transactions, eventually the balance at all replicas will be $300$. Using
ordinary writes here would yield $100$ or $200$, depending on the order in which
the writes are applied. More generally, CRDTs ensure
that two replicas receiving the same set of updates are in the same state,
regardless of the receipt order. Together with Eventual Visibility, this implies
the expected guarantee of eventual consistency~\cite{bayou}. Due to space
constraints, we omit details about the use of CRDTs from our protocol
descriptions.

\section{Key Design Decisions in \System}
\label{sec:overview}

\begin{figure*}[t]
\begin{tabular}{@{}c c@{}}
\begin{minipage}[t]{5.7cm}
\includegraphics[width=\textwidth]{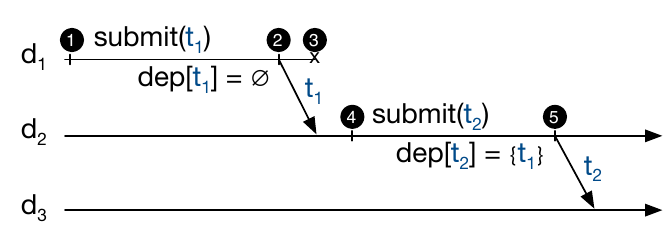}
\caption{Why \System may need to forward remote causal transactions.}
\label{fig:execution-causal}
\end{minipage}
&
\begin{minipage}[t]{11,5cm}
\includegraphics[width=\textwidth]{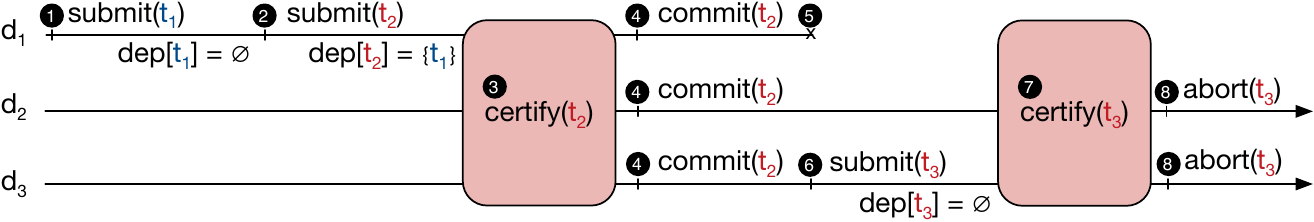}
\caption{Why \System needs to ensure that the dependencies of a strong
  transaction are uniform before committing it.}
\label{fig:execution-strong}
\end{minipage}
\end{tabular}
\end{figure*}

\paragraph{Baseline causal consistency.} A causal transaction in \System
first executes at a single data center on a causally consistent snapshot.
After this it immediately commits, and its updates are replicated to all other
data centers in the background. This minimizes the latency of causal
transactions and makes them highly available, i.e., they can be executed even
when the network connections between data centers fail.

As is common in causally consistent data stores~\cite{cops,gentlerain,cure}, to
ensure that causal transactions execute on consistent snapshots, a data center
exposes a remote transaction to clients only after exposing all its
dependencies. Then to satisfy the \liveness{} property under failures, a data
center receiving a remote causal transaction may need to forward it to other
data centers, as in reliable
broadcast~\cite{isis-reliable} and anti-entropy protocols for replica reconciliation~\cite{anti-entropy}.

Figure~\ref{fig:execution-causal}
depicts a scenario that demonstrates how \liveness{} could be
violated in the absence of this mechanism.
Let {\color{\CausalTxColor}$t_1$} be a causal %
transaction submitted at a data center $d_1$ (event \circled{1}). Assume that
$d_1$ replicates {\color{\CausalTxColor}$t_1$} to a correct data center $d_2$
(event \circled{2}) and then fails (event \circled{3}), so that
{\color{\CausalTxColor}$t_1$} does not get replicated anywhere else.
Let {\color{\CausalTxColor}$t_2$} be a %
transaction submitted at $d_2$ after {\color{\CausalTxColor}$t_1$} becomes
visible there, so that {\color{\CausalTxColor}$t_2$} depends on
{\color{\CausalTxColor}$t_1$} (event \circled{4}). Transaction
{\color{\CausalTxColor}$t_2$} will eventually be replicated to all correct data
centers (event \circled{5}). But it will never be exposed at any of them,
because its dependency {\color{\CausalTxColor}$t_1$} is missing. If data
centers can forward remote causal transactions, then $d_2$ can
eventually replicate {\color{\CausalTxColor}$t_1$} to all correct data centers,
preventing this problem.

\paragraph{On-demand strong consistency.}
\System uses optimistic concurrency control for strong transactions: they are
first executed speculatively and the results are then {\em certified} to
determine whether the transaction can commit, or must abort due to a conflict
with a concurrent strong transaction~\cite{wv}.
Certifying a strong transaction requires synchronization between the replicas of
partitions it accessed, located in different data centers.
\System implements this using an existing
fault-tolerant protocol that combines two-phase commit and
Paxos~\cite{discpaper} while minimizing commit latency.
However, just using such a protocol is not enough to make the overall system
fault tolerant: for this, before a strong transaction commits, all its causal
dependencies must be uniform in the following sense.
\begin{definition}
  A transaction is {\em uniform} if both the transaction and its causal
  dependencies are guaranteed to be eventually replicated at all correct data
  centers.
\label{def:uniform}
\end{definition}
This adapts the classical notion of uniformity in distributed computing to
causal consistency~\cite{cachin-book}. \System considers a transaction to be
uniform once it is visible at $f+1$ data centers, because at least one of these
must be correct, and
data centers can forward causal transactions to others.

The following scenario, depicted in Figure~\ref{fig:execution-strong},
demonstrates why committing a strong transaction before its dependencies become
uniform can compromise the liveness of the system.
Assume that a causal transaction {\color{\CausalTxColor}$t_1$} and a strong
transaction {\color{\StrongTxColor}$t_2$} are submitted at a data center $d_1$
in such a way that {\color{\CausalTxColor}$t_1$} becomes a dependency of
{\color{\StrongTxColor}$t_2$} (events \circled{1} and \circled{2}). Assume also
that {\color{\StrongTxColor}$t_2$} is certified, committed and delivered to all
relevant replicas (events \circled{3} and \circled{4}) before
{\color{\CausalTxColor}$t_1$} is replicated to any data center, and thus before
it is uniform. Now if %
$d_1$ fails before replicating
{\color{\CausalTxColor}$t_1$} (event \circled{5}), no remote data center will be
able to expose {\color{\StrongTxColor}$t_2$}, because its dependency
{\color{\CausalTxColor}$t_1$} is missing. This violates the \liveness{} property,
and even worse, %
no strong transaction conflicting with {\color{\StrongTxColor}$t_2$} can 
commit from now on. For instance, let {\color{\StrongTxColor}$t_3$} be such a
transaction, submitted at $d_3$ (event \circled{6}). Because $d_3$ cannot expose
{\color{\StrongTxColor}$t_2$}, transaction {\color{\StrongTxColor}$t_3$}
executes on a snapshot excluding {\color{\StrongTxColor}$t_2$}. Hence,
{\color{\StrongTxColor}$t_3$} will abort during certification (events
\circled{7} and \circled{8}): committing it would violate the Conflict Ordering
property, since transactions {\color{\StrongTxColor}$t_2$} and
{\color{\StrongTxColor}$t_3$} conflict, but neither of them is visible to the
other. Ensuring that {\color{\CausalTxColor}$t_1$} is uniform before committing
{\color{\StrongTxColor}$t_2$} prevents this problem, because it guarantees that
{\color{\CausalTxColor}$t_1$} will eventually be replicated at $d_3$. After this
{\color{\StrongTxColor}$t_2$} will be exposed to conflicting transactions at
this data center, which will allow them to commit.

\paragraph{Minimizing the latency of strong transactions.} Ensuring that
all the causal dependencies of a strong transaction are uniform before
committing it may significantly increase its latency, since this requires
additional communication between data centers. \System mitigates this problem by
executing causal transactions on a snapshot that is slightly in the past, which
is allowed by causal consistency. Namely, \System makes a remote causal
  transaction visible to the clients only after it is uniform.
This minimizes the latency of a strong transaction, since to commit it
only needs to wait for causal transactions originating at the local data center
to become uniform. We cannot delay the visibility of the latter transactions due
to the need to guarantee {\em read your writes} to local clients.

\paragraph{On-demand durability of causal transactions.} Client applications
interacting with the external world require hard durability guarantees:
e.g., a banking application has to ensure that a withdrawal is durably recorded
before authorizing the operation. \System guarantees that, once a strong
transaction commits, the transaction and its dependencies are durable. However,
\System returns from a causal transaction before it is replicated, and thus the
transaction may be lost if its origin data center fails. Ensuring the durability
of every single causal transaction would require synchronization between data
centers on its critical path, defeating the benefits of causal consistency.
Instead, \System reuses the mechanism for tracking uniformity to let the clients
pay the cost of durability only when necessary. Even though \System replicates
causal transactions asynchronously, it allows clients to execute a \emph{uniform
  barrier}, which ensures that the transactions they have observed so far are
uniform, and thus durable.

\paragraph{Client migration.} Clients may need to migrate between
data centers, e.g., because of roaming or for load balancing. \System also uses
the uniformity mechanism to preserve session guarantees during migration.
A client wishing to migrate from its local data center $d$ to another data
center $i$ first invokes a uniform barrier at $d$. This guarantees that the
transactions the client has observed or issued at $d$ are durable and will
eventually become visible at $i$, even if $d$ fails. The client then makes an
{\em attach} call at the destination data center $i$ that waits until $i$ stores
all the above transactions. After this, the client can operate at $i$ knowing
that the state of the data center is consistent with the client's previous
actions. 

Currently \System does not support consistent client migration in response
to a data center failure: if the data center a client is connected to fails, the
client will have to restart its session when connecting to a different data
center. 
As shown in~\cite{swiftcloud}, this limitation can be lifted without defeating
the benefits of causal consistency. We leave integrating the corresponding
mechanisms into \System for future work.

\section{Fault-Tolerant Causal Consistency Protocol}

We first describe the \System protocol for the case when all transactions are
causal.  We give its pseudocode in Algorithms~\ref{alg:txncoord1} and
\ref{alg:clock}; for now the reader should ignore highlighted lines, which are
needed for strong transactions. For simplicity, we assume that each handler in
the algorithms executes atomically (although our implementation is
parallelized). We reference pseudocode lines using the format
algorithm\#:line\#.

\subsection{Metadata}
\label{sec:metadata}

Most metadata in our protocol are represented by vectors with an entry per each
data center, where each entry stores a scalar timestamp. However, different
pieces of metadata use the vectors in different ways, which we now describe.

\paragraph{Tracking causality.}
The first use of the vectors is as vector
clocks~\cite{vectorclocks1,vectorclocks2}, to track causality between
transactions. Given vectors $V_1$ and $V_2$, we write $V_1 < V_2$ if each entry
of $V_1$ is no greater than the corresponding entry of $V_2$, and at least one
is strictly smaller. Each update transaction is tagged with a \emph{commit
  vector} $\commitvector$. The order on these vectors is consistent with the
causal order $\prec$ from \S\ref{sec:consistency}: if $\commitvector_1$ and
$\commitvector_2$ are the commit vectors of two update transactions $t_1$ and
$t_2$ such that $t_1 \prec t_2$, then $\commitvector_1 < \commitvector_2$. For a
transaction originating at a data center $d$ with a commit vector $\commitvector$,
we call $\commitvector[d]$ its {\em local timestamp}.

Each replica $\partition^m_d$ maintains a
log $\store[k]$ of update operations performed on each data item $k$ stored at
the replica. Each log entry stores, together with the operation, the commit
vector of the transaction that performed it. This allows reconstructing
different versions of a data item from its log.

\paragraph{Representing causally consistent snapshots.}
The protocol also uses a vector to represent a snapshot of the data store on
which a transaction operates: a snapshot vector $V$ represents all transactions
with a commit vector $\le V$. This snapshot is causally consistent. Indeed,
consider a transaction $t_1$ included into it, i.e., $\commitvector_1 \leq
V$. Since any causal dependency $t_0$ of $t_1$ is such that
$\commitvector_0 < \commitvector_1$, we have $\commitvector_0 < V$, so that
$t_0$ is also included into the snapshot. A client also maintains a vector
$\past$
that represents its \emph{causal past}: a causally consistent snapshot including
the update transactions the client has previously observed.

\paragraph{Tracking what is replicated where.}  
Each replica $\partition^m_d$ maintains three vectors that are used to compute
which transactions are uniform. These respectively track the sets of
transactions replicated at $\partition^m_d$, the local data center $d$, and
$f+1$ data centers. Each of these vectors $V$ represents the set of update
transactions originating at a data center $i$ with a local timestamp $\le
V[i]$. Note that this set may not form a causally consistent snapshot. The first
vector maintained by $\partition^m_d$ is $\replicavectorclock$. For each data
center $i$, it defines the prefix of update transactions originating at $i$ (in
the order of local timestamps) that $\partition^m_d$ knows about.
\begin{property}
\label{prop:knownvc}
For each data center $i$, the replica $\partition^m_d$ stores the updates to
partition $m$ by all transactions originating at $i$ with local timestamps
$\leq \replicavectorclock[i]$.
\end{property}
Our protocol ensures that $\replicavectorclock[d]\leq \clock$ at any replica in
data center $d$. The vector $\replicavectorclock$ at $\partition^m_d$ records
whether the updates to partition $m$ by a given transaction are stored at this
replica. In contrast, the next vector $\stablesnapshot$ records whether the
updates to {\em all} partitions by a transaction are stored at the local data
center $d$.
\begin{property}
\label{prop:stablevc}
For each data center $i$, the data center $d$ stores the updates by all
transactions originating at $i$ with local timestamps $\leq \stablesnapshot[i]$.
More precisely, we are guaranteed that $\replicavectorclock[i]$ at any replica of $d$ ${}\ge{}$
$\stablesnapshot[i]$ at any $\partition^m_d$.
\end{property}
Finally, the last vector $\uniformsnapshot$ defines the set of update
transactions that $\partition^m_d$ knows to have been replicated at
$f+1$ data centers, including $d$.
\begin{property}
\label{prop:uniformity}
Consider $\uniformsnapshot[i]$ at $\partition^m_d$. All update
transactions originating at $i$ with local timestamps $\leq$
$\uniformsnapshot[i]$ are replicated at $f+1$ datacenters including
$d$. More precisely: $\replicavectorclock[i]$
at any replica of these data centers $\geq \uniformsnapshot[i]$ at $\partition^m_d$.
\end{property}

When $\uniformsnapshot$ is reinterpreted as a causally consistent snapshot, it
defines transactions that $\partition^m_d$ knows to be uniform according to
Definition~\ref{def:uniform}:
\begin{property}
\label{prop:uniformsnapshot}
Consider $\uniformsnapshot$ at $\partition^m_d$. All update transactions with
commit vectors $\leq \uniformsnapshot$ are uniform.
\end{property}
\noindent {\em Proof sketch.}
Consider a transaction $t_1$ that originates at a data center $i$ with a commit
vector $\commitvector_1\leq \uniformsnapshot$ at $\partition^m_d$. In
particular, $\commitvector_1[i]\leq\uniformsnapshot[i]$, and by
Property~\ref{prop:uniformity}, $t_1$ is replicated at $f+1$ data centers. We
assume at most $f$ failures. Then the transaction forwarding mechanism
of our protocol 
(\S\ref{sec:overview}) guarantees that $t_1$
will eventually be replicated at all correct data centers.
Consider now any
causal dependency $t_2$ of $t_1$ with a commit vector $\commitvector_2$. Since
commit vectors are consistent with causality,
$\commitvector_2 <\commitvector_1\leq\uniformsnapshot$. Then as above, we can
again establish that $t_2$ will be replicated at all correct data centers, as
required by Definition~\ref{def:uniform}.\qed

\subsection{Causal Transaction Execution}
\label{sec:bluetxs}

\paragraph{Starting a transaction.}
A client can submit a transaction to any replica in its local data center by
calling $\STARTTX(\avc)$, where $\avc$ is the client's causal past $\past$
(\algline{alg:txncoord1}{line:starttx:call}, for brevity, we omit the pseudocode
of the client). A replica $\partition^m_d$ receiving such a request acts as the
transaction {\em coordinator}. It generates a unique transaction identifier
$\tx$, computes a snapshot $\vecsnapshottime[\tx]$ on which the transaction will
execute, and returns $\tx$ to the client (we explain
\alglines{alg:txncoord1}{line:starttx:start}{line:starttx:end} and similar ones
later). The snapshot is computed by combining uniform transactions from
$\uniformsnapshot$ (\algline{alg:txncoord1}{alg:coord:starttx:init}) with the
transactions from the client's causal past originating at $d$
(\algline{alg:txncoord1}{alg:coord:starttx:end}). The former is crucial to
minimize the latency of strong transactions (\S\ref{sec:overview}), while the
latter ensures \emph{read your writes}.

\paragraph{Transaction execution.}  The client proceeds to execute the
transaction $\tx$ by issuing a sequence of operations at its coordinator via
\UPDATETX{} (\algline{alg:txncoord1}{alg:coord:execop}). When the coordinator
receives an operation $\op$ on a data item $k$, it sends a $\EXECOP$ message
with the transaction's snapshot $\vecsnapshottime[\tx]$ to the local replica
responsible for $k$
(\algline{alg:txncoord1}{alg:coord:sendop}). Upon receiving the message
(\algline{alg:txnpartition}{line:updateuniformop:receive}), the replica first
ensures that it is as up-to-date as required by the snapshot
(\algline{alg:txnpartition}{line:waitexecute}). It then computes the latest
version of $k$ within the snapshot by applying the operations from $\store[k]$
by all transactions
with commit vectors $\leq \vecsnapshottime[\tx]$. The result is sent to the
coordinator in a $\OPRET$ message. After receiving it
(\algline{alg:txncoord1}{line:coord:receive-ret}), the coordinator further
applies the operations on $k$ previously executed by the transaction, which are
stored in a buffer $\writeset[\tx]$; this ensures {\em read your writes} within
the transaction. If the operation is an update, the coordinator then appends
it to $\writeset[\tx]$. Finally, the coordinator executes the desired operation
$\op$ and forwards its return value to the client.

\begin{algorithm}
  \renewcommand{\SpaceHandler}{\vspace{5pt}}
  \begin{algorithmic}[1]
    \small
    \Function{\STARTTX}{\avc}\label{line:starttx:call}
      \For{$i \in \DC \setminus \{d\}$}\label{line:starttx:start}
        \State $\uniformsnapshot[i] \gets
          \max\{\avc[i]$, $\uniformsnapshot[i]\}$\label{line:starttx:end}
      \EndFor
       \State \textbf{var} $\tx \gets \newtxid$()
      \State $\vecsnapshottime[\tx] \gets \uniformsnapshot$\label{alg:coord:starttx:init}
      \State $\vecsnapshottime[\tx][d] \gets
        \max\{\avc[d]$, $\uniformsnapshot[d]\}$ \label{alg:coord:starttx:end}
      \State \colorbox{\BackColor}{{\color{\StrongColor}
        $\vecsnapshottime[\tx][\red] \,{\gets}\,
        \max\{\avc[\red]$, $\stablesnapshot[\red]\}$}}\label{alg:coord:starttx:end-red}
      \State \Return \tx
    \EndFunction

    \SpaceHandler

    \Function{\UPDATETX}{\tx, $k$, \op}\label{alg:coord:execop}
      \State {\bf var} $\areplica \gets \partitionof$($k$)
      \State {\bf send}
        \EXECOP($\vecsnapshottime[\tx]$, $k$)
        \textbf{to} $\partition^\areplica_d$\label{alg:coord:sendop}
      \State \textbf{wait receive} \OPRET(\cstate)
      \textbf{from} $\partition^\areplica_d$\label{line:coord:receive-ret}
      \ForAll{$\langle k, \op' \rangle \in \writeset[\tx][\areplica]$}
        $\cstate \gets \apply(\op', \cstate)$\!\!
      \EndFor \label{line:formversion:end}
      \State \colorbox{\BackColor}{{\color{\StrongColor} $\readset[\tx] \gets \readset[\tx] \cup
        \{ \langle k, \op \rangle\}$}}\label{line:readset}
      \If{$\op$ is an update}
        \State $\writeset[\tx][\areplica] \gets
          \writeset[\tx][\areplica] \cdot \langle k, \op \rangle$
      \EndIf
      \State \Return $\effval{\op}{\cstate}$
    \EndFunction

   \SpaceHandler

       \WhenRcv[\EXECOP(\avecsnapshottime, $k$)]
    \textbf{from} \partition \label{line:updateuniformop:receive}
      \For{$i \in \DC \setminus \{d\}$}\label{line:updateuniformop:start}
        \State $\uniformsnapshot[i] \gets \max\{\avecsnapshottime[i]$,
          $\uniformsnapshot[i]\}$\label{line:updateuniformop:end}
      \EndFor
      \State {\bf wait until}
        $\replicavectorclock[d] \geq \avecsnapshottime[d] \wedge {}$\label{line:waitexecute}
          \Statex \hspace{1.69cm}\colorbox{\BackColor}{{\color{\StrongColor} $\replicavectorclock[\red]
        \geq \avecsnapshottime[\red]$}}
      \State \textbf{var} $\cstate \gets \bot$
      \ForAll{\label{line:formversion:start}
        $\langle \op', \commitvector\rangle {\in} \store[k].\hspace{1pt}
        \commitvector {\leq} \avecsnapshottime$}
        \State $\cstate \gets \apply(\op',\cstate)$
      \EndFor
      \State
        \textbf{send}
        \OPRET(\cstate)
        \textbf{to} \partition
    \EndWhenRcv

   \SpaceHandler
    
    \Function{\COMMITBLUE}{\tx} \label{alg:line:coordcommit}
       \State \textbf{var} $L \gets \{ \areplica \mid \writeset[\tx][\areplica] \neq
      \emptyset\}$
      \If{$L = \emptyset$}
        \Return $\vecsnapshottime[\tx]$\label{alg:line:commitro}
      \EndIf
      \State {\bf send}
      \BLUEPREPARE(\tx, $\writeset[\tx][\areplica]$, $\vecsnapshottime[\tx]$)
      \textbf{to} $\partition^\areplica_d,\, \areplica \in L$\label{alg:line:sendprepare}
      \State {\bf var} $\commitvector \gets \vecsnapshottime[\tx]$\label{alg:line:commitvectorany}
      \ForAll{$\areplica \in L$} \label{alg:line:rcvprepare}
        \State \textbf{wait receive} \BLUEPREPARED(\tx, \ts) \textbf{from} $\partition^\areplica_d$
        \State $\commitvector[d] \gets \max \{\commitvector[d]$, $\ts\}$\label{alg:line:commitvectorlocal}
      \EndFor
      \State {\bf send}
      \COMMIT(\tx, $\commitvector$)
      \textbf{to} $\partition^\areplica_d,\, \areplica \in L$\label{alg:line:sendcommit}
      \State \Return $\commitvector$\label{alg:line:returncommit}
    \EndFunction

    \SpaceHandler

        \WhenRcv[\BLUEPREPARE(\tx, \aws, $\avecsnapshottime$)] \textbf{from} \partition
    \label{line:updateuniformprepare:receive}
      \For{$i \in \DC \setminus \{d\}$}\label{line:updateuniformprepare:start}
        \State $\uniformsnapshot[i] \gets \max\{\avecsnapshottime[i]$,
          $\uniformsnapshot[i]\}$\label{line:updateuniformprepare:end}
      \EndFor
      \State \textbf{var} $\ts \gets \clock$ \label{alg:line:PrepTime}
      \State $\preparedblue \gets \preparedblue \cup
        \{ \langle \tx, \aws, \ts\rangle \}$
      \State \textbf{send} \BLUEPREPARED(\tx, \ts) \textbf{to} \partition
    \EndWhenRcv

    \SpaceHandler

    \WhenRcv[\COMMIT($\tx$, $\commitvector$)]
    \label{alg:line:commit:receive}
      \State \textbf{wait until} $\clock \geq \commitvector[d]$ \label{alg:line:commitwait}
      \State $ \langle \tx, \aws, \_\rangle \gets$
        $\mathtt{find}(\tx, \preparedblue)$
      \State $\preparedblue \gets \preparedblue \setminus \{ \langle \tx,
        \_, \_\rangle \}$
      \ForAll{$\langle k, \op \rangle \in \aws$}
        \State $\store[k] \gets \store[k] \cdot
          \langle \op, \commitvector \rangle$
      \EndFor
      \State $\committedset[d] \gets \committedset[d] \cup {}$
        \Statex\hspace{4.5cm}$\{\langle \tx, \aws,
        \commitvector\rangle\}$
    \EndWhenRcv

    \SpaceHandler

    \Function{\MAKEUNIFORM}{\avc}\label{line:barrier}
      \State \textbf{wait until} $\uniformsnapshot[d] \ge \avc[d]$ \label{line:waituniform}
    \EndFunction

    \SpaceHandler

    \Function{\ATTACH}{\avc}\label{line:attach}
    \State \textbf{wait until} $\forall i \in \DC \setminus \{d\}.\, \uniformsnapshot[i] \ge \avc[i]$\label{line:waitattach}
    \EndFunction
  \end{algorithmic}
  \caption{Transaction execution at $\partition^m_d$.}
  \label{alg:txncoord1}
  \label{alg:txnpartition}
\end{algorithm}

\paragraph{Commit.}
A client commits a causal transaction by calling $\COMMITBLUE$
(\algline{alg:txncoord1}{alg:line:coordcommit}). This returns immediately if the
transaction is read-only, since it already read a consistent snapshot
(\algline{alg:txncoord1}{alg:line:commitro}).  To commit an update transaction,
\System uses a variant of two-phase commit protocol (recall that for simplicity
we only consider whole-data center failures, not those of individual replicas,
\S\ref{sec:sysmodel}).
The coordinator first sends a $\BLUEPREPARE$ message to the replicas in the
local data center storing the data items updated by the transaction
(\algline{alg:txncoord1}{alg:line:sendprepare}). The message to each replica
contains the part of the write buffer relevant to that replica. When a replica
receives the message
(\algline{alg:txncoord1}{line:updateuniformprepare:receive}), it computes the
transaction's \emph{prepare time} $\ts$ from its local clock and adds the
transaction to $\preparedblue$, which stores the set of causal transactions that
are prepared to commit at the replica. The replica then returns $\ts$ to the
coordinator in a $\BLUEPREPARED$ message.

When the coordinator receives replies from all replicas updated by the
transaction, it computes the transaction's commit vector $\commitvector$: it
sets the local timestamp $\commitvector[d]$ to the maximum among the prepare
times proposed by the replicas
(\algline{alg:txncoord1}{alg:line:commitvectorlocal}), and it copies the other
entries of $\commitvector$ from the snapshot vector $\vecsnapshottime[\tx]$
(\algline{alg:txncoord1}{alg:line:commitvectorany}). The latter reflects the
fact that the transactions in the snapshot become causal dependencies of $\tx$.

After computing $\commitvector$, the coordinator sends it in a $\COMMIT$ message
to the relevant replicas at the local data center
(\algline{alg:txncoord1}{alg:line:sendcommit}) and returns it to the client
(\algline{alg:txncoord1}{alg:line:returncommit}). The client then sets its
causal past $\past$ to the commit vector. When a replica receives the $\COMMIT$
message (\algline{alg:txnpartition}{alg:line:commit:receive}), it removes the
transaction from $\preparedblue$, adds the transaction's updates to $\store$,
and adds the transaction to a set $\committedset[d]$, which stores transactions
waiting to be replicated to sibling replicas at other data centers.

\subsection{Transaction Replication}
\label{sec:replication}

Each replica $\partition^m_d$ periodically replicates locally committed update
transactions to sibling replicas in other data centers by executing
$\PROPAGATELOCAL$ (\algline{alg:replication}{line:replicatelocal}). Transactions
are replicated in the order of their local timestamps. The prefix of
transactions that is ready to be replicated is determined by
$\replicavectorclock[d]$: according to Property~\ref{prop:knownvc},
$\partition^m_d$ stores updates to $m$ by all transactions originating at $d$
with local timestamps $\leq \replicavectorclock[d]$. Thus, the replica first
updates $\replicavectorclock[d]$ while preserving Property~\ref{prop:knownvc}.

There are two cases of this update. If the replica does not have any prepared
transactions ($\preparedblue=\emptyset$), it sets $\replicavectorclock[d]$ to
the current value of the $\clock$
(\algline{alg:replication}{line:updateknownvc-start}). This preserves
Property~\ref{prop:knownvc} because in this case a new transaction originating
at $d$ and updating $m$ will get a prepare time at $m$ higher than the current
$\clock$ (\algline{alg:txnpartition}{alg:line:PrepTime}), and thus also a higher
local timestamp (\algline{alg:txncoord1}{alg:line:commitvectorlocal}). If the
replica has some prepared transactions, then they may end up getting local
timestamps lower than the current $\clock$. In this case, the replica sets
$\replicavectorclock[d]$ to just below the smallest prepared time
(\algline{alg:replication}{line:localknownprepared}). This preserves
Property~\ref{prop:knownvc} because: {\em (i)} currently prepared transactions
will get a local timestamp no lower than their prepare time; and {\em (ii)} as
we argued above, new transactions will get a prepare time higher than the
current $\clock$ and, hence, than the smallest prepare time.

After updating $\replicavectorclock[d]$, the replica sends a $\REPLICATE$
message to its siblings with the transactions in $\committedset[d]$ such that
$\commitvector[d]\leq\replicavectorclock[d]$, and then removes them from
$\committedset[d]$. In other words, the replica sends all transactions from the
prefix determined by $\replicavectorclock[d]$ that it has not yet replicated.

\begin{algorithm}[t]
  \renewcommand{\SpaceHandler}{\vspace{6pt}}
  \begin{algorithmic}[1]
    \small
    \Function{\PROPAGATELOCAL}{$ $} \Comment{Run
      periodically}\label{line:replicatelocal}

        \If{$\preparedblue = \emptyset$}
        $\replicavectorclock[d] \gets \clock$\label{line:updateknownvc-start}
          \Else {}
          $\replicavectorclock[d] \,{\gets}\, \min\{\ts \,{\mid}\,
          \langle \_,\_, \ts\rangle\,{\in}\,\preparedblue\}{-}1$\label{line:localknownprepared}
          \EndIf
          \State {\bf var} $\atxs \gets \{\langle \_, \_, \commitvector\rangle
          \in \committedset[d] \mid $
          \Statex \hspace{3.7cm} $\commitvector[d]\leq\replicavectorclock[d]\}$
          \If{$\atxs \neq \emptyset$}
              \State {\bf send}
                $\REPLICATE(d, \atxs)$
                \textbf{to} $\partition^m_i,\, i \in \DC \setminus \{d\}$
                \State $\committedset[d] \gets \committedset[d]\setminus \atxs$
          \Else{}
          {\bf send}
              $\BHEARTBEAT(d, \replicavectorclock[d])$
              \textbf{to} $\partition^m_i,\, i \in \DC \setminus \{d\}$\!\!\!\!\label{line:updateknownprepared}
          \EndIf
    \EndFunction

    \SpaceHandler

  \WhenRcv[\REPLICATE($i$, $\atxs$)]\label{line:receivereplicate}
      \ForAll{$\langle \tx, \aws, \commitvector\rangle \,{\in}\, \atxs$ in
        $\commitvector[i]$ order}
      \If{$\commitvector[i]>\replicavectorclock[i]$}\label{line:receivepreparedprecondition}
      \ForAll{$\langle k, \op \rangle \in \aws$}
        \State $\store[k] \gets \store[k] \cdot
          \langle \op, \commitvector \rangle$
      \EndFor
      \State $\committedset[i] \gets \committedset[i] \cup {}$
        \Statex\hspace{4.6cm}$
        \{\langle \tx, \aws, \commitvector\rangle\}$
       \State $\replicavectorclock[i] \gets \commitvector[i]$
       \EndIf
       \EndFor
    \EndWhenRcv

    \SpaceHandler

    \vspace{-1pt}

    \WhenRcv[\BHEARTBEAT($i$, $\ts$)]\label{line:heartbeatreceive}
      \State \textbf{pre:} $\ts > \replicavectorclock[i]$
      \State $\replicavectorclock[i] \gets \ts$
    \EndWhenRcv

    \SpaceHandler

    \vspace{-1pt}

    \Function{\PROPAGATEREMOTE}{$i, j$}\label{line:forward}
          \State {\bf var} $\atxs \gets \{\langle \_, \_, \commitvector\rangle
            \in \committedset[j] \mid $\label{line:txstoforward}
          \Statex \hspace{3cm} $\commitvector[j]>\rvc[i][j]\}$
          \If{$\atxs \neq \emptyset$}
                  {\bf send}
                $\REPLICATE(j, \atxs)$
                \textbf{to} $\partition^m_i$
          \Else {}
            {\bf send}
              $\BHEARTBEAT(j, \replicavectorclock[j])$
              \textbf{to} $\partition^m_i$\label{line:heartbeatforward}
          \EndIf
    \EndFunction

    \SpaceHandler
    \vspace{-1pt}

    \Function{\CASTVC}{$ $} \Comment{Run periodically} \label{alg:line:bcast}
    \State {\bf send} $\UPDCSS(m, \replicavectorclock)$ 
    \textbf{to} $\partition^\areplica_d,\, \areplica \in \Partitions$ \label{alg:line:send-knownvc}
    \State {\bf send}
    $\UPDUSS(d, \stablesnapshot)$ \textbf{to} $\partition^m_i,\, i \in \DC$\label{alg:line:updategsssend}
    \State {\bf send}
    $\UPDRCV(d, \replicavectorclock)$ \textbf{to} $\partition^m_i,\, i \in \DC$\label{line:updateglobalmatrixsend}
    \EndFunction

    \SpaceHandler
    \vspace{-1pt}

    \WhenRcv[\UPDCSS($\areplica$, $\areplicavectorclock$)] \label{alg:line:updategss}
      \State $\pmc[\areplica] \gets \areplicavectorclock$\label{alg:line:localknownmatrix}
      \For{$i \in \DC$}
        $\stablesnapshot[i] \gets \min \{\pmc[\asecondreplica][i] \mid \asecondreplica \in \Partitions\}$\label{line:computestable}
      \EndFor
      \State \colorbox{\BackColor}{{\color{\StrongColor}
          $\stablesnapshot[\red]\gets \min\{\pmc[\asecondreplica][\red] \mid \asecondreplica \in \Partitions\}$}}\label{line:stablered}
    \EndWhenRcv

    \SpaceHandler
    \vspace{-1pt}

    \WhenRcv[\UPDUSS($i$, $\astablesnapshot$)]\label{line:updateuniformvc}
      \State $\gmc[i] \gets \astablesnapshot$\label{line:updatestablematrix}
      \State $G\gets$ all groups with $f+1$ replicas that include $\partition^m_d$\label{line:enumerate}
      \For{$j \in \DC$}
        \State \textbf{var} $\ts \gets \max \{ \min \{ \gmc[\athirddc][j] \mid \athirddc \in g \} \mid g \in G\}$
        \State $\uniformsnapshot[j] \gets \max \{ \uniformsnapshot[j]$, $\ts\}$
      \EndFor
    \EndWhenRcv

    \SpaceHandler
    \vspace{-1pt}

    \WhenRcv[\UPDRCV($\areplica$, $\areplicavectorclock$)]\label{line:updateglobalmatrix}
      \State $\rvc[\areplica] \gets \areplicavectorclock$
    \EndWhenRcv
  \end{algorithmic}
  \caption{Transaction replication at $\partition^m_d$.}
  \label{alg:replication}
  \label{alg:clock}
\end{algorithm}

When a replica $\partition^m_d$ receives a $\REPLICATE$ message with a set of
transactions $\atxs$ originating at a sibling replica $\partition^m_i$
(\algline{alg:replication}{line:receivereplicate}), it iterates over $\atxs$ in
$\commitvector[i]$ order. For each new transaction in $\atxs$ with commit vector
$\commitvector$, the replica adds the transaction's operations to its log and
sets $\replicavectorclock[i]=\commitvector[i]$. Since communication channels are
FIFO, $\partition^m_d$ processes all transactions from $\partition^m_i$ in their
local timestamp order. Hence, the above update to $\replicavectorclock[i]$
preserves Property~\ref{prop:knownvc}: $\partition^m_d$ stores updates
originating at $\partition^m_i$ by all transactions with
$\commitvector[i]\leq \replicavectorclock[i]$.  Finally, the replica adds the
transactions to $\committedset[i]$, which is used to implement transaction
forwarding (\S\ref{sec:overview}).  Due to the forwarding, $\partition^m_d$ may
receive the same transaction from different data centers. Thus, when processing
transactions in the $\REPLICATE$ message, it checks for duplicates
(\algline{alg:replication}{line:receivepreparedprecondition}).

\subsection{Advancing the Uniform Snapshot}
\label{sec:clockcomputation}

Replicas run a background protocol that refreshes the information about uniform
transactions.
This proceeds in two stages. First, a replica keeps track of which transactions
have been replicated at the replicas of other partitions in the same data
center. To this end, replicas in the same data center periodically exchange
$\UPDCSS$ messages with their $\replicavectorclock$ vectors, which they store in
a matrix $\pmc$
(\algtwolines{alg:clock}{alg:line:send-knownvc}{alg:line:updategss});
in our implementation this is done via a dissemination tree. This
matrix is then used to compute the vector $\stablesnapshot$, which represents
the set of transactions that have been fully replicated at the local data center
as per Property~\ref{prop:stablevc}. To ensure this, a replica computes an entry
$\stablesnapshot[i]$ as the minimum of $\replicavectorclock[i]$ it received from
the replicas of other partitions in the same data center
(\algline{alg:clock}{line:computestable}).

In the second stage of the background protocol, sibling replicas periodically
exchange $\UPDUSS$ messages with their $\stablesnapshot$ vectors, which they
store in a matrix $\gmc$
(\algtwolines{alg:clock}{alg:line:updategsssend}{line:updateuniformvc}). This
matrix is then used by a replica
to compute $\uniformsnapshot$, which characterizes the update transactions that
are replicated at $f+1$ data centers as per Property~\ref{prop:uniformity}. To
this end, a replica first enumerates all groups $G$ of $f+1$ data centers that
include its local data center (\algline{alg:clock}{line:enumerate}). For each
data center $j$ the replica performs the following computation. First, for each
group $g \in G$, it computes the minimum $j$-th entry in the stable vectors of
all data centers $\athirddc \in g$:
$\min\{\gmc[\athirddc][j]\mid \athirddc\in g\}$. By Property~\ref{prop:stablevc}
all update transactions originating at $j$ with local timestamp $\le$ the
minimum have been replicated at all data centers in $g$. The replica then sets
$\uniformsnapshot[j]$ to the maximum of the resulting values computed for all
groups $g \in G$, to cover transactions that are replicated at any such group.
According to Property~\ref{prop:uniformsnapshot}, the transactions with commit
vectors $\leq \uniformsnapshot$ are uniform, and now become visible to
transactions coordinated by $\partition^m_d$ (\S\ref{sec:bluetxs}).

Replicas also update $\uniformsnapshot$ in lines
\ref{alg:txncoord1}:\ref{line:starttx:start}-\ref{line:starttx:end},
\ref{alg:txnpartition}:\ref{line:updateuniformop:start}-\ref{line:updateuniformop:end}
and
\ref{alg:txnpartition}:\ref{line:updateuniformprepare:start}-\ref{line:updateuniformprepare:end}
by incorporating $\vecsnapshottime[i]$ for remote data centers $i$. This is
safe because a transaction executes on a snapshot that only includes uniform
remote transactions.

Finally, if a replica does not receive new transactions for a long time, it
sends the value of its $\replicavectorclock[d]$ as a heartbeat
(\algtwolines{alg:replication}{line:updateknownprepared}{line:heartbeatreceive}).
This allows advancing $\stablesnapshot$ and $\uniformsnapshot$ even under skewed
load distributions.

\subsection{Transaction Forwarding}
\label{sec:forward}

As we explained in \S\ref{sec:overview}, to guarantee that a transaction
originating at a correct data center eventually becomes exposed at all correct
data centers despite failures (\liveness), replicas may have to forward remote
update transactions. To determine which transactions to forward, each replica
keeps track of the update transactions that have been replicated at sibling
replicas in other data centers. To this end, sibling replicas periodically
exchange $\UPDRCV$ messages with their $\replicavectorclock$ vectors, which they
store in a matrix $\rvc$
(\algtwolines{alg:clock}{line:updateglobalmatrixsend}{line:updateglobalmatrix}).
Thus, $\partition^m_i$ has received all update transactions from
$\partition^m_j$ with $\commitvector[j]\leq \rvc[i][j]$.

A replica $\partition^m_d$ only forwards transactions when it suspects that a
data center $j$ may have failed before replicating all the update transactions
originating at it to a data center $i$ (this information is provided by a
separate module).
In this case, $\partition^m_d$ executes $\PROPAGATEREMOTE(i, j)$
(\algline{alg:replication}{line:forward}). The function forwards the set of 
transactions $\atxs$ received from $\partition^m_j$ that have not been
replicated at $\partition^m_i$ according to
$\rvc[i][j]$. For example, in
Figure~\ref{fig:execution-causal}, \System will eventually invoke
$\PROPAGATEREMOTE(d_1, d_3)$ at replicas in $d_2$ to forward {\color{\CausalTxColor}$t_1$}. The replica
$\partition^m_d$ sends the transactions in $\atxs$ to $\partition^m_i$ in a $\REPLICATE$
message. If there are no update transactions to forward, $\partition^m_d$ sends
a heartbeat to $\partition^m_i$ with $\replicavectorclock[j]$.

\System periodically deletes from $\committedset$ transactions that have been
replicated at every data center (omitted from the pseudocode for brevity).

\subsection{On-Demand Durability and Client Migration}
\label{sec:clientmigration}

A client may wish to ensure that the transactions it has observed so far are
durable. To this end, the client can call $\MAKEUNIFORM(\avc)$ at any replica in
its local data center $d$, where $\avc$ is the client's causal past $\past$
(\algline{alg:txncoord1}{line:barrier}). The replica returns to the client only
when all the transactions from $\past$ that originate at $d$ are uniform, and
thus durable. Then the same holds for all transactions from $\past$, because the
protocol only exposes remote transactions to clients when they are already
uniform (\S\ref{sec:bluetxs}).

A client wishing to migrate from its local data center $d$ to another data
center $i$ first calls $\MAKEUNIFORM(\avc)$ at any replica in $d$ with
$\avc = \past$, to ensure that the transactions the client has observed or
issued at $d$ will eventually become visible at $i$. The client then calls
$\ATTACH(\avc)$ at any replica in $i$ %
(\algline{alg:txncoord1}{line:attach}). The replica returns when its
$\uniformsnapshot$ includes all remote transactions from $\avc$
(\algline{alg:txncoord1}{line:waitattach}). The client can then be sure that its
transactions at $i$ will operate on snapshots including all the transactions it
has observed before.

\section{Adding Strong Transactions}
\label{sec:redtransactions}

We now describe the full \System protocol with both causal and strong
transactions. It is obtained by adding the highlighted lines to
Algorithms~\ref{alg:txncoord1}-\ref{alg:clock} and a new
Algorithm~\ref{alg:txncoord2}.

\subsection{Metadata}

The Conflict Ordering property of our consistency model requires any two
conflicting strong transactions to be related one way or another by the causal
order $\prec$ (\S\ref{sec:consistency}). To ensure this, the protocol assigns to
each strong transaction a scalar {\em strong timestamp}, analogous to those used
in optimistic concurrency control for serializability~\cite{wv}.  Several
vectors used as metadata in the causal consistency protocol
(\S\ref{sec:metadata}) are then extended with an extra $\red$ entry.

First, we extend commit vectors and those representing causally consistent
snapshots. Commit vectors are compared using the previous order $<$, but
considering all entries; as before, this order is consistent with the causal
order $\prec$. Furthermore, conflicting strong transactions are causally ordered
according to their strong timestamps.
\begin{property}
  For any conflicting strong transactions $t_1$ and $t_2$ with commit vectors
  $\commitvector_1$ and $\commitvector_2$, we have:
  $t_1 \prec t_2 \Longleftrightarrow \commitvector_1[\red] < \commitvector_2[\red]$.
\end{property}

A consistent snapshot vector $V$ now defines the set of transactions with a
commit vector $\le V$, according to the new $<$. The vectors
$\replicavectorclock$ and $\stablesnapshot$ maintained by a replica
$\partition^m_d$ are also extended with a $\red$ entry. The entries
$\replicavectorclock[\red]$ and $\stablesnapshot[\red]$ define the prefix of
strong transactions that have been replicated at $\partition^m_d$ and the local
data center $d$, respectively:
\begin{property}
\label{prop:knownvcred}
Replica $\partition^m_d$ stores the updates to $m$ by all strong
transactions with $\commitvector[\red]\leq \replicavectorclock[\red]$.
\end{property}
\begin{property}
\label{prop:stablevcred}
Data center $d$ stores the updates by all strong transactions with
$\commitvector[\red]\leq \stablesnapshot[\red]$.
\end{property}
To ensure Property~\ref{prop:stablevcred}, the $\red$ entry of $\stablesnapshot$
is updated at \algline{alg:clock}{line:stablered} similarly to its other
entries. We do not extend $\uniformsnapshot$, because our commit protocol for
strong transactions automatically guarantees their uniformity.

\subsection{Transaction Execution}
\label{sec:redexecution}

\System uses optimistic concurrency control for strong transactions, with the
same protocol for executing causal and speculatively executing strong
transactions. To this end, Algorithm~\ref{alg:txncoord1} is modified as
follows. First, the computation of the snapshot vector $\vecsnapshottime[\tx]$
is extended to compute the $\red$ entry
(\algline{alg:txncoord1}{alg:coord:starttx:end-red}), which is now taken into
account when checking that a replica state is up to date
(\algline{alg:txnpartition}{line:waitexecute}).  The $\red$ entry of the
snapshot vector is computed so as to include all strong transactions known to be
fully replicated in the local data center, as defined by
$\stablesnapshot[\red]$. To ensure {\em read your writes}, the snapshot
additionally includes strong transactions from the client's causal past, as
defined by $\avc[\red]$. Finally, the coordinator of a transaction now maintains
not only its write set, but also its read set $\readset$ that records all
operations by the transaction, including read-only ones
(\algline{alg:txncoord1}{line:readset}). The latter is used to certify strong
transactions.

After speculatively executing a strong transaction, the client tries to commit
it by calling $\COMMITRED$ at its coordinator
(\algline{alg:txncoord2}{line:commitred}). The coordinator first waits until the
snapshot on which the transaction operated becomes uniform by calling
$\MAKEUNIFORM$ (\algline{alg:txncoord2}{line:uniformred}): as we argued in
\S\ref{sec:overview}, this is crucial for liveness.
The coordinator next submits the transaction to a \emph{certification service},
which determines whether the transaction commits or aborts
(\algline{alg:txncoord2}{line:certifyred}, see \S\ref{sec:certification}).
In the former case, the service also determines its commit vector,
which the coordinator returns to the client. If the transaction commits,
the client sets its causal past $\past$ to the commit vector; otherwise, 
it re-executes the transaction. 

The certification service also notifies replicas in all data centers about
updates by strong transactions affecting them via $\DELIVERUPD$ upcalls, invoked
in an order consistent with strong timestamps of the transactions
(\algline{alg:txncoord2}{line:deliverred}). A replica receiving an upcall adds
the new operations to its log and refreshes $\replicavectorclock[\red]$ to
preserve Property~\ref{prop:knownvcred}.

Finally, a replica $\partition^m_d$ that has not seen any strong transactions
updating its partition $m$ for a long time submits a dummy strong transaction
that acts as a heartbeat (\algline{alg:txncoord2}{line:stronghb}). Similarly to
heartbeats for causal transactions, this allows coping with skewed load
distributions.

\begin{algorithm}[t]
  \begin{algorithmic}[1]
    \small
    \Function{\COMMITRED}{\tx}\label{line:commitred}
      \State $\MAKEUNIFORM(\vecsnapshottime[\tx])$\label{line:uniformred}
      \State \Return
        \CERTIFY(\tx, $\writeset[\tx]$, $\readset[\tx]$, $\vecsnapshottime[\tx]$) \label{line:certifyred}
    \EndFunction

    \SpaceHandler

    \Upon[\DELIVERUPD($W$)]\label{line:deliverred}
    \ForAll{\hspace{-1pt}$\langle \aws, \commitvector\rangle {\in} \hspace{1pt} W$\hspace{1pt}in\hspace{1pt}$\commitvector[\red]$\hspace{1pt}order}
      \ForAll{$\langle k, \op \rangle \in \aws$}
      \State $\store[k] \gets \store[k] \cdot \langle \op, \commitvector \rangle$\label{line:addred}
      \EndFor
      \State $\replicavectorclock[\red] \gets \commitvector[\red]$\label{line:setred}
      \EndFor
    \EndUpon

    \SpaceHandler

    \Function{\RHEARTBEAT}{$ $} \Comment{Run periodically}\label{line:stronghb}
      \State \Return \CERTIFY($\bot$, $\emptyset$, $\emptyset$, $\vec{0}$)
    \EndFunction

  \end{algorithmic}
  \caption{Committing strong transactions at $\partition^m_d$.}
  \label{alg:txncoord2}
\end{algorithm}

\subsection{Certification Service}
\label{sec:certification}

We implement the certification service using an existing fault-tolerant protocol
from~\cite{discpaper}, with transaction commit vectors computed using the
techniques from~\cite{multicast-dsn19}. The protocol integrates two-phase commit
across partitions accessed by the transaction and Paxos among the replicas of
each partition. It furthermore uses white-box optimizations between the two
protocols to minimize the commit latency.  The use of Paxos ensures that a
committed strong transaction is durable and its updates will eventually be
delivered at all correct data centers
(\algline{alg:txncoord2}{line:deliverred}). For each partition, a single replica
functions as the Paxos leader. The protocol is described and formally specified
elsewhere~\cite{discpaper}, and here we discuss it only briefly. Its pseudocode
and formal specification are given 
in~\tr{\ref{section:unistore-protocol}}{\nappfull}
and~\tr{\ref{section:tcs}}{\napptcs}, respectively.

The certification service accepts the read and write sets of a transaction and
its snapshot vector (\algline{alg:txncoord2}{line:certifyred}). Even though the
service is distributed, it guarantees that commit/abort decisions are computed
like in a centralized database with optimistic concurrency control -- in a total
{\em certification order}. To ensure Conflict Ordering,
the decisions are computed using a concurrency-control policy similar to that
for serializability~\cite{wv}: a transaction commits if its snapshot includes
all conflicting transactions that precede it in the certification order. The
certification service also computes a commit vector for each committed
transaction by copying its per-data center entries from the transaction's
snapshot vector and assigning a strong timestamp consistent with the
certification order.

\section{Proof of Correctness}
\label{sec:correctness}

We have rigorously proved that \System correctly implements the specification of
PoR consistency for the case when the data store manages last-writer-wins registers. The
proof uses the formal framework
from~\cite{sebastian-book,distrmm-popl,framework-concur15} and establishes
Properties~\ref{prop:knownvc}-\ref{prop:stablevcred} stated earlier. Due to
space constraints, we defer the proof
to~\tr{\ref{section:correctness-proof}}{\nappproof}.

\section{Evaluation}
\label{sec:eval}

We have implemented \System and several other protocols (listed in the
following) in the same codebase, consisting of
10.3K SLOC of Erlang. We evaluate the protocols on Amazon EC2 using
\texttt{m4.2xlarge} VMs from 5 different
regions. 
Each VM has 8 virtual cores and 32GB of RAM. 
The RTT between regions ranges from 26ms to 202ms. 
Unless otherwise stated, our experiments deploy 3 data centers, thus tolerating
a single data center failure: Virginia (US-East), California (US-West) and
Frankfurt (EU-FRA).
All Paxos leaders are located in Virginia.  By default we use 4 replica machines
per data center. Each machine stores replicas of 8
partitions, matching the number of cores. 
Clients are hosted on separate machines in each data center.
We run each experiment for at least 5 minutes, with the first and the last
minute ignored. Replicas propagate local update transactions
(\algline{alg:replication}{line:replicatelocal}) and broadcast vectors
(\algline{alg:clock}{alg:line:bcast}) every 5ms.

\subsection{Does \System combine causal and strong consistency
  effectively?}
\label{sec:performance}

We start by analyzing the performance of \System using RUBiS -- a popular
benchmark that emulates an online auction website such as eBay~\cite{por,
  red-blue}. It defines 11 read-only transactions and 5 update transactions,
e.g., selling items, bidding on items, and consulting outstanding auctions. As
in previous work~\cite{por}, to make the benchmark more challenging, we add an
extra update transaction \texttt{closeAuction} to declare the winner of an
auction. We also borrow from~\cite{por} a conflict relation between RUBiS
transactions that preserves key integrity invariants in the PoR
consistency model. This marks four transactions as strong
(\texttt{registerUser}, \texttt{storeBuyNow}, \texttt{storeBid} and
\texttt{closeAuction}) and declares three conflicts between them. For example,
\texttt{storeBid}, which places a bid on a item, conflicts with
\texttt{closeAuction} if both act on the same item: this is needed to preserve
the invariant that the winner of an auction is the highest bidder. Our RUBiS
database is configured according to the benchmark specification: it is populated
with 33,000 items for sale and 1 million users; client think times are
500ms.
We run the bidding mix workload of RUBiS with 15\% of update
transactions, which yields 10\% of strong transactions.

We compare \System with \Strong, \RedBlue and \Causal.
\Strong implements serializability~\cite{wv} as a special case of \System where
all transactions are strong and all pairs of operations conflict.
\RedBlue implements red-blue consistency~\cite{red-blue}, which like PoR,
combines causal and strong consistency. However, it declares conflicts between
all strong transactions. \RedBlue certifies strong transactions at a centralized
replicated service, with a replica at each data center. \Causal implements causal
consistency as a special case of \System where
all transactions are causal. It cannot
preserve the integrity invariants of RUBiS, but gives an upper bound on the
expected performance.

\begin{figure}[t]
\includegraphics[width=\columnwidth]{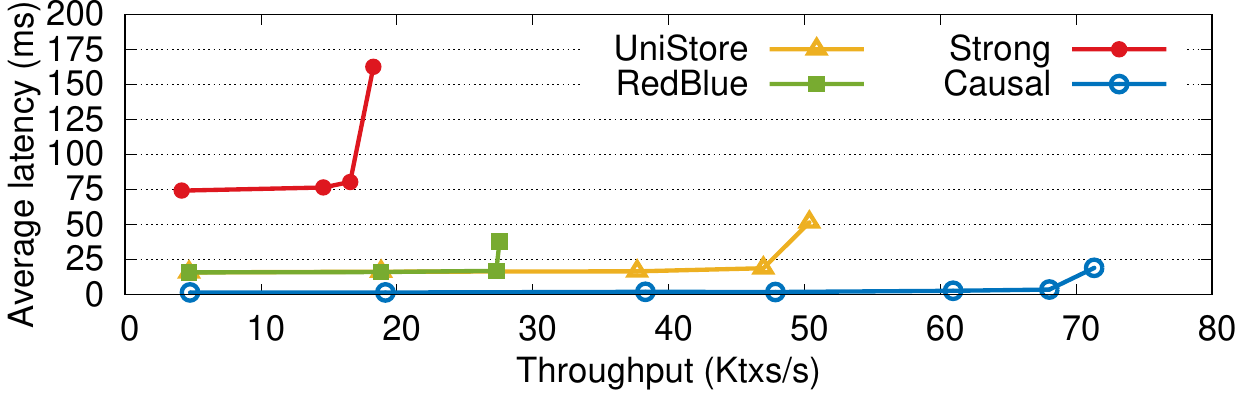}
\caption{RUBiS benchmark: throughput vs. average latency.}
\label{fig:rubis}
\end{figure}

\paragraph{Throughput and average latency.}
Figure~\ref{fig:rubis} evaluates
average transaction latency and throughput. As the figure shows, \System
exhibits a high throughput: 72\% and 183\% higher than \RedBlue and \Strong
respectively at their saturation point. This is expected, as \System implements 
the consistency model that enables the most concurrency. \Strong
classifies all transactions as strong. This impacts performance 
because executing a strong transaction is significantly more expensive
than executing a causal one. \RedBlue uses a centralized certification service
that saturates before the \System's distributed service, creating
a bottleneck. \System exhibits an average latency of 16.5ms, lower than 
80.4ms of \Strong. The latency of \RedBlue is comparable to that of
\System. This is because both systems mark the same set of
transactions as strong. Still, \RedBlue declares conflicts between all
strong transactions and thus aborts more transactions than
\System: 0.12\% vs 0.027\%. The clients whose transactions abort have to retry
them, thus increasing latency. Since the abort rate remains low in both cases, 
the difference in latency is negligible in our experiment.
We expect a more significant difference in workloads with higher contention.
Finally, in comparison
to \Causal, \System penalizes throughput by 45\%. This is the unavoidable price
to pay to preserve application-specific invariants.

\paragraph{Latency of each transaction type.}
In \System, the latency of strong transactions is dominated by
the RTT between Virginia (the leader's
region) and California (Virginia's closest data center) -- 61ms. Strong transactions
exhibit a latency of 73.9ms on average. The latency varies depending on
the client's location: from 65.4ms on average at the leader's site to 93.2ms at the
site furthest from the leader (Frankfurt). Since causal transactions do not
require coordination between data centers, they
exhibit a very low latency -- 1.2ms on average, which is comparable to
that of \Causal. This
demonstrates that \System is able to mix causal and strong
consistency effectively, as the latency of causal transactions remains
low regardless of concurrently executing strong transactions.

\subsection{How does \System scale with the number of machines?}
\label{sec:scalability}

We evaluate the
peak throughput of \System as we increase
the number of machines per each data center from 2 to
8, i.e., the number of partitions from 16 to 64. We
use a microbenchmark with 100\% of update transactions, 
where each transaction accesses three data items. We vary the ratio of
strong transactions from 0\% to 100\%
to understand their impact on scalability.

\begin{figure}[t]
\includegraphics[width=\columnwidth]{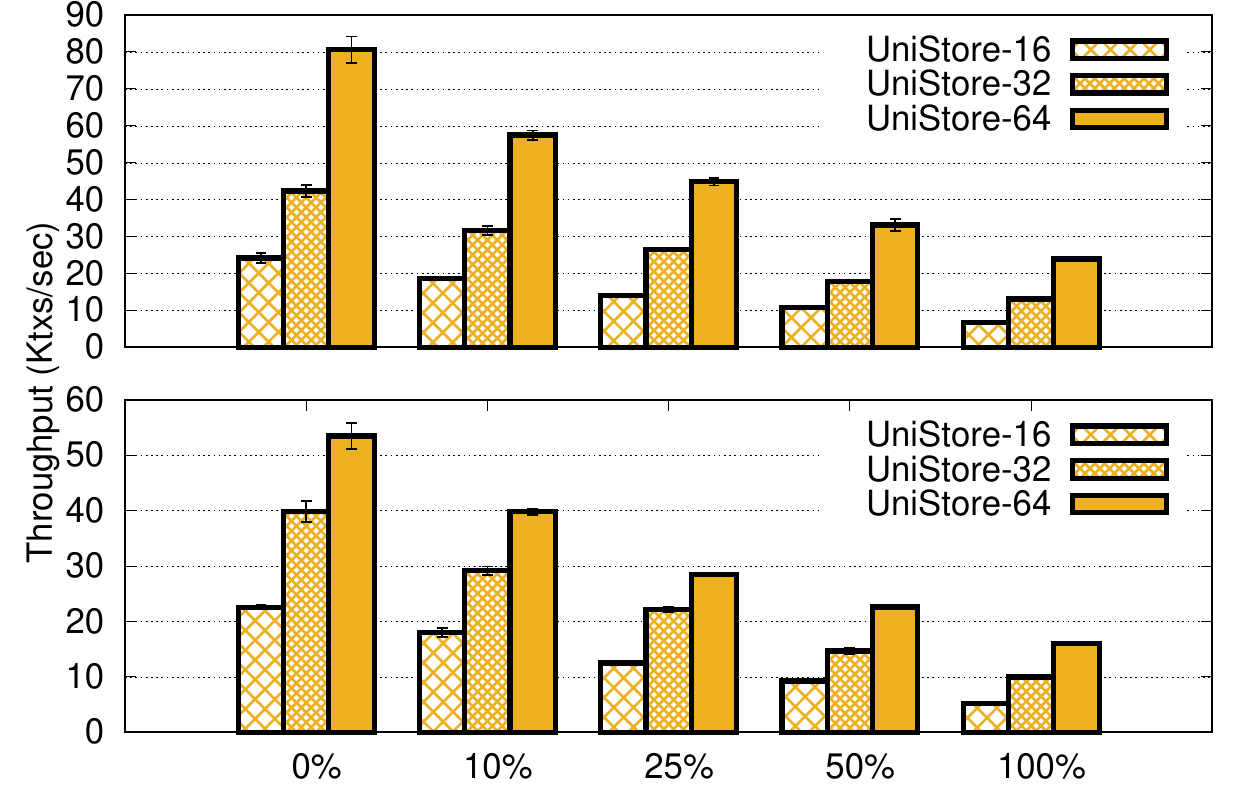}
\caption{Scalability when varying the ratio of strong transactions
  with uniform data access (top) and under contention (bottom).}
\label{fig:scalability}
\end{figure}

\paragraph{Scalability under low contention.}
For this set of experiments, the data items accessed by each
transaction are picked uniformly at random. This yields a very low contention: e.g., with 16 partitions, the
probability of two transactions accessing the same partition is
0.031. As shown by the top plot of Figure~\ref{fig:scalability}, \System is able to scale almost linearly even when the
workload includes strong transactions: a 9.76\% throughput drop compared
to the optimal scalability.  This is because, with uniform accesses, the task of committing
transactions is balanced among partitions. Thus, when the number of
partitions increases, so does the system's capacity. The
scalability is not perfect due to the cost of the background
protocol that computes $\stablesnapshot$, which grows logarithmically with the
number of partitions.
The plot also shows that strong transactions are expensive: 25.72\% of
throughput drop on average with 10\% of strong 
transactions. The performance is dominated by the
number of strong transactions that a partition can certify per
second.

\paragraph{Impact of contention.}
For this set of experiments, we set the ratio of strong transactions that access
a designated partition to 20\% to create contention. 
As shown by the bottom plot of Figure~\ref{fig:scalability}, \System is still able to
scale fairly well under contention. But, as expected,
contention has an impact on scalability: a 17.15\% throughput drop
from the optimal scalability compared to the 9.76\% throughput drop in
the experiments without contention.

\subsection{What is the cost of uniformity?}
\label{sec:cost-uniformity}
We compare \CureFT to \Uniform. 
\CureFT implements Cure~\cite{cure}, a causally consistent data store,
and makes it fault tolerant by adding transaction forwarding (\S\ref{sec:overview}).
\Uniform is a simplified version
of \System that removes
all the mechanisms related to strong transactions. \Uniform tracks uniformity and makes remote transactions visible only
when these are uniform; \CureFT does not. We use a microbenchmark with only causal
transactions and 15\% of update transactions. Each transaction
accesses three data items. 

\paragraph{Throughput penalty.}

\begin{figure}[t]
\includegraphics[width=\columnwidth]{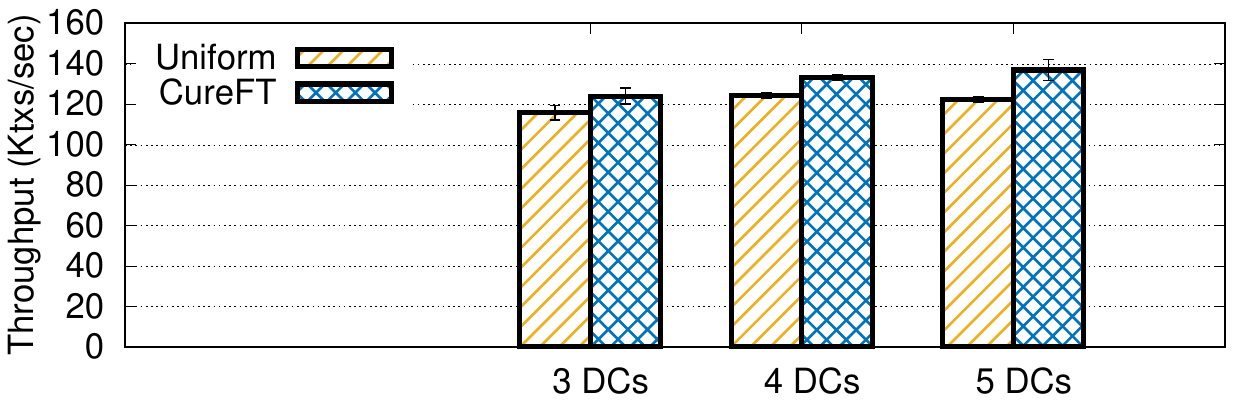}
\caption{Throughput penalty of tracking uniformity.}
\label{fig:thpenalty}
\end{figure}

Figure~\ref{fig:thpenalty} evaluates the cost of tracking
uniformity. It shows the peak throughput when the number of data centers increases from 3 to
5. We first add Ireland and then Brazil. As we do this, the throughput
remains almost constant. This is because each
data center stores replicas of all partitions and the
computational power gained when adding a data center is 
offset by the cost of replicating update transactions. As the
figure shows, the cost of tracking 
uniformity is small: a 7.97\% drop on average. The gap grows as we
increase the number of data centers: a 10.61\% drop on average with 5
data centers.
This is because, to track uniform transactions, sibling replicas exchange
messages: the more data centers, the more messages exchanged.
The penalty can be 
reduced by decreasing the frequency at which sibling replicas
exchange their $\stablesnapshot$
(\algline{alg:clock}{alg:line:updategsssend}), at the expense of an extra delay
in the visibility of remote transactions.

\begin{figure}[t]
\includegraphics[width=\columnwidth]{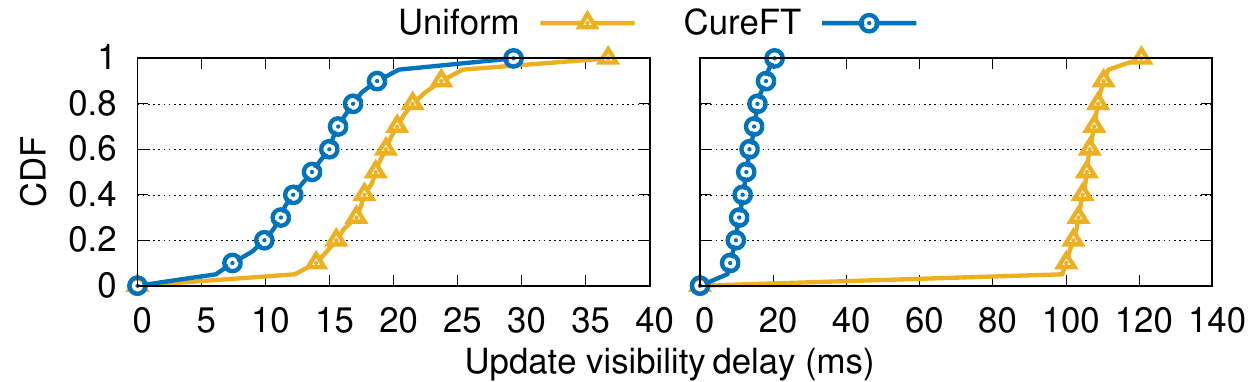}
\caption{Left: California to Brazil (best case). Right: California to
  Virginia (worst case).}
\label{fig:visibility}
\end{figure}

\paragraph{Reading from a uniform snapshot.}

Figure~\ref{fig:visibility} evaluates the delay on the visibility
of remote transactions when reading from a uniform
snapshot. We deploy four data centers: Virginia,
California, Frankfurt and Brazil. We set $f=2$ to tolerate 2 data
center failures (when $f=1$, \Uniform shows no delay). Under such a configuration, a data
center makes a transaction visible when it knows that
3 data centers store the transaction and its dependencies
(\S\ref{sec:clockcomputation}). The figure shows
the cumulative distribution of the delay before updates
from California are visible in Brazil
and Virginia. 

The extra delay at Brazil is only of 5ms at the 90$^{\rm th}$
percentile. This is the best case scenario for
\Uniform because Brazil learns that Virginia
stores a transaction originating at California only 2ms after 
receiving it. The worst case scenario for \Uniform is when the origin and the
destination data center are the closest ones. This is why the extra
delay at Virginia is 92ms at the 90$^{\rm th}$
percentile: Brazil learns that Frankfurt stores a transaction
originating at California 88ms after receiving it. Note that when clients 
communicate only via the data store, the delay is
unnoticeable. Even if clients communicate out of band, as the maximum extra
delay is less than 100ms, it is unlikely that a client will
miss an update.

\section{Related Work}
\label{sec:related}

\paragraph{Systems with multiple consistency levels.}
A number of data stores have combined weak and strong consistency,
including several commercial and academic systems that combine eventual and strong
consistency~\cite{cosmosdb,documentdb,pileus,google,dynamodb,cassandra,
tapir}.
Several academic data stores combined causal and strong consistency~\cite{lazy, red-blue, valter, por,
  walter, pileus}. Pileus~\cite{pileus} funnels all updates through a single
data center. In the fault-tolerant version of lazy
replication~\cite{lazy}, causal operations require synchronization
between replicas on its critical path. In both cases, causal
operations are not highly available, defeating the benefits of
causal consistency. Walter~\cite{walter} restricts causal operations to
a specific type and lacks fault tolerance due to the use
of two-phase commit across data centers.
The remaining works~\cite{red-blue, valter, por} support highly available causal
operations, but are not fault tolerant. First, they do not make causal
operations uniform on demand to guarantee the liveness of strong
operations. Thus, they suffer from the liveness issue we explained
in \S\ref{sec:overview} (Figure~\ref{fig:execution-strong}). In addition, these
systems do not use fault-tolerant mechanisms even for strong transactions.
They guard the use of strong transactions using mechanisms similar to locks;
if the lock holder fails before releasing it, no other data center can execute a
strong transaction requiring the same lock. This occurs even
when the service handing locks is fault-tolerant, as in~\cite{por}.
Finally, the above systems either do not include mechanisms for partitioning the
key space among different machines in a data center or include
per-data center centralized services, which
limits their scalability (\S\ref{sec:scalability}).

Some group communication systems mix causal and atomic
broadcast~\cite{isis,horus}. However, these systems do not provide mechanisms
for maintaining transactional data consistency.

Several papers have proposed tools that use formal verification technology
to ensure that consistency choices do not violate application
invariants~\cite{valter,cheng-tool,cise-tool,cise-popl16,suresh,sreeja}. Such
tools can make it easier for programmers to use our system.

\paragraph{Causal consistency implementations.}
Our subprotocol for causal consistency belongs to a family of highly scalable
protocols that avoid using any centralized components or dependency check
messages~\cite{gentlerain, cure, wren, paris, pocc}; other alternatives are
less scalable~\cite{cops, eiger, orbe, bolton, chainreaction, swiftcloud,
  saturn, eunomia, occult}. While we base our causal consistency subprotocol on
an existing one, Cure~\cite{cure}, we have extended it in nontrivial ways, by
integrating mechanisms for tracking uniformity (\S\ref{sec:clockcomputation})
and for transaction forwarding (\S\ref{sec:forward}). Some of the above
protocols~\cite{eunomia,paris} use hybrid clocks instead of real time~\cite{hlc}
to improve performance with large clock skews; this technique can also be
integrated into \System.

SwiftCloud~\cite{swiftcloud} implements
\emph{k-stability}~\cite{k-stability}, a notion similar to uniformity,
to enable client migration. SwiftCloud relies on a single per-data
center sequencer, which makes tracking k-stability easy, but the data store
less scalable. Our protocol is more sophisticated, since we distribute the
responsibility of tracking uniformity among the replicas in a data
center.

\paragraph{Paxos variants.}
Several Paxos variants~\cite{generalized,epaxos,atlas,generic} lower the latency
by allowing commutative operations to execute at replicas in arbitrary
orders. In contrast to them, \System implements PoR consistency, which allows
causal transactions to execute without any synchronization at all.

\section{Conclusion}

This paper presented \System, the first fault-tolerant and
scalable data store that combines causal and strong
consistency. \System carefully integrates state-of-the-art scalable protocols and extends
them in nontrivial ways. To maintain liveness despite data
center failures, unlike previous work, \System commits a strong transaction only when all its
causal dependencies are uniform. Our results show that
\System combines causal and strong
consistency effectively: 3.7$\times$
lower latency on average than a strongly consistent system with 1.2ms
latency on average for causal transactions.
We expect that the key ideas in \System will pave
the way for practical systems that combine causal and strong
consistency.

\paragraph{Acknowledgements.}
We thank our shepherd, Heming Cui, as well as Gregory Chockler, Vitor Enes, Luís
Rodrigues and Marc Shapiro for comments and suggestions. This work was partially
supported by an ERC Starting Grant RACCOON, the Juan de la Cierva Formaci\'on
funding scheme (FJC2018-036528-I), the CCF-Tencent Open Fund (CCF-Tencent
RAGR20200124) and the AWS Cloud Credit for Research program.

\bibliographystyle{abbrv}
\bibliography{biblio}

\begin{thebibliography}{10}

\bibitem{pacelc}
D.~Abadi.
\newblock Consistency tradeoffs in modern distributed database system design:
  {CAP} is only part of the story.
\newblock {\em {IEEE} Computer}, 45(2), 2012.

\bibitem{causal-memory}
M.~Ahamad, G.~Neiger, J.~E. Burns, P.~Kohli, and P.~W. Hutto.
\newblock Causal memory: Definitions, implementation, and programming.
\newblock {\em Distributed Comput.}, 9(1), 1995.

\bibitem{cure}
D.~D. Akkoorath, A.~Z. Tomsic, M.~Bravo, Z.~Li, T.~Crain, A.~Bieniusa,
  N.~Pregui\c{c}a, and M.~Shapiro.
\newblock Cure: Strong semantics meets high availability and low latency.
\newblock In {\em International Conference on Distributed Computing Systems
  (ICDCS)}, 2016.

\bibitem{chainreaction}
S.~Almeida, J.~Leit{\~{a}}o, and L.~Rodrigues.
\newblock {ChainReaction}: A causal+ consistent datastore based on chain
  replication.
\newblock In {\em European Conference on Computer Systems (EuroSys)}, 2013.

\bibitem{dynamodb}
Amazon.
\newblock Read consistency.
\newblock
  \\https://docs.aws.amazon.com/amazondynamodb/latest/\\developerguide/HowItWorks.ReadConsistency.html,
  2020.

\bibitem{cassandra}
{Apache Cassandra}.
\newblock Read repair.
\newblock
  \\https://cassandra.apache.org/doc/latest/operating/\\read\_repair.html,
  2020.

\bibitem{hagit-cc}
H.~Attiya, F.~Ellen, and A.~Morrison.
\newblock Limitations of highly-available eventually-consistent data stores.
\newblock {\em {IEEE} Trans. Parallel Distributed Syst.}, 28(1), 2017.

\bibitem{bolton}
P.~Bailis, A.~Ghodsi, J.~M. Hellerstein, and I.~Stoica.
\newblock Bolt-on causal consistency.
\newblock In {\em International Conference on Management of Data (SIGMOD)},
  2013.

\bibitem{valter}
V.~Balegas, N.~Pregui\c{c}a, R.~Rodrigues, S.~Duarte, C.~Ferreira,
  M.~Najafzadeh, and M.~Shapiro.
\newblock Putting the consistency back into eventual consistency.
\newblock In {\em European Conference on Computer Systems (EuroSys)}, 2015.

\bibitem{isis}
K.~Birman, A.~Schiper, and P.~Stephenson.
\newblock Lightweight causal and atomic group multicast.
\newblock {\em ACM Trans. Comput. Syst.}, 9(3), 1991.

\bibitem{isis-reliable}
K.~P. Birman and T.~A. Joseph.
\newblock Reliable communication in the presence of failures.
\newblock {\em {ACM} Trans. Comput. Syst.}, 5(1), 1987.

\bibitem{saturn}
M.~Bravo, L.~Rodrigues, and P.~Van~Roy.
\newblock Saturn: A distributed metadata service for causal consistency.
\newblock In {\em European Conference on Computer Systems (EuroSys)}, 2017.

\bibitem{sebastian-book}
S.~Burckhardt.
\newblock {\em Principles of Eventual Consistency}.
\newblock Now Publishers, 2014.

\bibitem{distrmm-popl}
S.~Burckhardt, A.~Gotsman, H.~Yang, and M.~Zawirski.
\newblock Replicated data types: specification, verification, optimality.
\newblock In {\em Symposium on Principles of Programming Languages (POPL)},
  2014.

\bibitem{cachin-book}
C.~Cachin, R.~Guerraoui, and L.~E.~T. Rodrigues.
\newblock {\em Introduction to Reliable and Secure Distributed Programming (2nd
  ed.)}.
\newblock Springer, 2011.

\bibitem{framework-concur15}
A.~Cerone, G.~Bernardi, and A.~Gotsman.
\newblock A framework for transactional consistency models with atomic
  visibility.
\newblock In {\em International Conference on Concurrency Theory (CONCUR)},
  2015.

\bibitem{giza}
Y.~L. Chen, S.~Mu, J.~Li, C.~Huang, J.~Li, A.~Ogus, and D.~Phillips.
\newblock Giza: Erasure coding objects across global data centers.
\newblock In {\em {USENIX} Annual Technical Conference (USENIX ATC)}, 2017.

\bibitem{discpaper}
G.~Chockler and A.~Gotsman.
\newblock Multi-shot distributed transaction commit.
\newblock In {\em Symposium on Distributed Computing (DISC)}, 2018.

\bibitem{spanner}
J.~C. Corbett, J.~Dean, M.~Epstein, A.~Fikes, C.~Frost, J.~J. Furman,
  S.~Ghemawat, A.~Gubarev, C.~Heiser, P.~Hochschild, W.~C. Hsieh, S.~Kanthak,
  E.~Kogan, H.~Li, A.~Lloyd, S.~Melnik, D.~Mwaura, D.~Nagle, S.~Quinlan,
  R.~Rao, L.~Rolig, Y.~Saito, M.~Szymaniak, C.~Taylor, R.~Wang, and
  D.~Woodford.
\newblock {Spanner: Google's Globally-Distributed Database}.
\newblock In {\em Symposium on Operating Systems Design and Implementation
  ({OSDI})}, 2012.

\bibitem{dynamo}
G.~DeCandia, D.~Hastorun, M.~Jampani, G.~Kakulapati, A.~Lakshman, A.~Pilchin,
  S.~Sivasubramanian, P.~Vosshall, and W.~Vogels.
\newblock Dynamo: {A}mazon's highly available key-value store.
\newblock In {\em Symposium on Operating Systems Principles (SOSP)}, 2007.

\bibitem{orbe}
J.~Du, S.~Elnikety, A.~Roy, and W.~Zwaenepoel.
\newblock Orbe: Scalable causal consistency using dependency matrices and
  physical clocks.
\newblock In {\em Symposium on Cloud Computing (SoCC)}, 2013.

\bibitem{gentlerain}
J.~Du, C.~Iorgulescu, A.~Roy, and W.~Zwaenepoel.
\newblock Gentlerain: Cheap and scalable causal consistency with physical
  clocks.
\newblock In {\em Symposium on Cloud Computing (SoCC)}, 2014.

\bibitem{psync}
C.~Dwork, N.~Lynch, and L.~Stockmeyer.
\newblock {Consensus in the Presence of Partial Synchrony}.
\newblock {\em Journal of the ACM}, 35(2), 1988.

\bibitem{atlas}
V.~Enes, C.~Baquero, T.~F. Rezende, A.~Gotsman, M.~Perrin, and P.~Sutra.
\newblock State-machine replication for planet-scale systems.
\newblock In {\em European Conference on Computer Systems (EuroSys)}, 2020.

\bibitem{faunadb}
FaunaDB.
\newblock {What is FaunaDB?}
\newblock \\https://docs.fauna.com/fauna/current/introduction.html, 2020.

\bibitem{vectorclocks1}
C.~Fidge.
\newblock Timestamps in message-passing systems that preserve the partial
  ordering.
\newblock In {\em Australian Computer Science Conference (ASCS)}, 1988.

\bibitem{cap}
S.~Gilbert and N.~Lynch.
\newblock Brewer's conjecture and the feasibility of consistent, available,
  partition-tolerant web services.
\newblock {\em SIGACT News}, 33(2), 2002.

\bibitem{google}
Google.
\newblock Balancing strong and eventual consistency with datastore.
\newblock
  \\https://cloud.google.com/datastore/docs/articles/\\balancing-strong-and-eventual-consistency-with-google-cloud-datastore,
  2020.

\bibitem{multicast-dsn19}
A.~Gotsman, A.~Lefort, and G.~Chockler.
\newblock White-box atomic multicast.
\newblock In {\em International Conference on Dependable Systems and Networks
  (DSN)}, 2019.

\bibitem{cise-popl16}
A.~Gotsman, H.~Yang, C.~Ferreira, M.~Najafzadeh, and M.~Shapiro.
\newblock {'C}ause {I}'m strong enough: reasoning about consistency choices in
  distributed systems.
\newblock In {\em Symposium on Principles of Programming Languages (POPL)},
  2016.

\bibitem{eunomia}
C.~Gunawardhana, M.~Bravo, and L.~Rodrigues.
\newblock Unobtrusive deferred update stabilization for efficient
  geo-replication.
\newblock In {\em {USENIX} Annual Technical Conference (USENIX ATC)}, 2017.

\bibitem{linearizability}
M.~P. Herlihy and J.~M. Wing.
\newblock Linearizability: A correctness condition for concurrent objects.
\newblock {\em ACM Trans. Program. Lang. Syst.}, 12(3), 1990.

\bibitem{suresh}
G.~Kaki, K.~Earanky, K.~C. Sivaramakrishnan, and S.~Jagannathan.
\newblock Safe replication through bounded concurrency verification.
\newblock {\em Proc. {ACM} Program. Lang.}, 2({OOPSLA}), 2018.

\bibitem{mdcc}
T.~Kraska, G.~Pang, M.~J. Franklin, S.~Madden, and A.~Fekete.
\newblock {MDCC}: Multi-data center consistency.
\newblock In {\em European Conference on Computer Systems (EuroSys)}, 2013.

\bibitem{hlc}
S.~S. Kulkarni, M.~Demirbas, D.~Madappa, B.~Avva, and M.~Leone.
\newblock Logical physical clocks.
\newblock In {\em International Conference on Principles of Distributed Systems
  ({OPODIS})}, 2014.

\bibitem{lazy}
R.~Ladin, B.~Liskov, L.~Shrira, and S.~Ghemawat.
\newblock Providing high availability using lazy replication.
\newblock {\em ACM Trans. Comput. Syst.}, 10(4), 1992.

\bibitem{paxos}
L.~Lamport.
\newblock The part-time parliament.
\newblock {\em ACM Trans. Comput. Syst.}, 16(2), 1998.

\bibitem{generalized}
L.~Lamport.
\newblock Generalized consensus and {P}axos.
\newblock Technical Report MSR-TR-2005-33, Microsoft Research, 2005.

\bibitem{cheng-tool}
C.~Li, J.~Leit{\~{a}}o, A.~Clement, N.~M. Pregui{\c{c}}a, R.~Rodrigues, and
  V.~Vafeiadis.
\newblock Automating the choice of consistency levels in replicated systems.
\newblock In {\em {USENIX} Annual Technical Conference (USENIX ATC)}, 2014.

\bibitem{red-blue}
C.~Li, D.~Porto, A.~Clement, R.~Rodrigues, N.~Pregui\c{c}a, and J.~Gehrke.
\newblock Making geo-replicated systems fast if possible, consistent when
  necessary.
\newblock In {\em Symposium on Operating Systems Design and Implementation
  (OSDI)}, 2012.

\bibitem{por}
C.~Li, N.~Pregui{\c c}a, and R.~Rodrigues.
\newblock Fine-grained consistency for geo-replicated systems.
\newblock In {\em {USENIX} Annual Technical Conference (USENIX ATC)}, 2018.

\bibitem{cops}
W.~Lloyd, M.~J. Freedman, M.~Kaminsky, and D.~G. Andersen.
\newblock Don't settle for eventual: Scalable causal consistency for wide-area
  storage with cops.
\newblock In {\em Symposium on Operating Systems Principles (SOSP)}, 2011.

\bibitem{eiger}
W.~Lloyd, M.~J. Freedman, M.~Kaminsky, and D.~G. Andersen.
\newblock Stronger semantics for low-latency geo-replicated storage.
\newblock In {\em Conference on Networked Systems Design and Implementation
  (NSDI)}, 2013.

\bibitem{alvisi-cc}
P.~Mahajan, L.~Alvisi, and M.~Dahlin.
\newblock Consistency, availability, and convergence.
\newblock Technical Report TR-11-22, University of Texas at Austin, 2011.

\bibitem{k-stability}
P.~Mahajan, S.~Setty, S.~Lee, A.~Clement, L.~Alvisi, M.~Dahlin, and M.~Walfish.
\newblock Depot: Cloud storage with minimal trust.
\newblock {\em ACM Trans. Comput. Syst.}, 29(4), 2011.

\bibitem{vectorclocks2}
F.~Mattern.
\newblock Virtual time and global clocks in distributed systems.
\newblock In {\em International Workshop on Parallel and Distributed
  Algorithms}, 1988.

\bibitem{occult}
S.~A. Mehdi, C.~Littley, N.~Crooks, L.~Alvisi, N.~Bronson, and W.~Lloyd.
\newblock I can't believe it's not causal! {S}calable causal consistency with
  no slowdown cascades.
\newblock In {\em Conference on Networked Systems Design and Implementation
  (NSDI)}, 2017.

\bibitem{cosmosdb}
Microsoft.
\newblock Consistency levels in {A}zure {C}osmos {DB}.
\newblock
  \\https://docs.microsoft.com/en-us/azure/\\cosmos-db/consistency-levels,
  2020.

\bibitem{epaxos}
I.~Moraru, D.~G. Andersen, and M.~Kaminsky.
\newblock There is more consensus in egalitarian parliaments.
\newblock In {\em Symposium on Operating Systems Principles (SOSP)}, 2013.

\bibitem{sreeja}
S.~S. Nair, G.~Petri, and M.~Shapiro.
\newblock Proving the safety of highly-available distributed objects.
\newblock In {\em European Symposium on Programming (ESOP)}, 2020.

\bibitem{cise-tool}
M.~Najafzadeh, A.~Gotsman, H.~Yang, C.~Ferreira, and M.~Shapiro.
\newblock The {CISE} tool: proving weakly-consistent applications correct.
\newblock In {\em Workshop on the Principles and Practice of Consistency for
  Distributed Data (PaPoC)}, 2016.

\bibitem{generic}
F.~Pedone and A.~Schiper.
\newblock Generic broadcast.
\newblock In {\em International Symposium on Distributed Computing (DISC)},
  1999.

\bibitem{anti-entropy}
K.~Petersen, M.~J. Spreitzer, D.~B. Terry, M.~M. Theimer, and A.~J. Demers.
\newblock Flexible update propagation for weakly consistent replication.
\newblock In {\em Symposium on Operating Systems Principles (SOSP)}, 1997.

\bibitem{redis}
{Redis Labs}.
\newblock Causal consistency in an active-active database.
\newblock
  \\https://docs.redislabs.com/latest/rs/administering/\\database-operations/causal-consistency-crdb/,
  2020.

\bibitem{smr}
F.~B. Schneider.
\newblock Implementing fault-tolerant services using the state machine
  approach: A tutorial.
\newblock {\em ACM Comput. Surv.}, 22(4), 1990.

\bibitem{crdts}
M.~Shapiro, N.~M. Pregui\c{c}a, C.~Baquero, and M.~Zawirski.
\newblock Conflict-free replicated data types.
\newblock In {\em International Symposium on Stabilization, Safety, and
  Security of Distributed Systems (SSS)}, 2011.

\bibitem{documentdb}
D.~Shukla, S.~Thota, K.~Raman, M.~Gajendran, A.~Shah, S.~Ziuzin, K.~Sundaram,
  M.~G. Guajardo, A.~Wawrzyniak, S.~Boshra, R.~Ferreira, M.~Nassar,
  M.~Koltachev, J.~Huang, S.~Sengupta, J.~J. Levandoski, and D.~B. Lomet.
\newblock Schema-agnostic indexing with azure documentdb.
\newblock {\em Proc. {VLDB} Endow.}, 8(12), 2015.

\bibitem{walter}
Y.~Sovran, R.~Power, M.~K. Aguilera, and J.~Li.
\newblock Transactional storage for geo-replicated systems.
\newblock In {\em Symposium on Operating Systems Principles (SOSP)}, 2011.

\bibitem{wren}
K.~{Spirovska}, D.~{Didona}, and W.~{Zwaenepoel}.
\newblock Wren: Nonblocking reads in a partitioned transactional causally
  consistent data store.
\newblock In {\em International Conference on Dependable Systems and Networks
  (DSN)}, 2018.

\bibitem{paris}
K.~{Spirovska}, D.~{Didona}, and W.~{Zwaenepoel}.
\newblock Paris: Causally consistent transactions with non-blocking reads and
  partial replication.
\newblock In {\em International Conference on Distributed Computing Systems
  (ICDCS)}, 2019.

\bibitem{pocc}
K.~{Spirovska}, D.~{Didona}, and W.~{Zwaenepoel}.
\newblock Optimistic causal consistency for geo-replicated key-value stores.
\newblock {\em IEEE Transactions on Parallel and Distributed Systems}, 32(3),
  2020.

\bibitem{cockroachdb}
R.~Taft, I.~Sharif, A.~Matei, N.~VanBenschoten, J.~Lewis, T.~Grieger, K.~Niemi,
  A.~Woods, A.~Birzin, R.~Poss, P.~Bardea, A.~Ranade, B.~Darnell, B.~Gruneir,
  J.~Jaffray, L.~Zhang, and P.~Mattis.
\newblock {CockroachDB: The Resilient Geo-Distributed {SQL} Database}.
\newblock In {\em International Conference on Management of Data ({SIGMOD})},
  2020.

\bibitem{session}
D.~B. Terry, A.~J. Demers, K.~Petersen, M.~Spreitzer, M.~Theimer, and B.~W.
  Welch.
\newblock Session guarantees for weakly consistent replicated data.
\newblock In {\em International Conference on Parallel and Distributed
  Information Systems (PDIS)}, 1994.

\bibitem{pileus}
D.~B. Terry, V.~Prabhakaran, R.~Kotla, M.~Balakrishnan, M.~K. Aguilera, and
  H.~Abu{-}Libdeh.
\newblock Consistency-based service level agreements for cloud storage.
\newblock In {\em Symposium on Operating Systems Principles (SOSP)}, 2013.

\bibitem{bayou}
D.~B. Terry, M.~M. Theimer, K.~Petersen, A.~J. Demers, M.~J. Spreitzer, and
  C.~H. Hauser.
\newblock Managing update conflicts in {B}ayou, a weakly connected replicated
  storage system.
\newblock In {\em Symposium on Operating Systems Principles (SOSP)}, 1995.

\bibitem{MongoDB}
M.~Tyulenev, A.~Schwerin, A.~Kamsky, R.~Tan, A.~Cabral, and J.~Mulrow.
\newblock Implementation of cluster-wide logical clock and causal consistency
  in {MongoDB}.
\newblock In {\em International Conference on Management of Data (SIGMOD)},
  2019.

\bibitem{horus}
R.~van Renesse, K.~P. Birman, and S.~Maffeis.
\newblock Horus: A flexible group communication system.
\newblock {\em Commun. ACM}, 39(4), 1996.

\bibitem{vogels}
W.~Vogels.
\newblock Eventually consistent.
\newblock {\em CACM}, 52(1), 2009.

\bibitem{wv}
G.~Weikum and G.~Vossen.
\newblock {\em Transactional Information Systems: Theory, Algorithms, and the
  Practice of Concurrency Control and Recovery}.
\newblock Morgan Kaufmann Publishers Inc., 2001.

\bibitem{yugabytedb}
YugabyteDB.
\newblock Replication.
\newblock
  \\https://docs.yugabyte.com/latest/architecture/\\docdb-replication/replication/,
  2020.

\bibitem{swiftcloud}
M.~Zawirski, N.~Pregui\c{c}a, S.~Duarte, A.~Bieniusa, V.~Balegas, and
  M.~Shapiro.
\newblock Write fast, read in the past: Causal consistency for client-side
  applications.
\newblock In {\em International Middleware Conference (Middleware)}, 2015.

\bibitem{tapir}
I.~Zhang, N.~K. Sharma, A.~Szekeres, A.~Krishnamurthy, and D.~R.~K. Ports.
\newblock Building consistent transactions with inconsistent replication.
\newblock {\em ACM Trans. Comput. Syst.}, 35(4), 2018.

\end{thebibliography}

\iflong
\clearpage
\appendix

\renewcommand{\thealgorithm}{\thesection\arabic{algorithm}}

\section{The Full \unistore{} Protocol for LWW Registers} \label{section:unistore-protocol}

Algorithms~\ref{alg:unistore-client} -- \ref{alg:unistore-recovery} given in
this section define the full \unistore{} protocol, including parts omitted from
the main text. This version of the protocol is specialized to the case when the
data store manages last-writer-wins (LWW) registers.

Algorithm~\ref{alg:unistore-client} shows the pseudocode of clients.
We assume that each client is associated with a unique client identifier.
Each client maintains the following variables:
\begin{itemize}
  \item $\lc$: The Lamport clock at this client.
  \item $\clientdc$: The data center to which this client is currently connected.
  \item $\clientcoord$: The coordinator partition of the current ongoing transaction.
  \item $\ctid$: The identifier of the current ongoing transaction.
  \item $\pastVC$: The client's causal past.
\end{itemize}

A client interacts with \unistore{} via the following procedures:
\begin{itemize}
  \item $\tidvar \gets \Call{\start}{\null}$:
    Start a transaction and obtain an identifier $\tidvar$.
  \item $v \gets \Call{\read}{k}$:
    Invoke a read operation on key $k$ in the ongoing transaction and obtain a
    return value $v$.
  \item $\ok \gets \Call{\updateproc}{k, v}$:
    Invoke an update operation on key $k$ and value $v$
    in the ongoing transaction.
  \item $\ok \gets \Call{\commitcausaltx}{\null}$:
    Commit a causal transaction.
  \item $\decvar \gets \Call{\commitstrongtx}{\null}$:
    Try to commit a strong transaction
    and obtain a decision $\decvar \in \set{\commit, \abort}$.
  \item $\ok \gets \Call{\fence}{\null}$:
    Execute a uniform barrier.
  \item $\ok \gets \Call{\clattach}{j}$:
    Attach to data center $j$.
\end{itemize}

Algorithms~\ref{alg:unistore-coord} -- \ref{alg:unistore-strong-commit} show the
pseudocode of replicas. The code needed for strong transactions in
Algorithms~\ref{alg:unistore-coord}, \ref{alg:unistore-replica}, and
\ref{alg:unistore-clock} is highlighted in red.  Each replica $p^{m}_{d}$
maintains a set of variables as follows.
\begin{itemize}
  \item{$\clockVar$:} The current time at $p^{m}_{d}$.
  \item{$\rset$:} The read sets of transactions coordinated by $p^{m}_{d}$,
    indexed by transaction identifier $\tidvar$.
  \item{$\wbuff$:} The write buffers of transactions coordinated by $p^{m}_{d}$,
    indexed by transaction identifier $\tidvar$, partition $l$, and key $k$.
  \item{$\snapVC$:} The snapshot vectors of transactions coordinated by $p^{m}_{d}$,
    indexed by transaction identifier $\tidvar$.
  \item{$\oplog$:} The log of updates performed on keys managed by $p^{m}_{d}$,
    indexed by key $k$.
  \item{$\knownVC, \stableVC, \uniformVC$:}
    The vectors used by $p^{m}_{d}$ to track what is replicated where.
  \item{$\preparedcausal$:} The set of causal transactions
    local to $p^{m}_{d}$ that are prepared to commit.
  \item{$\committedcausal$:} For each data center $i$,
    $\committedcausal[i]$ stores transactions waiting to be
    replicated by $p^{m}_{d}$ to sibling replicas at other data centers than $i$.
  \item{$\localmatrix$:} The set of $\knownVC$ received by $p^{m}_{d}$
    from other partitions in data center $d$.
    It is used to compute $\stableVC$.
  \item{$\stablematrix$:} The set of $\stableVC$ received by $p^{m}_{d}$
    from sibling replicas.
    It is used to compute $\uniformVC$.
  \item{$\globalmatrix$:} The set of $\knownVC$ received by $p^{m}_{d}$
    from sibling replicas.
    It is used to track what has been replicated at sibling replicas.
\end{itemize}

We specialize the \unistore{} protocol to LWW registers in several ways.
First, we add code for managing Lamport clocks, highlighted in blue.
In particular, each (committed) transaction is associated with a Lamport
timestamp, equal to the value of $\lc$ at its client when the transaction completes
(lines~\code{\ref{alg:unistore-client}}{\ref{line:commitcausaltx-lc}}
and \code{\ref{alg:unistore-client}}{\ref{line:commitstrongtx-lc}}).
Lamport timestamps are totally ordered, with client identifiers used for tie-breaking
(see also Definition~\ref{def:lco}).

Second, in Algorithm~\ref{alg:unistore-coord}
we replace \doop{} by the following two procedures:

\begin{itemize}
  \item $v \gets \Call{\doread}{\tidvar, k}$:
    Execute a read operation on key $k$
    in a transaction with identifier $\tidvar$ and obtain a value $v$.
  \item $\Call{\doupdate}{\tidvar, k, v}$:
    Execute an update operation on key $k$ and value $v$
    in a transaction with identifier $\tidvar$.
\end{itemize}

Finally, in Algorithm~\ref{alg:unistore-replica}
we modify the handler of message \getversion{}:
\begin{itemize}
  \item $\Call{\readkey}{\snapvc, k}$:
    Read the latest value from key $k$
    based on the snapshot vector $\snapvc$.
    Specifically, it returns the last update to key $k$ of the transaction
    with the latest $\commitvc$ in terms of their Lamport clock order
    such that $\commitvc \leq \snapvc$
    (line~\code{\ref{alg:unistore-replica}}{\ref{line:readkey-read}}).
\end{itemize}

Algorithms~\ref{alg:unistore-certify} -- \ref{alg:unistore-recovery} contain the
implementation of the transaction certification service (see also
\S\ref{section:tcs}). The certification service uses an instance of the leader
election failure detector $\Omega_m$ for each partition $m$\footnote{Tushar
  Deepak Chandra, Vassos Hadzilacos, and Sam Toueg. The weakest failure detector
  for solving consensus. J. ACM, 1996.}. This primitive ensures that from some
point on, all correct processes nominate the same correct process as the leader.
For the case when the data store manages only registers, we assume that any two
writes to the same object conflict. Other data types and conflict relations can
be easily supported by modifying the certification check in
Algorithm~\ref{alg:unistore-certification}.

\setcounter{algorithm}{0}

\begin{algorithm*}[t]
  \caption{Client operations at client $\cl$}
  \label{alg:unistore-client}
  \begin{algorithmic}[1]
    \Function{\start}{\null}
      \label{line:function-starttx}
      \State $\clientcoord \gets$ an arbitrary replica in data center $\clientdc$
        \label{line:starttx-random-partition}
      \State $\ctid \gets \rpc \Call{\starttx}{\pastVC} \at \clientcoord$
        \label{line:starttx-call-start}
        \Comment{\tscolor{$\timestamp(\start) \gets
          \snapVC^{\clientcoord}_{\clientdc}[\ctid]$}}
        \label{line:starttx-ts}
      \State \Return $\ctid$
        \label{line:starttx-return}
    \EndFunction

    \Statex
    \Function{\read}{$k$} \label{line:function-read}
      \State $\langle v, \lccolor{\lcvar} \rangle \gets
        \rpc \Call{\doread}{\ctid, k} \at \clientcoord$
        \label{line:read-call-doread}
      \If{$\lcvar \neq \bot$}
        \label{line:read-is-external}
        \State \lccolor{$\lc \gets \max\set{\lc, \lcvar}$}
          \label{line:read-lc}
      \EndIf
      \State \Return $v$
        \label{line:read-return}
    \EndFunction

    \Statex
    \Function{\updateproc}{$k, v$} \label{line:function-update}
      \State $\rpc \Call{\doupdate}{\ctid, k, v} \at \clientcoord$
        \label{line:update-call-doupdate}
      \State \Return $\ok$
        \label{line:update-return}
    \EndFunction

    \Statex
    \Function{\commitcausaltx}{\null}
      \label{line:function-commitcausaltx}
      \State \lccolor{$\lc \gets \lc + 1$}
        \Comment{\lccolor{$\lclock(\commitcausaltx) \gets \lc$}}
        \label{line:commitcausaltx-lc}
      \State $\vcvar \gets \rpc \Call{\commitcausal}{\ctid, \lccolor{\lc}} \at \clientcoord$
        \label{line:commitcausaltx-call-commitcausal}
      \State $\pastVC \gets \vcvar$
        \label{line:commitcausaltx-pastvc}
        \Comment{\tscolor{$\timestamp(\commitcausaltx) \gets \pastVC$}}
        \label{line:commitcausaltx-ts}
      \State \Return $\ok$
        \label{line:commitcausaltx-return}
    \EndFunction

    \Statex
    \Function{\commitstrongtx}{\null}
      \label{line:function-commitstrongtx}
      \State \lccolor{$\lc \gets \lc + 1$}
        \label{line:commitstrongtx-lc-so}
      \State $\langle \decvar, \vcvar, \lccolor{\lcvar} \rangle \gets \rpc
        \Call{\commitstrong}{\ctid, \lccolor{\lc}} \at \clientcoord$
        \label{line:commitstrongtx-call-commitstrong}
      \If{$\decvar = \commit$}
        \label{line:commitstrongtx-if-commit}
        \State $\pastVC \gets \vcvar$
          \label{line:commitstrongtx-pastvc}
          \Comment{\tscolor{$\timestamp(\commitstrongtx) \gets \pastVC$}}
          \label{line:commitstrongtx-ts}
        \State \lccolor{$\lc \gets \lcvar$}
          \Comment{\lccolor{$\lclock(\commitstrongtx) \gets \lc$}}
          \label{line:commitstrongtx-lc}
      \EndIf
      \State \Return $dec$
        \label{line:commitstrongtx-return}
    \EndFunction

    \Statex
    \Function{\fence}{\null} \label{line:function-fence}
      \State \var $\p \gets$ an arbitrary replica in data center $\clientdc$
        \label{line:fence-partition}
      \State $\rpc \Call{\uniformbarrier}{\pastVC} \at \p$
        \label{line:fence-call-uniformbarrier}
        \Comment{\tscolor{$\timestamp(\fence) \gets \pastVC$}}
        \label{line:fence-ts}
      \State \lccolor{$\lc \gets \lc + 1$}
        \Comment{\lccolor{$\lclock(\fence) \gets \lc$}}
        \label{line:fence-lc}
      \State \Return $\ok$
        \label{line:fence-return}
    \EndFunction

    \Statex
    \Function{\clattach}{$j$} \label{line:function-migrate}
      \State \var $\p \gets$ an arbitrary replica in data center $j$
        \label{line:clattach-partition}
      \State $\rpc \Call{\attach}{\pastVC} \at \p$
        \label{line:clattach-call-attach}
        \Comment{\tscolor{$\timestamp(\clattach) \gets \pastVC$}}
        \label{line:clattach-ts}
      \State \lccolor{$\lc \gets \lc + 1$}
        \Comment{\lccolor{$\lclock(\clattach) \gets \lc$}}
        \label{line:clattach-lc}
      \State $\clientdc \gets j$
        \label{line:clattach-j}
      \State \Return $\ok$
        \label{line:clattach-return}
    \EndFunction
  \end{algorithmic}
\end{algorithm*}

\begin{algorithm*}[t]
  \caption{Transaction coordinator at $p^{m}_{d}$: causal commit}
  \label{alg:unistore-coord}
  \begin{algorithmic}[1]
    \Function{\starttx}{$\vc$} \label{line:function-start}
      \For{$i \in \D \setminus \set{d}$}
        \label{line:start-uniformvc-index}
        \State $\uniformVC[i] \gets \max\set{\vc[i], \uniformVC[i]}$
        \label{line:start-uniformvc}
      \EndFor

      \State {\bf var} $\tidvar \gets \generatetid()$
        \label{line:start-tid}
      \State $\snapVC[\tidvar] \gets \uniformVC$
        \label{line:start-snapvc}
      \State $\snapVC[\tidvar][d] \gets \max\set{\vc[d], \uniformVC[d]}$
        \label{line:start-snapvc-d}
      \State \strongcolor{$\snapVC[\tidvar][\strongentry] \gets
        \max\set{\vc[\strongentry], \stableVC[\strongentry]}$}
        \label{line:start-snapvc-strong}
      \State \Return $\tidvar$
        \label{line:start-return}
        \label{line:start-snapshotvc-of-t}
    \EndFunction

    \Statex
    \Function{\doread}{$\tidvar, k$}
      \label{line:function-doread}
      \State \var $l \gets \partitionofproc(k)$
        \label{line:doread-partition-of-k}
      \If{$\wbuff[\tidvar][l][k] \neq \bot$}
        \label{line:doread-from-buffer}
        \State \Return $\langle \wbuff[\tidvar][l][k], \lccolor{\bot} \rangle$
          \label{line:doread-return-from-buffer}
      \EndIf
      \State $\send \Call{\readkey}{\snapVC[\tidvar], k} \sendto p^{l}_{d}$
        \label{line:doread-from-snapshot}
      \State \wait\receive $\versionproc(v, \lccolor{\lcvar}) \from p^{l}_{d}$
      \State \strongcolor{$\rset[\tidvar] \gets \rset[\tidvar] \cup \set{k}$}
        \label{line:doread-readset}
      \State \Return $\langle v, \lccolor{\lcvar} \rangle$
        \label{line:doread-return-from-snapshot}
    \EndFunction

    \Statex
    \Function{\doupdate}{$\tidvar, k, v$}
        \label{line:function-doupdate}
      \State \var $l \gets \partitionofproc(k)$
        \label{line:doupdate-partitionof-k}
      \State $\wbuff[\tidvar][l][k] \gets v$
        \label{line:doupdate-wbuff}
      \State \strongcolor{$\rset[\tidvar] \gets \rset[\tidvar] \cup \set{k}$}
        \label{line:doupdate-readset}
    \EndFunction

    \Statex
    \Function{\commitcausal}{$\tidvar, \lccolor{\lcvar}$}
      \label{line:function-commitcausal}
      \State \var $L \gets \set{l \mid \wbuff[\tidvar][l] \neq \emptyset}$
      \If{$L = \emptyset$}
        \label{line:commitcausal-ro}
        \State \Return $\snapVC[\tidvar]$
        \label{line:commitcausal-return-ro}
        \label{line:commitcausal-commitvc-of-t-ro}
      \EndIf

      \hStatex
      \State \send $\prepare(\tidvar, \wbuff[\tidvar][l], \snapVC[\tidvar])
        \sendto p^{l}_{d}$, $l \in L$
        \label{line:commitcausal-call-prepare}
      \State \var $\commitvc \gets \snapVC[\tidvar]$
        \label{line:commitcausal-commitvc}
      \ForAll{$l \in L$}
        \State \wait\receive $\prepareack(\tidvar, \tsvar) \from p^{l}_{d}$
          \label{line:commitcausal-wait-prepareack}
        \State $\commitvc[d] \gets \max \set{\commitvc[d], \tsvar}$
          \label{line:commitcausal-commitvc-d}
      \EndFor
      \State \send $\Call{commit}{\tidvar, \commitvc, \lccolor{\lcvar}}
        \sendto p^{l}_{d}$, $l \in L$
      \label{line:commitcausal-call-commit}
      \State \Return $\commitvc$
        \label{line:commitcausal-return}
        \label{line:commitcausal-commitvc-of-t-rw}
    \EndFunction
  \end{algorithmic}
\end{algorithm*}

\begin{algorithm*}[t]
  \caption{Transaction execution at $p^m_d$}
  \label{alg:unistore-replica}
  \begin{algorithmic}[1]
    \WhenRcv[$\Call{\readkey}{\snapvc, k} \from p$]
      \label{line:function-readkey}
      \For{$i \in \D \setminus \set{d}$}
        \label{line:readkey-uniformvc-i}
        \State $\uniformVC[i] \gets \max\{\snapvc[i], \uniformVC[i]\}$
          \label{line:readkey-uniformvc}
      \EndFor
      \State {\bf wait until} $\knownVC[d] \geq \snapvc[d]
          \land \strongcolor{\knownVC[\strongentry] \ge \snapvc[\strongentry]}$
        \label{line:readkey-wait-util-knownvc}
      \hStatex
      \State $\langle v, \commitvc, \lccolor{\lcvar} \rangle \gets \snapshotproc(\oplog[k], \snapvc)$
        \label{line:readkey-read}
        \Comment{returns the last update to key $k$ by a transaction}
      \State  \Comment{
          with the highest Lamport timestamp such that $\commitvc \leq \snapvc$}
      \State \send $\versionproc(v, \lccolor{\lcvar}) \sendto p$
        \label{line:readkey-return}
    \EndWhenRcv

    \Statex
    \WhenRcv[$\Call{\prepare}{\tidvar, \wbuffvar, \snapvc} \from p$]
      \label{line:function-preparecausal}
      \For{$i \in \D \setminus \set{d}$}
        \State $\uniformVC[i] \gets \max\set{\snapvc[i], \uniformVC[i]}$
        \label{line:preparecausal-uniformvc}
      \EndFor

      \State \var $\tsvar \gets \clockVar$
        \label{line:preparecausal-ts}
      \State $\preparedcausal \gets \preparedcausal \cup
        \set{\langle \tidvar, \wbuffvar, \tsvar \rangle}$
        \label{line:preparecausal-preparedcausal}
      \State \send $\Call{\prepareack}{\tidvar, \tsvar} \sendto p$
        \label{line:preparecausal-call-preparecausalack}
    \EndWhenRcv

    \Statex
    \WhenRcv[$\Call{\commit}{\tidvar, \commitvc, \lccolor{\lcvar}}$]
      \label{line:function-commit}
      \State \wait\until $\clockVar \geq \commitvc[d]$
        \label{line:commit-wait-clock}

      \hStatex
      \State $\langle \tidvar, \wbuffvar, \_ \rangle \gets \find(\tidvar, \preparedcausal)$
        \label{line:commit-find}
      \State $\preparedcausal \gets \preparedcausal \setminus \set{\langle \tidvar, \_, \_\rangle}$
        \label{line:commit-preparedcausal}

      \ForAll{$\langle k, v \rangle \in \wbuffvar$}
        \State $\oplog[k] \gets \oplog[k] \cdot \langle v, \commitvc, \lccolor{\lcvar} \rangle$
        \label{line:commit-oplog}
      \EndFor

      \State $\committedcausal[d] \gets \committedcausal[d] \cup
        \set{\langle \tidvar, \wbuffvar, \commitvc, \lccolor{\lcvar} \rangle}$
        \label{line:commit-committedcausal}
    \EndWhenRcv

    \Statex
    \Function{\uniformbarrier}{$\vc$} \label{line:function-uniformbarrier}
      \State \wait\until $\uniformVC[d] \ge \vc[d]$
        \label{line:uniformbarrier-wait-uniformvc-d}
    \EndFunction

    \Statex
    \Function{\attach}{$\vc$}
      \label{line:function-attach}
      \State \wait\until $\forall i \in \D \setminus \set{d}.\; \uniformVC[i] \ge \vc[i]$
        \label{line:attach-wait-condition}
    \EndFunction
  \end{algorithmic}
\end{algorithm*}

\begin{algorithm*}[t]
  \caption{Transaction replication at $p^m_d$}
  \label{alg:unistore-replication}
  \begin{algorithmic}[1]
    \Function{\propagate}{\null} \Comment{Run periodically}
      \label{line:function-propagate}
      \If{$\preparedcausal = \emptyset$}
          \label{line:propagate-preparedcausal-empty}
        \State $\knownVC[d] \gets \clockVar$
          \label{line:propagate-knownvc-clock}
      \Else
        \State $\knownVC[d] \gets
          \min\set{\tsvar \mid \langle \_, \_, \tsvar \rangle \in \preparedcausal} - 1$
        \label{line:propagate-knownvc-ts}
      \EndIf

      \hStatex
      \State \var $\txsvar \gets \set{\langle \_, \_, \commitvc, \lccolor{\_}
        \rangle \in \committedcausal[d] \mid \commitvc[d] \le \knownVC[d]}$
        \label{line:propagate-txs}

      \If{$\txsvar \neq \emptyset$}
        \label{line:propagate-txs-nonempty}
        \State \send $\replicate(d, \txsvar) \sendto p^{m}_{i}$,
          $i \in \D \setminus \set{d}$
          \label{line:propagate-call-replicate}
        \State $\committedcausal[d] \gets \committedcausal[d] \setminus \txsvar$
          \label{line:propagate-committedblue}
      \Else
        \State \send $\heartbeat(d, \knownVC[d]) \sendto p^{m}_{i}$,
          $i \in \D \setminus \set{d}$
          \label{line:propagate-call-heartbeat}
      \EndIf
    \EndFunction

    \Statex
    \WhenRcv[\replicate($i, \txsvar$)]
      \label{line:function-replicate}
      \ForAll{$\langle \tidvar, \wbuffvar, \commitvc, \lccolor{\lcvar} \rangle \in \txsvar$
        in $\commitvc[i]$ order}
        \label{line:replicate-increasing-order}
        \If{$\commitvc[i] > \knownVC[i]$}
          \label{line:replicate-precondition}
          \ForAll{$\langle k, v \rangle \in \wbuffvar$}
            \State $\oplog[k] \gets \oplog[k] \cdot {\langle v, \commitvc, \lccolor{\lcvar} \rangle}$
            \label{line:replicate-oplog}
          \EndFor
          \State $\committedcausal[i] \gets \committedcausal[i] \cup
            \set{\langle \tidvar, \wbuffvar, \commitvc, \lccolor{\lcvar} \rangle}$
            \label{line:replicate-committedcausal}
          \State $\knownVC[i] \gets \commitvc[i]$
            \label{line:replicate-knownvc}
        \EndIf
      \EndFor
    \EndWhenRcv

    \Statex
    \WhenRcv[\heartbeat($i, \tsvar$)]
      \label{line:function-heartbeat}
      \State \pre $\tsvar > \knownVC[i]$
        \label{line:heartbeat-precondition}

      \hStatex
      \State $\knownVC[i] \gets \tsvar$
        \label{line:heartbeat-knownvc}
    \EndWhenRcv

    \Statex
    \Function{\forward}{$i, j$}
      \Comment{forward transactions received from data center $j \neq d$
        to data center $i \notin \set{d, j}$}
      \label{line:function-forward}
      \State \var $\txsvar \gets \set{\langle \tidvar, \_, \commitvc, \lccolor{\_} \rangle
        \in \committedcausal[j] \mid \commitvc[j] > \globalmatrix[i][j]}$
        \label{line:forward-txs}
      \If{$\txsvar \neq \emptyset$}
        \label{line:forward-txs-nonempty}
        \State \send $\replicate(j, \txsvar) \sendto p^m_i$
          \label{line:forward-call-replicate}
      \Else
        \State \send $\heartbeat(j, \knownVC[j]) \sendto p^m_i$
          \label{line:forward-call-heartbeat}
      \EndIf
    \EndFunction
  \end{algorithmic}
\end{algorithm*}
\clearpage

\begin{algorithm*}[t]
  \caption{Updating metadata at $p^m_d$}
  \label{alg:unistore-clock}
  \begin{algorithmic}[1]
    \Function{\bcast}{\null} \Comment{Run periodically}
      \label{line:function-bcast}
      \State \send $\Call{\knownvclocal}{m, \knownVC} \sendto p^{l}_{d}$, $l \in \P$
        \label{line:bcast-call-knownvclocal}
      \State \send $\Call{\stablevcproc}{d, \stableVC} \sendto p^{m}_{i}$, $i \in \D$
        \label{line:bcast-call-stablevc}
      \State \send $\Call{\knownvcglobal}{d, \knownVC} \sendto p^{m}_{i}$, $i \in \D$
        \label{line:bcast-call-knownvcglobal}
    \EndFunction

    \Statex
    \WhenRcv[\knownvclocal($l, \knownvc$)]
      \label{line:function-knownvclocal}
      \State $\localmatrix[l] \gets \knownvc$
        \label{line:knownvclocal-localknownmatrix}
      \For{$i \in \D$}
        \State $\stableVC[i] \gets \min \set{\localmatrix[n][i] \mid n \in \P}$
        \label{line:knownvclocal-stablevc-causal}
      \EndFor
      \State \strongcolor{$\stableVC[\strongentry] \gets
        \min\set{\localmatrix[n][\strongentry] \mid n \in \P}$}
        \label{line:knownvclocal-stablevc-strong}
    \EndWhenRcv

    \Statex
    \WhenRcv[\stablevcproc($i, \stablevc$)]
      \label{line:function-stablevc}
      \State $\stablematrix[i] \gets \stablevc$
        \label{line:stablevc-stablematrix}
      \State $G\gets$ all groups with $f+1$ replicas that include $p^{m}_{d}$
        \label{line:stablevc-g}
      \For{$j \in \D$}
        \State \var $\tsvar \gets \max \set{\min \set{\stablematrix[h][j] \mid h \in g}
          \mid g \in G}$
          \label{line:stablevc-ts}
        \State $\uniformVC[j] \gets \max \set{\uniformVC[j], \tsvar}$
          \label{line:stablevc-uniformvc}
      \EndFor
    \EndWhenRcv

    \Statex
    \WhenRcv[\knownvcglobal($l, \knownvc$)]
      \label{line:function-knownvcglobal}
      \State $\globalmatrix[l] \gets \knownvc$
        \label{line:knownvcglobal-globalmatrix}
    \EndWhenRcv
  \end{algorithmic}
\end{algorithm*}

\begin{algorithm*}[t]
  \caption{Committing strong transactions at $p^{m}_{d}$}
  \label{alg:unistore-strong-commit}
  \begin{algorithmic}[1]
    \Function{\commitstrong}{$\tidvar, \lccolor{\lcvar}$}
      \label{line:function-commitstrong}
      \State $\uniformbarrier(\snapVC[\tidvar])$
        \label{line:commitstrong-call-uniformbarrier}
      \State $\langle \decvar, \vcvar, \lcvar \rangle \gets
        \Call{\certify}{\normalMode,
          \tidvar, \wbuff[\tidvar], \rset[\tidvar], \snapVC[\tidvar], \lccolor{\lcvar}}$
        \label{line:commitstrong-call-certify}
      \State \Return $\langle \decvar, \vcvar, \lccolor{\lcvar} \rangle$
        \label{line:commitstrong-return}
        \label{line:commitstrong-commitvc-of-t-strong}
    \EndFunction

    \Statex
    \Upon[\deliverupdates($\txsvar$)]
      \label{line:function-deliverupdates}
      \For{$\langle \_, \wbuffvar, \commitvc, \lccolor{\lcvar} \rangle \in \txsvar$ in $\commitvc[\strongentry]$ order}
        \label{line:deliverupdates-foreach-wbuff}
          \For{$\langle k, v \rangle \in \wbuffvar$}
            \label{line:deliverupdates-foreach-kv}
            \State $\oplog[k] \gets \oplog[k] \cdot \langle v, \commitvc, \lccolor{\lcvar} \rangle$
              \label{line:deliverupdates-oplog}
          \EndFor
        \State $\knownVC[\strongentry] \gets \commitvc[\strongentry]$
          \label{line:deliverupdates-knownvc-strongentry}
      \EndFor
    \EndUpon

    \Statex
    \Function{\heartbeatstrong}{\null} \Comment{Run periodically}
      \label{line:function-heartbeatstrong}
      \State \Return $\Call{\certify}{\normalMode, \bot, \emptyset, \emptyset, \vec{0}, \lccolor{\bot}}$
        \label{line:heartbeatstrong-call-certify}
    \EndFunction
  \end{algorithmic}
\end{algorithm*}

\begin{algorithm*}[t]
  \caption{Certification service at coordinator $\pvar^{m}_{d}$}
  \label{alg:unistore-certify}
  \begin{algorithmic}[1]
    \Function{\certify}{$\callerMode, \tidvar, \wbuffvar, \rsetvar, \snapvc, \lccolor{\lcvar}$}
      \label{line:function-certify}
      \State \var $\reqIdVar \gets \generateReqId()$
      \State \var $L \gets \set{l \mid \wbuffvar[l] \neq \emptyset}
        \cup \set{\partitionofproc(k) \mid k \in \rsetvar}$
        \label{line:certify-L}
      \Repeat
        \State \send
          $\Call{\preparestrong}{\reqIdVar, \callerMode,
            \tidvar, \wbuffvar, \rsetvar, \snapvc, \lccolor{\lcvar}}
          \sendto \Omega_{l}$, $l \in L$
          \label{line:certify-call-preparestrong}
        \State \asyncwait \receive $\Call{\alreadydecided}{\tidvar, \decisionvar,
          \commitvc, \lccolor{\lcvar}}$
          \label{line:certify-wait-alreadydecided}
        \Statex \hspace{1.95cm} $\lor$ \receive $\Call{\unknowntxAck}{l, \reqIdVar, \tidvar}$
          \textbf{from a quorum from some} $l \in L$
        \Statex \hspace{1.95cm} $\lor$ \receive
          $\Call{\acceptack}{l, \ballotvar_{l}, \tidvar, \votevar_{l}, \tsvar_{l}, \lccolor{\lcvar_{l}}}$
            \textbf{from a quorum for all} $l \in L$
          \label{line:certify-wait-acceptack}
      \Until{\notkw\timeout}

      \hStatex
      \If{\received \textbf{a quorum of} $\Call{\unknowntxAck}{l, \reqIdVar, \tidvar}$}
        \State \Return $\unknowntx$
      \ElsIf{\received $\Call{\alreadydecided}{\tidvar, \decisionvar, \commitvc, \lccolor{\lcvar}}$}
        \label{line:certify-received-alreadydecided}
        \State \send $\Call{\decision}{\ballotVar, \tidvar, \decisionvar, \commitvc, \lccolor{\lcvar}}
          \sendto \Omega_{\pvar}$
        \State \Return $\langle \decisionvar, \commitvc, \lccolor{\lcvar} \rangle$
          \label{line:certify-return-alreadydecided}
      \Else
        \State $\commitvc \gets \snapvc$
          \label{line:certify-commitvc}
        \State $\commitvc[\strongentry] \gets \max \set{\tsvar_{l} \mid l \in L}$
          \label{line:certify-commitvc-strongentry}
        \IfThenElse{$\exists l \in L.\; \votevar_{l} = \abort$}
          {$\decisionvar \gets \abort$ \label{line:certify-decision-abort}}
          {$\decisionvar \gets \commit$ \label{line:certify-decision-commit}}
        \State \lccolor{$\lcvar \gets \max\set{\lcvar_{l} \mid l \in L}$}
          \label{line:certify-lc}
        \State \send $\Call{\decision}{\ballotvar_{l}, \tidvar, \decisionvar, \commitvc, \lccolor{\lcvar}}
          \sendto \Omega_{l}$, $l \in L$
          \label{line:certify-call-desision}
        \State \Return $\langle \decisionvar, \commitvc, \lccolor{\lcvar} \rangle$
          \label{line:certify-return-acceptack}
      \EndIf
    \EndFunction
  \end{algorithmic}
\end{algorithm*}

\begin{algorithm*}[t]
  \caption{Strong transaction certification at $p^{m}_{d}$}
  \label{alg:unistore-certification}
  \begin{algorithmic}[1]
    \Function{\certification}{$W, \rsetvar, \snapvc, \lccolor{\lcvar}$}
      \label{function:certification}
      \ForAll{$\langle \_, \wbuffvar', \rsetvar', \_, \commit, \_, \lccolor{\_} \rangle \in \preparedstrong$}
        \label{line:certification-preparedstrong}
        \If{$(\exists \langle k, \_ \rangle \in \wbuffvar'[m].\; k \in \rsetvar) \lor
             (\exists k \in \rsetvar'.\; \langle k, \_ \rangle \in W)$}
          \label{line:certification-conflict}
          \State \Return $\langle \abort, \lccolor{\bot} \rangle$
            \label{line:certification-preparedstrong-abort}
        \EndIf
      \EndFor

      \hStatex
      \ForAll{$\langle \_, \wbuffvar', \commit, \commitvc, \lccolor{\lcvar'} \rangle \in \decidedstrong$}
        \label{line:certification-decidedstrong}
        \If{$(\exists \langle k, \_ \rangle \in \wbuffvar'[m].\; k \in \rsetvar)
          \land \lnot (\commitvc \leq \snapvc)$}
          \State \Return $\langle \abort, \lccolor{\bot} \rangle$
            \label{line:certification-decidedstrong-abort}
        \EndIf
        \If{\lccolor{$\lcvar \le \lcvar'$}}
          \State \lccolor{$\lcvar \gets \lcvar' + 1$}
          \label{line:certification-lc}
        \EndIf
      \EndFor
      \State \Return $\langle \commit, \lccolor{\lcvar} \rangle$
        \label{line:certification-return}
    \EndFunction
  \end{algorithmic}
\end{algorithm*}

\begin{algorithm*}[t]
  \caption{Atomic transaction commit protocol at $p^{m}_{d}$}
  \label{alg:unistore-atomic-commit}
  \begin{algorithmic}[1]
    \WhenRcv[$\preparestrong(\reqIdVar, \senderMode, \tidvar, \wbuffvar, \rsetvar, \snapvc, \lccolor{\lcvar}) \from \pvar$]
      \label{line:function-preparestrong}
      \State \pre $\statusVar \in \set{\leaderproc, \restoring}$
        \label{line:preparestrong-precondition}
      \hStatex
      \If{$\exists \langle \tidvar, \_, \decisionvar, \commitvc, \lccolor{\lcvar} \rangle \in \decidedstrong$}
        \label{line:preparestrong-case-alreadydecided}
        \State \send $\Call{\alreadydecided}{\tidvar, \decisionvar, \commitvc, \lccolor{\lcvar}} \sendto p$
          \label{line:preparestrong-call-alreadydecided}
      \ElsIf{$\exists \langle \tidvar, \_, \_, \snapvc, \votevar, \tsvar, \lccolor{\lcvar} \rangle \in \preparedstrong$}
        \State \send $\Call{\accept}{\ballotVar, \tidvar, \wbuffvar, \rsetvar, \snapvc, \votevar, \tsvar, \pvar, \lccolor{\lcvar}}
          \sendto \replicas(\mvar)$
      \ElsIf{$\senderMode = \restoring$}
        \State \send $\Call{\unknownTx}{\ballotVar, \reqIdVar, \tidvar, \pvar} \sendto \replicas(\mvar)$
      \ElsIf{$\statusVar = \leaderproc$}
        \State $\wait\until \clockVar > \snapvc[\strongentry]$
          \label{line:preparestrong-clock}
        \State $\tsvar \gets \clockVar$
          \label{line:preaprestrong-ts}
        \State $\langle \votevar, \lccolor{\lcvar} \rangle \gets
          \Call{\certification}{\wbuffvar[m], \rsetvar, \snapvc, \lccolor{\lcvar}}$
          \label{line:preparestrong-call-certification}
        \State \send $\Call{\accept}{\ballotVar, \tidvar, \wbuffvar, \rsetvar, \snapvc,
          \votevar, \tsvar, \pvar, \lccolor{\lcvar}} \sendto \replicas(\mvar)$
          \label{line:preparestrong-call-accept}
      \EndIf
    \EndWhenRcv

    \Statex
    \WhenRcv[$\accept(\ballotvar, \tidvar, \wbuffvar, \rsetvar,
      \snapvc, \votevar, \tsvar, \pvar, \lccolor{\lcvar})$]
      \label{line:function-accept}
      \State \pre $\statusVar \in \set{\leaderproc, \follower, \restoring} \land \ballotVar = b$
        \label{line:accept-precondition}

      \hStatex
      \State $\preparedstrong \gets \preparedstrong \cup
        \set{\langle \tidvar, \wbuffvar, \rsetvar, \snapvc, \votevar, \tsvar, \lccolor{\lcvar} \rangle}$
        \label{line:accept-preparedstrong}
      \State \send $\Call{\acceptack}{\mvar, \ballotvar, \tidvar, \votevar, \tsvar, \lccolor{\lcvar}}
        \sendto p$
        \label{line:accept-call-acceptack}
    \EndWhenRcv

    \Statex
    \WhenRcv[$\decision(\ballotvar, \tidvar, \decisionvar, \commitvc, \lccolor{\lcvar})$]
      \label{line:function-decision}
      \State \pre $\statusVar \in \set{\leaderproc, \restoring} \land \ballotVar = b$
        \label{line:decision-precondition}

      \hStatex
      \State \wait\until $\clockVar \ge \commitvc[\strongentry]$
        \label{line:decision-wait-clock}
      \State \send $\Call{\learndecision}{\ballotvar, \tidvar, \decisionvar, \commitvc, \lccolor{\lcvar}}
        \sendto \replicas(\mvar)$
        \label{line:decision-call-desision}
    \EndWhenRcv

    \Statex
    \WhenRcv[$\learndecision(\ballotvar, \tidvar, \decisionvar, \commitvc, \lccolor{\lcvar})$]
      \label{function:learndecision}
      \State \pre $\statusVar \in \set{\leaderproc, \follower, \restoring} \land \ballotVar = b
        \land \exists \langle \tidvar, \wbuffvar, \_, \_, \_, \_, \lccolor{\_} \rangle \in \preparedstrong$
        \label{line:learndecision-precondition}

      \hStatex
      \State $\preparedstrong \gets \preparedstrong \setminus
        \set{\langle \tidvar, \_, \_, \_, \_, \_, \lccolor{\_} \rangle}$
        \label{line:decision-preparedstrong}
      \State $\decidedstrong \gets \decidedstrong \cup
        \set{\langle \tidvar, \wbuffvar, \decisionvar, \commitvc, \lccolor{\lcvar} \rangle}$
        \label{line:decision-decidedstrong}
    \EndWhenRcv

    \Statex \Upon[] \label{function:upcall} \label{line:function-upcall}
      \Statex \vspace{-0.70cm}
        \begin{align*}
          \exists &\langle \_, \_, \commit, \commitvc, \lccolor{\_} \rangle \in \decidedstrong. \\
            &\phantom{\land\;} \commitvc[\strongentry] > \lastdeliveredVar \\
            &\land (\lnot\exists \langle \_, \_, \_, \_, \commit, \tsvar, \lccolor{\_} \rangle \in \preparedstrong.\
            \lastdeliveredVar < \tsvar \le \commitvc[\strongentry]) \\
            &\land (\lnot\exists \langle \_, \_, \commit, \commitvc', \lccolor{\_} \rangle \in \decidedstrong.\
            \lastdeliveredVar < \commitvc'[\strongentry] < \commitvc[\strongentry])
        \end{align*}
        \label{line:upcall-upon}
      \Statex \vspace{-1.10cm}
      \State \pre $\statusVar = \leaderproc$
        \label{line:upcall-precondition}

      \hStatex
      \State $\send \deliver(\ballotVar, \commitvc[\strongentry]) \sendto \replicas(\mvar)$
    \EndUpon

    \Statex
    \WhenRcv[$\deliver(\ballotvar, \tsvar)$]
      \label{function:deliver}
      \State \pre $\statusVar \in \set{\leaderproc, \follower}
        \land \ballotVar = \ballotvar \land \lastdeliveredVar < \tsvar$
        \label{line:deliver-precondition}

      \State $\lastdeliveredVar \gets \tsvar$
        \label{line:deliver-lastdeliverd}
      \State \var $W \gets \set{\langle \tidvar, \wbuffvar[m], \commitvc, \lccolor{\lcvar} \rangle \mid
        \exists \langle \tidvar, \wbuffvar, \commit, \commitvc, \lccolor{\lcvar}
        \rangle \in\decidedstrong.\;$
        \Statex \hspace{6.22cm} $\commitvc[\strongentry] = \tsvar}$
        \label{deliver-W}
      \State \upcall $\Call{\deliverupdates}{W} \sendto p^{\mvar}_{\dvar}$
        \label{line:deliver-call-deliverupdates}
    \EndWhenRcv

    \Statex
    \WhenRcv[$\unknownTx(\ballotvar, \reqIdVar, \tidvar, \pvar)$]
      \State \pre $\statusVar \in \set{\leaderproc, \follower, \restoring}
        \land \ballotVar = \ballotvar$
      \State $\send \unknowntxAck(\mvar, \reqIdVar, \tidvar) \sendto \pvar$
    \EndWhenRcv

    \Statex
    \Function{\retry}{$\tidvar$}
      \Comment{Run periodically}
      \label{line:function-retry}
      \State \pre $\Call{\certify}{\_, \tidvar, \_, \_, \_, \lccolor{\_}}\ \text{is not executing}
        \land\statusVar = \leaderproc \land \exists
        \langle \tidvar, \wbuffvar, \rsvar, \snapvc, \_, \_, \lccolor{\lcvar} \rangle
        \in \preparedstrong$\!\!
        \label{line:retry-precondition}
      \State $\Call{\certify}{\normalMode, \tidvar, \wbuffvar, \rsvar, \snapvc, \lccolor{\lcvar}}$
    \EndFunction
  \end{algorithmic}
\end{algorithm*}

\begin{algorithm*}[t]
  \caption{Atomic transaction commit protocol at $\pvar^{m}_{d}$: recovery}
  \label{alg:unistore-recovery}
  \begin{algorithmic}[1]
    \Upon[$\Omega_{m} \neq \trustedVar$]
      \State $\trustedVar \gets \Omega_{m}$
      \If{$\trustedVar = \pvar^{m}_{d}$}
        \recover()
      \Else{}
        $\send \nack(\ballotVar) \sendto \trustedVar$
      \EndIf
    \EndUpon

    \Statex
    \WhenRcv[\nack($b$)]
      \State \pre $\trustedVar = \pvar^{m}_{d} \land \ballotvar > \ballotVar$
      \hStatex
      \State $\ballotVar \gets \ballotvar$
      \State \recover()
    \EndWhenRcv

    \Statex
    \Function{\recover}{\null}
      \label{line:function-recover}
      \State \send $\Call{\newleader}{\text{any ballot } \ballotvar
        \text{ such that } \ballotvar > \ballotVar
          \land \leaderof(\ballotvar) = \pvar^{\mvar}_{\dvar}}
          \sendto \replicas(\mvar)$
    \EndFunction

    \Statex
    \WhenRcv[$\newleader(\ballotvar) \from \pvar$]
      \label{line:function-newleader}
      \If{$\trustedVar = \pvar \land \ballotVar < \ballotvar$}
        \label{line:newleader-if}
        \State $\statusVar \gets \recovering$
          \label{line:newleader-status}
        \State $\ballotVar \gets \ballotvar$
          \label{line:newleader-ballot}
        \State $\doNotWaitFor \gets \emptyset$
        \State $\send \Call{\newleaderack}{\ballotVar, \cballot,
          \preparedstrong, \decidedstrong} \sendto \pvar$
          \label{line:newleader-call-newleaderack}
      \Else{}
        $\send \Call{\nack}{\ballotVar} \sendto \pvar$
      \EndIf
    \EndWhenRcv

    \Statex
    \WhenRcv[$\set{\newleaderack(\ballotvar, \cballotvar_{j},
      \preparedstrongvar_{j}, \decidedstrongvar_{j}) \mid p_{j} \in Q}
        \;\text{\bf from a quorum}\; Q$]
        \label{function:newleaderack}
      \State \pre $\statusVar = \recovering \land \ballotVar = \ballotvar$
        \label{line:newleaderack-precondition}
      \hStatex
      \State \var $J \gets$ the set of $j$ with maximal $\cballotvar_{j}$
        \label{line:newleaderack-J}
      \State $\decidedstrong \gets \bigcup\limits_{j \in J} \decidedstrongvar_{j}$
        \label{line:newleaderack-decidedstrong}
      \State $\preparedstrong \gets \set{\langle \tidvar, \_, \_, \_, \_, \_, \lccolor{\_} \rangle
        \in \bigcup\limits_{j \in J} \preparedstrongvar_{j} \mid
          \langle \tidvar, \_, \_, \_, \lccolor{\_} \rangle \notin \decidedstrong}$
        \label{line:newleaderack-preparedstrong}
      \State \var $\maxPrep \gets \max\set{\tsvar \mid
        \langle \_, \_, \_, \_, \_, \tsvar, \lccolor{\_} \rangle \in \preparedstrong}$
        \label{line:newleaderack-maxprep}
      \State \var $\maxDec \gets \max\set{\commitvc[\strongentry] \mid
        \langle \_, \_, \_, \commitvc, \lccolor{\_} \rangle \in \decidedstrong}$
        \label{line:newleaderack-maxdec}
      \State \wait\until $\clockVar \ge \max\set{\maxPrep, \maxDec}$
        \label{line:newleaderack-wait-clock}

      \hStatex
      \State $\cballot \gets \ballotvar$
        \label{line:newleaderack-cballot}
      \State \send $\newstate(\ballotVar, \preparedstrong, \decidedstrong)
        \sendto \replicas(\mvar) \setminus \set{\pvar^{\mvar}_{\dvar}}$
        \label{line:newleaderack-call-newstate}
    \EndWhenRcv

    \Statex
    \WhenRcv[\newstate($\ballotvar, \preparedstrongvar, \decidedstrongvar$) \from $\pvar$]
      \label{line:function-newstate}
      \State \pre $\statusVar = \recovering \land \ballotvar \ge \ballotVar$
        \label{line:newstate-precondition}
      \hStatex
      \State $\cballot \gets \ballotvar$
        \label{line:newstate-cballot}
      \State $\preparedstrong \gets \preparedstrongvar$
        \label{line:newstate-preparedstrong}
      \State $\decidedstrong \gets \decidedstrongvar$
        \label{line:newstate-decidedstrong}
      \State $\statusVar \gets \follower$
        \label{line:newstate-status}
      \State $\send \Call{\newstateack}{\ballotvar} \sendto \pvar$
        \label{line:newstate-call-newstateack}
    \EndWhenRcv

    \Statex
    \WhenRcv[\newstateack($\ballotvar$)] {\bf from a set of processes}
      \textbf{that together with} $\pvar^{\mvar}_{\dvar}$ \textbf{form a quorum}
      \State \pre $\statusVar = \recovering \land \ballotVar = \ballotvar$
      \hStatex
      \State $\statusVar \gets \restoring$
      \ForAll{$t = \langle \tidvar, \wbuffvar, \rsvar, \snapvc, \_, \_, \lccolor{\lcvar} \rangle
        \in \preparedstrong$}
        \If{$\Call{\certify}{\restoring, \tidvar, \wbuffvar, \rsetvar, \snapvc, \lccolor{\lcvar}} = \unknowntx$}
          \State $\doNotWaitFor \gets \doNotWaitFor \cup \set{t}$
        \EndIf
    \EndFor
    \EndWhenRcv

    \Statex
    \Upon[$\preparedstrong \subseteq \doNotWaitFor \land \statusVar = \restoring$]
      \State $\statusVar \gets \leaderproc$
      \State $\doNotWaitFor \gets \emptyset$
    \EndUpon
  \end{algorithmic}
\end{algorithm*}

\clearpage

\section{Consistency Model Specification} \label{section:spec}

\subsection{Relations} \label{ss:relations}

For a binary relation $\relation \subseteq A \times A$
and an element $a \in A$, we define
$\relation^{-1}(a) = \set{b \mid (b, a) \in \relation}$.
For a non-empty set $A$ and a total order $\relation \subseteq A \times A$,
we let $\max\limits_{\relation}(A)$ be the maximum element in $A$
according to $\relation$. Formally,
\[
  \max_{\relation}(A) = a \iff A \neq \emptyset
    \land \forall b \in A.\; a = b \lor (b, a) \in \relation.
\]
If $A$ is empty, then $\max\limits_{\relation}(A)$ is undefined.
We implicitly assume that $\max\limits_{\relation}(A)$ is defined
whenever it is used.

We call a binary relation a \emph{(strict) partial order}
if it is irreflexive and transitive.
We call it a \emph{total order} if it additionally
relates every two distinct elements one way or another.
\subsection{Operations and Events} \label{ss:operations}

Transactions in \unistore{} can start, read and write keys, and commit.
We assume that each transaction is associated with a unique transaction identifier
$\tidvar$ from a set $\tids$ (corresponding to
line~\code{\ref{alg:unistore-coord}}{\ref{line:start-return}}).
Besides, clients can issue on-demand barriers and migrate between data centers.

Let $\Key$ and $\Val$ be the set of keys and values, respectively.
We define $\OP$ as the set of all possible operations
\begin{align*}
  \OP = \;&\set{\start(\tidvar) \mid \tidvar \in \tids} \;\cup \\
    &\set{\commitcausaltx(\tidvar) \mid \tidvar \in \tids} \;\cup \\
    &\set{\commitstrongtx(\tidvar, \decvar) \mid \\
      &\quad \tidvar \in \tids, \decvar \in \set{\commit, \abort}} \;\cup \\
    &\set{\fence} \;\cup \\
    &\set{\clattach(j) \mid j \in \D} \;\cup \\
    &\set{\read(\tidvar, k, v), \updateproc(\tidvar, k, v) \mid \\
      &\quad \tidvar \in \tids, k \in \Key, v \in \Val}.
\end{align*}
We denote each invocation of such an operation by an \emph{event}
from a set $E$, usually ranged over by $e$.
A function $\opfunc : E \to \OP$ determines the operation a given event denotes.
Formally, we use the following notation to denote different types of events.
\begin{itemize}
  \item $E$: The set of all events.
  \item $S$: The set of \start{} events. That is,
    \[
      S = \set{e \in E \mid \exists \tidvar \in \tids.\;
        \opfunc(e) = \start(\tidvar)}.
    \]
  \item $R$: The set of \read{} (read) events. That is,
    \begin{align*}
      R = \set{&e \in E \mid \exists \tidvar \in \tids, k \in \Key, v \in \Val.\; \\
        &\quad \opfunc(e) = \read(\tidvar, k, v)}.
    \end{align*}
  \item $U$: The set of \updateproc{} (update) events. That is,
    \begin{align*}
      U = \set{&e \in E \mid \exists \tidvar \in \tids, k \in \Key, v \in \Val.\; \\
        &\quad \opfunc(e) = \updateproc(\tidvar, k, v)}.
    \end{align*}
  \item $C_{\causalentry}$: The set of \commitcausaltx{} events. That is,
    \begin{align*}
      C_{\causalentry} = \set{&e \in E \mid \exists \tidvar \in \tids.\; \\
        &\quad \opfunc(e) = \commitcausaltx(\tidvar)}.
    \end{align*}
  \item $C_{\strongentry}$: The set of \commitstrongtx{} events
    with decision $\decvar = \commit$. That is,
    \begin{align*}
      C_{\strongentry} &= \set{e \in E \mid \exists \tidvar \in \tids.\; \\
        &\opfunc(e) = \commitstrongtx(\tidvar, \commit)}.
    \end{align*}

  \item $C \triangleq C_{\causalentry} \uplus C_{\strongentry}$:
    The set of all commit events.
  \item $\Fence$: The set of \fence{} events. That is,
    \[
      \Fence = \set{e \in E \mid \opfunc(e) = \fence}.
    \]
  \item $\Attach$: The set of \clattach{} events. That is,
    \[
      \Attach = \set{e \in E \mid \exists j \in \D.\; \opfunc(e) = \clattach(j)}.
    \]
  \item $R_{k}$: The set of read events on key $k$. That is,
    \begin{align*}
      R_{k} = \set{&e \in E \mid \exists \tidvar \in \tids, v \in Val.\; \\
        &\quad \opfunc(e) = \read(\tidvar, k, v)}.
    \end{align*}
  \item $U_{k}$: The set of update events on key $k$. That is,
    \begin{align*}
      U_{k} = \set{&e \in E \mid \exists \tidvar \in \tids, v \in Val.\; \\
        &\quad \opfunc(e) = \updateproc(\tidvar, k, v)}.
    \end{align*}
\end{itemize}

For different types of events, we define
\begin{itemize}
  \item $\key(e)$: The key that the read or update event
    $e \in R \cup U$ accesses.
  \item $\rval(e)$: The return value of the read event $e \in R$.
  \item $\uval(e)$: The value written by the update event $e \in U$.
\end{itemize}
\subsection{Transactions} \label{ss:transactions}

\begin{appdefinition}[Transactions] \label{def:tx}
  A transaction $\tvar$ is a triple $(\tidvar, \txnevents, \po)$, where
  \begin{itemize}
    \item $\tidvar \in \tids$ is a unique transaction identifier;
    \item $\txnevents \subseteq E \setminus (\Fence \cup \Attach)$
      is a finite, non-empty set of events;
    \item $\po \subseteq \txnevents \times \txnevents$
      is the program order, which is total.
  \end{itemize}
  We only consider well-formed transactions: according to the $\po$ order,
  $\tvar$ starts with a \start{} event, then performs some number of
  \read{}/\updateproc{} events, and ends with a commit event (\commitcausaltx{}
  or \commitstrongtx).
\end{appdefinition}

In the following, we denote components of $t$ as in $t.\tidvar$.
For simplicity, we assume a dedicated \emph{initial} transaction $\tvar_{0}$
which installs initial values to all possible keys before the system launches.

We use the following notations to denote different types of transactions.
\begin{itemize}
  \item $\txs$: The set of all committed transactions.
  \item $\txs_{k}$: The set of committed transactions that update key $k$.
    We also use $C_{k}$ to denote the set of commit events of transactions
    in $\txs_{k}$.
  \item $\causaltxs$: The set of transactions that
    end with the \commitcausaltx{} events.
    We call them causal transactions.
    Causal transactions will always be committed.
  \item $\allstrongtxs$: The set of transactions that
    end with the \commitstrongtx{} events.
    We call them strong transactions.
    Strong transactions can be committed or aborted.
  \item $\strongtxs$: The set of \emph{committed} strong transactions.
\end{itemize}

We have $\txs = \causaltxs \uplus \strongtxs$.
For each transaction $\tvar \in \txs$, we define
\begin{itemize}
  \item $\tidselector(\tvar) \in \tids$:
    The transaction identifier $t.\tidvar$ of $\tvar$.
  \item $\events(\tvar) \subseteq S \cup R \cup U \cup C$:
    The set $t.\txnevents$ of events in $\tvar$.
  \item $\ws(\tvar) \subseteq \Key \times \Val$:
    The write set of $\tvar$.
    It is the set of keys with their values that $\tvar$ updates,
    which contains at most one value per key.
    Formally,
    \[
      \ws(\tvar) \triangleq \set{(\key(e), \uval(e)) \mid e \in t.\txnevents \cap U}.
    \]
  \item $\rs(\tvar) \subseteq \Key$: The read set of $\tvar$.
    It is the set of keys that $\tvar$ reads.
    Formally,
    \[
      \rs(\tvar) \triangleq \set{\key(e) \mid e \in t.\txnevents \cap R}.
    \]
  \item $\startoftx(\tvar) \in S$:
    The \start{} event of $\tvar$.
    Formally, it is the unique event in the set $t.\txnevents \cap S$.
  \item $\commitoftx(\tvar) \in C$: The commit event of $\tvar$.
    Formally, it is the unique event in the set $t.\txnevents \cap C$.
  \item $\ud(\tvar, k) \in U_{k}$: The \emph{last} update event on key $k$, if any,
    in transaction $\tvar$. Formally,
    \[
      \ud(\tvar, k) \triangleq \max\limits_{\po}(\events(\tvar) \cap U_{k}).
    \]
\end{itemize}

Besides, we define
\begin{align}
  \W(\tvar) &\triangleq \set{k \in \Key \mid \langle k, \_ \rangle \in \ws(\tvar)}, \label{eq:W-def}\\
  \R(\tvar) &\triangleq \rs(\tvar) \cup \W(\tvar).\label{eq:R-def}
\end{align}

For a read event $e$ on key $k$ in transaction $\tvar$,
if there exist update events on $k$ preceding $e$ in $\tvar$,
then $e$ is called an \emph{internal} read event.
Otherwise, $e$ is called an \emph{external} read event.
We denote the sets of internal reads and external reads
by $\intread$ and $\extread$, respectively.
That is, $R = \intread \uplus \extread$.

We also distinguish commit events
for read-only transactions from those for update transactions,
and denote their sets by $\rocommit$ and $\updatecommit$, respectively.
That is, $C = \rocommit \uplus \updatecommit$.

For notational convenience,
for an event $e \in E \setminus (\Fence \cup \Attach)$, we also define
$\txfunc(e)$ to be the transaction containing $e$ and
\begin{align*}
  \startoftx(e) &\triangleq \startoftx(\txfunc(e)), \\
  \commitoftx(e) &\triangleq \commitoftx(\txfunc(e)).
\end{align*}
\subsection{Abstract Executions} \label{ss:cm}

Clients interacts with \unistore{} by issuing transactions
and \fence{} and \attach{} events.
We use histories to record such interactions in a single computation.
Note that histories only record committed transactions.
\begin{appdefinition}[Histories] \label{def:histories}
  A \emph{history} is a tuple
  \[
    H = (X, \client, \dc, \so)
  \]
  such that
  \begin{itemize}
    \item $X \subseteq \txs \cup \Fence \cup \Attach$
      is a set of committed transactions and \fence{} and \attach{} events;
    \item $\client: X \to \clients$ is a function that returns
      \begin{itemize}
        \item the client $\client(t)$ which issues the transaction
          $\tvar \in (X \cap \txs)$,
        \item the client $\client(\fencerange)$ which issues the \fence{} event
          $\fencerange \in (X \cap \Fence)$, or
        \item the client $\client(\attachrange)$ which issues the \attach{} event
          $\attachrange \in (X \cap \Attach)$;
      \end{itemize}
    \item $\dc: X \to \D$ is a function
      that returns the original data center $\dc(\tvar)$
      of transaction $\tvar \in (X \cap \txs)$,
      $\dc(\fencerange)$ of \fence{} event $\fencerange \in (X \cap \Fence)$,
      or $\dc(\attachrange)$ of \attach{} event $\attachrange \in (X \cap \Attach)$;
    \item $\so \subseteq X \times X$ is the \emph{session order} on $X$.
      Consider $\txevents_{1}, \txevents_{2} \in X$.
      We say that $\txevents_{1}$ precedes $\txevents_{2}$
      in the session order, denoted $\txevents_{1} \rel{\so} \txevents_{2}$,
      if they are executed by the same client
      and $\txevents_{1}$ is executed before $\txevents_{2}$.
  \end{itemize}
\end{appdefinition}
In the following, we denote components of $H$ as in $H.X$
and often shorten $H.X$ by $X$ when it is clear.
Let $V_{H} \triangleq \bigcup (H.X \cap T).Y$
be the set of transactional events in history $H$.

A consistency model is specified by a set of histories.
To define this set, we extend histories with two relations,
declaratively describing how the system processes transactions
and \fence{} events.

\begin{appdefinition}[Abstract Executions] \label{def:ae}
  An \emph{abstract execution} is a triple
  \[
    A = ((X, \client, \dc, \so), \vis, \ar)
  \]
  such that
  \begin{itemize}
    \item $(X, \client, \dc, \so)$ is a history;
    \item Visibility $\vis \subseteq X \times X$ is a partial order;
    \item Arbitration $\ar \subseteq X \times X$ is a total order.
  \end{itemize}
\end{appdefinition}
For $H = (X, \client, \dc, \so)$,
we often shorten $((X, \client, \dc, \so), \vis, \ar)$ by $(H, \vis, \ar)$.
\subsection{Partial Order-Restrictions Consistency} \label{ss:por}

We aim to show that \unistore{} implements a transactional variant
of \emph{Partial Order-Restrictions consistency
(\por{} consistency)}~\cite{red-blue, por} for LWW registers.
A history $H$ of \unistore{} satisfies \por, denoted $H \models \por$,
if it can be extended to an abstract execution that satisfies several axioms,
defined in the following:
\begin{align*}
  H \models \por \iff\; &\exists \vis, \ar.\; (H, \vis, \ar) \models \\
                 &\qquad \retval \;\land \\
                 &\qquad \cc \;\land \\
                 &\qquad \conflictaxiom \;\land \\
                 &\qquad \ev.
\end{align*}
$\unistore$ satisfies $\por$, denoted $\unistore \models \por$,
if all its histories do.

Given an abstract execution $A = (H, \vis, \ar)$,
the axioms are defined as follows.
\begin{appdefinition}[$\retval$, \cite{framework-concur15}] \label{def:retval}
  The Return Value Consistency (\retval) specifies
  the return value of each read event.
  \[
    \retval \triangleq \intretval \land \extretval.
  \]
  Here $\intretval$ requires an internal read event $e$
  on key $k$ to read from the last update event on $k$ preceding $e$
  in the same transaction. Formally,
  \begin{align*}
    &\intretval \triangleq \forall e \in \intread \cap R_{k} \cap V_{H}. \\
        &\quad \rval(e) = \uval\big(\max_{\po}(\po^{-1}(e) \cap U_{k})\big).
  \end{align*}
  $\extretval$ requires an external read event $e$
  on key $k$ to read from the last update event on $k$
  in the last transaction preceding $\txfunc(e)$ in $\ar$,
  among the set of transactions visible to $\txfunc(e)$.
  Formally,
  \begin{align*}
    &\extretval \triangleq \forall e \in \extread \cap R_{k} \cap V_{H}. \\
        &\quad \rval(e) =
          \uval\Big(
                \ud\big(\max_{\ar}\big(\vis^{-1}(\txfunc(e)) \cap \txs_{k}\big), k\big)
               \Big).
  \end{align*}
\end{appdefinition}

\begin{appdefinition}[\cc, \cite{sebastian-book}] \label{def:cc}
  \begin{align*}
    \cc \triangleq\; &\cv \;\land \\
                   &\ca,
  \end{align*}
  where
  \[
    \cv \triangleq (\so\; \cup \vis)^{+} \subseteq \vis;
  \]
  \[
    \ca \triangleq \vis \subseteq \ar.
  \]
\end{appdefinition}

The Conflict Ordering property requires that
out of any two conflicting strong transactions,
one must be visible to the other.
Formally,
\begin{appdefinition}[Conflict Relation] \label{def:conflict-relation-tx}
  The conflict relation, denoted by $\conflict$, between strong transactions
  is a symmetric relation defined as follows:
  \begin{align*}
    &\forall \tvar, \tvar' \in \strongtxs.\; \tvar \conflict \tvar' \iff \\
      &\quad (\R(\tvar) \cap \W(\tvar') \neq \emptyset) \lor
             (\W(\tvar) \cap \R(\tvar') \neq \emptyset).
  \end{align*}
\end{appdefinition}

\begin{appdefinition}[\conflictaxiom] \label{def:conflictaxiom}
  \begin{align*}
    &\conflictaxiom \triangleq
      \forall \tvar_{1}, \tvar_{2} \in X \cap \strongtxs.\; \\
        &\quad \tvar_{1} \conflict \tvar_{2} \implies
          \tvar_{1} \rel{\vis} \tvar_{2} \lor \tvar_{2} \rel{\vis} \tvar_{1}.
  \end{align*}
\end{appdefinition}
The Eventual Visibility property requires that
a transaction that originates at a correct data center,
that is visible to some \fence{} events,
or that is a strong transaction
eventually becomes visible at all correct data centers.
Let $\C \subseteq \D$ be the set of correct data centers.
Formally,
\begin{appdefinition}[\ev] \label{def:ev}
  \begin{align*}
    &\ev \triangleq \forall \tvar \in X \cap \txs.\; \\
      &\quad \dc(\tvar) \in \C \lor
        (\exists \fencerange \in \Fence.\; \tvar \rel{\vis} \fencerange) \lor \tvar \in \strongtxs \\
        &\qquad \implies \left\lvert \left\{\tvar' \in \txs
          \mid \lnot (\tvar \rel{\vis} \tvar') \right\} \right\rvert < \infty.
  \end{align*}
\end{appdefinition}

\section{Transaction Certification Service Specification} \label{section:tcs}

\subsection{Interface} \label{ss:interface}

The \emph{Transaction Certification Service} (TCS)~\cite{discpaper}
is responsible for certifying strong transactions issued by transaction coordinators,
computing commit vectors and Lamport clocks for committed transactions,
and (asynchronously) delivering committed transactions to replicas.

Each strong transaction $\tvar \in \allstrongtxs$ submitted to TCS
may be associated with its read set $\rs(\tvar)$, write set $\ws(\tvar)$,
snapshot vector $\snapshotVC(\tvar)$ (Definition~\ref{def:snapshotvc}),
commit vector $\commitVC(\tvar)$ (Definition~\ref{def:commitvc}),
and Lamport clock $\lclock(\tvar)$ (Definition~\ref{def:lc-tx}).
From $\rs(\tvar)$ and $\ws(\tvar)$ we can then define 
$\W(t)$ and $\R(t)$ according to~(\ref{eq:W-def}) and~(\ref{eq:R-def}), respectively.
Note that we have $\W(t) \subseteq \R(t)$.

Transaction coordinators for strong transactions interact with TCS
using two types of \emph{actions}. Coordinators can make certification requests
(corresponding to procedure \certify{} of Algorithm~\ref{alg:unistore-certify})
of the form
\[
  \intcertify(\tidselector(\tvar), \ws(\tvar), \R(\tvar),
    \snapshotVC(\tvar), \lccertify(\tvar)),
\]
where $\tvar \in \allstrongtxs$ and $\lccertify(\tvar) \in \N$
denotes the contribution of $\client(\tvar)$ to the Lamport clock of $\tvar$.
The TCS responses are of the form
\[
  \intdecide(\tidselector(\tvar), \decvar, \vcvar, \lcvar),
\]
containing a decision $\decvar$ from $\Decision = \set{\commit, \abort}$ for $\tvar$,
a commit vector $\vcvar$ from $\Vector$ for $\tvar$ if $\decvar = \commit$,
and a Lamport clock $\lcvar$ from $\N$ for $\tvar$ if $\decvar = \commit$.
If $\decvar = \abort$, then $\vcvar$ and $\lcvar$ are irrelevant.

Besides, TCS can deliver some payload $\payload$
to a replica via \upcall \emph{actions} $\intdeliver(\payload)$
(corresponding to procedure \deliverupdates{}
of Algorithm~\ref{alg:unistore-strong-commit}).
We denote by $\intdeliver^{m}_{d}(\payload)$ the delivery of the payload $\payload$ to
a replica $p^{m}_{d}$, when the latter is relevant.
The payload $\payload$ in $\intdeliver^{m}_{d}(\payload)$ is a set of tuples of the form
$\langle \tidvar, \wbuffvar, \commitvc, \lcvar \rangle$,
each of which corresponds to the updates $\wbuffvar \subseteq \Key \times \Val$
performed at a particular partition $m$
by a particular committed strong transaction
with transaction identifier $\tidvar$, commit vector $\commitvc$,
and Lamport clock $\lcvar$.
\subsection{Certification Functions} \label{ss:cert-func}

TCS is specified using a certification function
\begin{align}
  \certfunc: 2^{\strongtxs} \times \allstrongtxs \to
    \Decision \times \Vector \times \N.
  \label{eqn:f}
\end{align}
For a strong transaction $\tvar \in \allstrongtxs$
and the set $\txsincertify \subseteq \strongtxs$ of previously committed strong transactions,
$\certfunc(\txsincertify, \tvar)$ returns not only the decision $\decvar \in \Decision$,
but also the commit vector $\vcvar \in \Vector$
and Lamport clock $\lcvar \in \N$ for $\tvar$.
We use $\fdec(\txsincertify, \tvar)$, $\fvec(\txsincertify, \tvar)$,
and $\flc(\txsincertify, \tvar)$ to select
the first, second, and third component
of $\certfunc(\txsincertify, \tvar)$, respectively.

The decision $\fdec(\txsincertify, \tvar)$ should satisfy
\begin{align}
  &\fdec(\txsincertify, \tvar) = \commit
    \iff \forall k \in \R(\tvar).\; \forall \tvar' \in \txsincertify.
    \nonumber\\
    &\quad \big(k \in \W(\tvar') \implies
      \commitVC(\tvar') \le \snapshotVC(\tvar)\big).
    \label{eqn:gcf-decision}
\end{align}

The commit vector $\fvec(\txsincertify, \tvar)$ should satisfy
\begin{align}
  &\phantom{\land\;} (\forall i \in \D.\; \fvec(\txsincertify, \tvar)[i] = \snapshotVC(\tvar)[i]) \nonumber\\
  &\land \fvec(\txsincertify, \tvar)[\strongentry] > \snapshotVC(\tvar)[\strongentry] \nonumber\\
  &\land \forall \tvar' \in \txsincertify.\; \tvar \conflict \tvar'
    \implies \fvec(\txsincertify, \tvar) \ge \commitVC(\tvar').
    \label{eqn:gcf-commitvc}
\end{align}

The Lamport clock $\flc(\txsincertify, \tvar)$ should satisfy
\begin{align}
  &\flc(\txsincertify, \tvar) \ge \lccertify(t) \;\land \nonumber\\
  &\quad \big(\forall \tvar' \in \txsincertify.\; \tvar \conflict \tvar' \implies
    \flc(\txsincertify, \tvar) > \lclock(\tvar')\big).
    \label{eqn:gcf-lc}
\end{align}
\subsection{Histories of TCS} \label{ss:histories}

TCS executions are represented by \emph{histories},
which are (possibly infinite) sequences
of $\intcertify$, $\intdecide$, and $\intdeliver$ actions.
For a TCS history $h$, we use $\act(h)$ to denote the set of actions in $h$.
For actions $\actvar, \actvar' \in \act(h)$,
we write $\actvar \prec_{h} \actvar'$
when $\actvar$ occurs before $\actvar'$ in $h$.
A strong transaction $\tvar \in \allstrongtxs$ \emph{commits} in a history $h$
if $h$ contains $\intdecide(\tidselector(\tvar), \commit, \_, \_)$.
We denote by $\committedVar(h)$ the projection of $h$ to actions
corresponding to the strong transactions that are committed in $h$.

Each history $h$ needs to meet the following requirements.
\begin{enumerate}[(R1)]
  \item \label{tcs-requirement:certify-once}
    For each strong transaction $\tvar \in \allstrongtxs$,
    there is at most one $\intcertify(\tidselector(\tvar), \_, \_, \_, \_)$
    action in $h$.
  \item \label{tcs-requirement:decide-once}
    For each action $\intdecide(\tidvar, \_, \_, \_) \in \act(h)$,
    there is exactly one $\intcertify(\tidvar, \_, \_, \_, \_)$ action in $h$
    such that
    \[
      \intcertify(\tidvar, \_, \_, \_, \_)
        \prec_{h} \intdecide(\tidvar, \_, \_, \_).
    \]
  \item \label{tcs-requirement:abort-cannot-deliver}
    For each action $\intdeliver(\payload) \in \act(h)$
    and each $\langle \tidvar, \_, \_, \_ \rangle \in \payload$,
    there is \emph{no} $\intdecide(\tidvar, \abort, \_, \_)$ action in $h$.
  \item \label{tcs-requirement:deliver-once}
    Every committed strong transaction is delivered at most once to each replica.
  \item \label{tcs-requirement:certify-before-deliver}
    For each action $\intdeliver(\payload) \in \act(h)$
    and each $\langle \tidvar, \_, \_, \_ \rangle \in \payload$,
    there is a $\intcertify(\tidvar, \_, \_, \_, \_)$ action such that
    \[
      \intcertify(\tidvar, \_, \_, \_, \_) \prec_{h} \intdeliver(\payload).
    \]
  \item \label{tcs-requirement:deliver-order}
    At each replica $p^{m}_{d}$, committed strong transactions are delivered
    in the increasing order of their strong timestamps.
    Formally, for any two distinct actions
    $\intdeliver^{m}_{d}(\payload_{1})$ and $\intdeliver^{m}_{d}(\payload_{2})$
    with payloads $\payload_{1}$ and $\payload_{2}$, respectively,
    \begin{align*}
      &\intdeliver^{m}_{d}(\payload_{1}) \prec_{h} \intdeliver^{m}_{d}(\payload_{2}) \implies \\
        &\quad \forall \langle \_, \_, \commitvc_{1}, \_ \rangle \in \payload_{1}.\;\\
        &\quad \forall \langle \_, \_, \commitvc_{2}, \_ \rangle \in \payload_{2}.\; \\
          &\qquad \commitvc_{1}[\strongentry] < \commitvc_{2}[\strongentry].
    \end{align*}
\end{enumerate}

A history is \emph{complete} if every $\intcertify$ action
in it has a matching $\intdecide$ action.
A complete history $h$ is \emph{sequential} if
it consists of consecutive pairs of $\intcertify$ and matching $\intdecide$ actions.
For a complete history $h$, a \emph{permutation} $h'$ of $h$
is a sequential history such that
\begin{itemize}
  \item $h$ and $h'$ contain the same actions, i.e., $\act(h) = \act(h')$.
  \item Transactions are certified in $h'$ according to their session order.
    \begin{align*}
      &\forall \tvar, \tvar' \in \allstrongtxs.\;
        \tvar \rel{\so} \tvar' \implies \\
          &\quad \intdecide(\tidselector(\tvar), \_, \_, \_) \prec_{h'}
            \intcertify(\tidselector(\tvar'), \_, \_, \_, \_).
    \end{align*}
\end{itemize}
\subsection{TCS Correctness: Safety and Liveness} \label{ss:tcs-correctness}

\subsubsection{Safety of TCS} \label{ss:tcs-safety}

A complete sequential history $h$ is \emph{legal} with respect to
a certification function $\certfunc$,
if its results are computed so as to satisfy
(\ref{eqn:gcf-decision}) -- (\ref{eqn:gcf-lc}) according to $\certfunc$:
\begin{align*}
  &\forall \actvar = \intdecide(\tidselector(\tvar), \decvar, \vcvar,
    \lcvar) \in \act(h). \\
    &\quad (\decvar, \vcvar, \lcvar)
      = \certfunc(\set{\tvar' \mid \\
        &\quad\quad \intdecide(\tidselector(\tvar'), \commit, \_, \_) \prec_{h} \actvar}, \tvar).
\end{align*}
A history $h$ is \emph{correct} with respect to $\certfunc$
if $h \mid \committedVar(h)$ has a legal permutation.
A TCS implementation is \emph{correct} with respect to $\certfunc$
if so are all its histories.
\subsubsection{Liveness of TCS} \label{ss:tcs-liveness}

TCS guarantees that every committed strong transaction
will eventually be delivered by every correct data center.
Formally,
\begin{align}
  &\forall \actvar = \intdecide(\tidvar, \commit, \_, \_) \in \act(h).
    \nonumber\\
    &\quad \forall m \in \partitionsfunc(\tidvar).\; \forall \cdrange \in \C.
    \nonumber\\
    &\qquad \exists \actvar' = \intdeliver^{m}_{\cdrange}(\payload) \in \act(h).\;
    \nonumber\\
      &\quad\qquad \langle \tidvar, \_, \_, \_ \rangle \in \payload
        \land \actvar \prec_{h} \actvar'.
    \label{eqn:tcs-liveness-delivery}
\end{align}
Here $\partitionsfunc(\tidvar)$ denotes the set of partitions
that a particular transaction with transaction identifier $\tidvar$ accesses.

A TCS implementation meets the liveness requirement if
every history produced by its maximal execution satisfies (\ref{eqn:tcs-liveness-delivery}).
\subsection{TCS Correctness} \label{ss:tcs-correctness-unistore}

The proof of TCS correctness is an adjustment
of the ones in~\cite{discpaper, multicast-dsn19}.

\begin{apptheorem} \label{thm:tcs-correctness}
  The TCS implementation in \unistore{}
  (Algorithms~\ref{alg:unistore-certify} -- \ref{alg:unistore-recovery})
  is correct with respect to the certification function $\certfunc$ in (\ref{eqn:f})
  and meets the liveness requirement in (\ref{eqn:tcs-liveness-delivery}).
\end{apptheorem}

\section{The Proof of \unistore{} Correctness} \label{section:correctness-proof}

Consider an execution of \unistore{} with a history $H = (X, \client, \dc, \so)$.
We prove that $H$ satisfies \por{}
by constructing an abstract execution $A$ (Theorem~\ref{thm:unistore-por}).
We also establish the liveness guarantees of \unistore{} (Theorem~\ref{thm:termination}).

\subsection{Assumptions} \label{ss:assumptions}

We take the following assumptions about \unistore.

\begin{appassumption} \label{assumption:clock}
  For any replica $p^{m}_{d}$ in data center $d$,
  $\clockVar$ at $p^{m}_{d}$ is strictly increasing until $d$ (may) crash.
\end{appassumption}

\begin{appassumption} \label{assumption:message}
  Replicas are connected by reliable FIFO channels:
  messages are delivered in FIFO order,
  and messages between correct data centers
  are guaranteed to be eventually delivered.
\end{appassumption}

\begin{appassumption} \label{assumption:failure-model}
  We assume that in an execution of \unistore, any clients
  and up to $f$ data centers may crash and that $D > 2f$.
\end{appassumption}

\begin{appassumption} \label{assumption:fairness}
  We assume fairness of procedures of \unistore:
  In an execution, if a procedure is enabled infinitely often,
  then it will be executed infinitely often.
\end{appassumption}

\begin{appassumption} \label{assumption:client-well-formed}
  We consider only \emph{well-formed} executions,
  in which for each client:
  \begin{itemize}
    \item transactions are issued in sequence; and
    \item both \fence{} and \clattach{} events
      can be issued only outside of transactions.
  \end{itemize}
\end{appassumption}

\begin{appassumption} \label{assumption:complete-execution}
  We consider only executions where
  every causal commit event (i.e., \commitcausaltx) completes
  and every strong commit event (i.e., \commitstrongtx) that calls the TCS completes.
\end{appassumption}

We make the last assumption to simplify the technical development. The other
assumptions come from the system model.

\subsection{Notations} \label{ss:proof-notations}

We use $\cl$ to range over clients from a finite set $\clients$.
We also use the following notations
to refer to different types of variables and their values
(below are some typical examples).
\begin{itemize}
  \item $\snapVC^{m}_{d}$:
    The variable $\snapVC$ at replica $p^{m}_{d}$.
  \item $(\snapVC^{m}_{d})_{e}$:
    The value of variable $\snapVC^{m}_{d}$
    after the event $e$ is performed at replica $p^{m}_{d}$.
  \item $\snapVC^{m}_{d}(\realtime)$:
    The value of $\snapVC^{m}_{d}$ at some specific time $\realtime$.
  \item $\pastVC_{\cl}$: The variable $\pastVC$ at client $\cl$.
  \item $(\pastVC_{\cl})_{e}$:
    The value of variable $\pastVC_{\cl}$
    after the event $e$ is performed at client $\cl$.
  \item $\snapvc_{(\readkey, e)}$:
    The actual value of \emph{parameter} $\snapvc$ of handler \readkey{}
    for event $e$.
  \item $\commitvc_{(\commitcausal, e)}$:
    The value of the \emph{local variable} $\commitvc$
    in procedure \commitcausal{}
    after event $e$ is performed.
\end{itemize}
Besides, we use $\coord(t)$ to denote the coordinator partition
of transaction $\tvar$.

Each transaction is associated with a snapshot vector and a commit vector.

\begin{appdefinition}[Snapshot Vector] \label{def:snapshotvc}
  Let $\tvar \in \txs$ be a transaction.
  Let $d \triangleq \dc(\tvar)$ and $m \triangleq \coord(\tvar)$.
  We define its snapshot vector $\snapshotVC(\tvar)$ as
  \[
    \snapshotVC(\tvar) \triangleq (\snapVC^{m}_{d})_{\startoftx(\tvar)}[\tvar].
  \]
\end{appdefinition}

\begin{appdefinition}[Commit Vector] \label{def:commitvc}
  Let $\tvar \in \txs$ be a transaction.
  Let $d \triangleq \dc(\tvar)$ and $m \triangleq \coord(\tvar)$.
  We define its commit vector $\commitVC(\tvar)$ as follows.
  \begin{itemize}
    \item If $\tvar$ is a read-only causal transaction, then
      \[
        \commitVC(\tvar) \triangleq (\snapVC^{m}_{d})_{\commitoftx(\tvar)}[\tvar].
      \]
    \item If $\tvar$ is an update causal transaction, then
      \[
        \commitVC(\tvar) \triangleq \commitvc_{(\commitcausal, \commitoftx(\tvar))}.
      \]
    \item If $\tvar$ is a committed strong transaction, then
      \[
        \commitVC(\tvar) \triangleq vc_{(\commitstrong, \commitoftx(\tvar))}.
      \]
  \end{itemize}
\end{appdefinition}

\begin{applemma} \label{lemma:snapshotvc-commitvc}
  \[
    \forall \tvar \in \txs.\; \commitVC(\tvar) \ge \snapshotVC(\tvar).
  \]
\end{applemma}

\begin{proof} \label{proof:snapshotvc-commitvc}
  We perform a case analysis according to the type of $\tvar$.
  \begin{itemize}
    \item $\textsc{Case I}$: $\tvar$ is a read-only causal transaction.
      By Definition~\ref{def:snapshotvc} of $\snapshotVC(\tvar)$,
      Definition~\ref{def:commitvc} of $\commitVC(\tvar)$,
      and Assumption~\ref{assumption:client-well-formed},
      \[
        \commitVC(\tvar) = \snapshotVC(\tvar).
      \]
    \item $\textsc{Case II}$: $\tvar$ is an update causal transaction.
      By lines~\code{\ref{alg:unistore-coord}}{\ref{line:commitcausal-commitvc}}
      and \code{\ref{alg:unistore-coord}}{\ref{line:commitcausal-commitvc-d}},
      \[
        \commitVC(\tvar) \ge \snapshotVC(\tvar).
      \]
    \item $\textsc{Case III}$: $\tvar$ is a strong transaction.
      By line~\code{\ref{alg:unistore-strong-commit}}{\ref{line:commitstrong-call-certify}}
      and (\ref{eqn:gcf-commitvc}),
      \[
        \commitVC(\tvar) \ge \snapshotVC(\tvar).
      \]
  \end{itemize}
\end{proof}

For client $\cl$, we use $\cldc(\cl)$ to denote the data center
to which $\cl$ is currently attached.
We also use $\txs|_{\cl}$ to denote the set of transactions issued by $\cl$.
Formally,
\[
  \txs|_{\cl} \triangleq \set{\tvar \in \txs \mid \client(\tvar) = cl}.
\]

For a transaction $\tvar$ and a partition $m$,
we use $\ws(\tvar)[m]$ to denote the subset of $\ws(\tvar)$
restricted to partition $m$.
Formally,
\[
  \ws(\tvar)[m] \triangleq \set{\langle k, v \rangle \in \ws(\tvar)
  \mid \partitionofproc(k) = m}.
\]

For notational convenience, we also define
\begin{align*}
  \log(\tvar) \triangleq \set{\langle k, v, \commitVC(\tvar), \lclock(\tvar)
    \rangle \mid \langle k, v \rangle \in \ws(\tvar)},
\end{align*}
and
\begin{align*}
  \log(\tvar)[m] \triangleq \set{\langle k, v, \commitVC(\tvar), \lclock(\tvar)
    \rangle \mid \langle k, v \rangle \in \ws(\tvar)[m]}.
\end{align*}
For a key $k \in \Key$ and a transaction $\tvar \in \txs_{k}$,
let $\log(\tvar)[k]$ be the unique tuple
\[
  \langle k, v, \commitVC(\tvar), \lclock(\tvar) \rangle
\]
in $\log(\tvar)$.

\subsection{Metadata for Causal Transactions}
\label{ss:metadata-causal}

A causal transaction is \emph{committed} when \commitcausal{} for it returns.
A causal transaction is \emph{committed at replica} $p^{m}_{d}$
when \commit{} for it at $p^{m}_{d}$ returns.
\subsubsection{Properties of $\knownVC$}
\label{sss:knownvc}

\begin{applemma} \label{lemma:knownvc-d-nondecreasing}
  For any replica $p^{m}_{d}$ in data center $d$,
  $\knownVC^{m}_{d}[d]$ is non-decreasing.
\end{applemma}

\begin{proof} \label{proof:knownvc-d-nondecreasing}
  Consider two points of time $\realtime_{1}$ and $\realtime_{2}$ such that $\realtime_{1} < \realtime_{2}$.
  We need to show that
  \[
    \knownVC^{m}_{d}(\realtime_{1})[d] \le \knownVC^{m}_{d}(\realtime_{2})[d].
  \]

  Note that $\knownVC^{m}_{d}[d]$ is updated only
  at lines~\code{\ref{alg:unistore-replication}}{\ref{line:propagate-knownvc-clock}}
  or \code{\ref{alg:unistore-replication}}{\ref{line:propagate-knownvc-ts}}.
  We distinguish between the following four cases.
  \begin{itemize}
    \item $\textsc{Case I}$:
      Both $\knownVC^{m}_{d}(\realtime_{1})[d]$ and $\knownVC^{m}_{d}(\realtime_{2})[d]$ are set
      at line~\code{\ref{alg:unistore-replication}}{\ref{line:propagate-knownvc-clock}}.
      By line~\code{\ref{alg:unistore-replication}}{\ref{line:propagate-knownvc-clock}}
      and Assumption~\ref{assumption:clock},
      \begin{align*}
        \knownVC^{m}_{d}(\realtime_{1})[d] &= \clockVar^{m}_{d}(\realtime_{1}) \\
          &< \clockVar^{m}_{d}(\realtime_{2}) \\
          &= \knownVC^{m}_{d}(\realtime_{2})[d].
      \end{align*}
    \item $\textsc{Case II}$:
      $\knownVC^{m}_{d}(\realtime_{1})[d]$ is set
      at line~\code{\ref{alg:unistore-replication}}{\ref{line:propagate-knownvc-clock}}
      and $\knownVC^{m}_{d}(\realtime_{2})[d]$ is set
      at line~\code{\ref{alg:unistore-replication}}{\ref{line:propagate-knownvc-ts}}.
      By line~\code{\ref{alg:unistore-replication}}{\ref{line:propagate-knownvc-clock}},
      \[
        \knownVC^{m}_{d}(\realtime_{1})[d] = \clockVar^{m}_{d}(\realtime_{1}).
      \]
      By the fact that $\preparedcausal^{m}_{d}(\realtime_{1}) = \emptyset$,
      $\realtime_{2} > \realtime_{1}$,
      and line~\code{\ref{alg:unistore-replica}}{\ref{line:preparecausal-ts}},
      \begin{align*}
        \forall \langle \_, \_, &\tsvar \rangle \in\; \preparedcausal^{m}_{d}(\realtime_{2}).\; \\
          &\tsvar > \clockVar^{m}_{d}(\realtime_{1}) = \knownVC^{m}_{d}(\realtime_{1})[d].
      \end{align*}
      Therefore, by line~\code{\ref{alg:unistore-replication}}{\ref{line:propagate-knownvc-ts}},
      \begin{align*}
        &\knownVC^{m}_{d}(\realtime_{1})[d] \\
        &\quad \le \min\set{\tsvar \mid \langle \_, \_, \tsvar \rangle \in \preparedcausal^{m}_{d}(\realtime_{2})} - 1 \\
        &\quad = \knownVC^{m}_{d}(\realtime_{2})[d].
      \end{align*}
    \item $\textsc{Case III}$:
      $\knownVC^{m}_{d}(\realtime_{1})[d]$ is set
      at line~\code{\ref{alg:unistore-replication}}{\ref{line:propagate-knownvc-ts}}
      and $\knownVC^{m}_{d}(\realtime_{2})[d]$ is set
      at line~\code{\ref{alg:unistore-replication}}{\ref{line:propagate-knownvc-clock}}.
      Let $\tvar_{1}$ be the transaction in $\preparedcausal^{m}_{d}(\realtime_{1})$
      that has the minimum $\tsvar$. Formally,
      \begin{align*}
        \tvar_{1} \triangleq \argmin\limits_{\tvar}\set{\tsvar \mid \langle \tidselector(\tvar), \_, \tsvar \rangle
          \in \preparedcausal^{m}_{d}(\realtime_{1})}.
      \end{align*}
      By lines~\code{\ref{alg:unistore-replication}}{\ref{line:propagate-knownvc-ts}},
      \code{\ref{alg:unistore-coord}}{\ref{line:commitcausal-commitvc-d}},
      \code{\ref{alg:unistore-replica}}{\ref{line:commit-wait-clock}},
      and \code{\ref{alg:unistore-replication}}{\ref{line:propagate-knownvc-clock}},
      \begin{align*}
        \knownVC^{m}_{d}(\realtime_{1})[d] &< \commitVC(\tvar_{1})[d] \\
          &\le \clockVar^{m}_{d}(\realtime_{2}) \\
          &= \knownVC^{m}_{d}(\realtime_{2})[d].
      \end{align*}
    \item $\textsc{Case IV}$:
      Both $\knownVC^{m}_{d}(\realtime_{1})[d]$ and $\knownVC^{m}_{d}(\realtime_{2})[d]$ are set
      at line~\code{\ref{alg:unistore-replication}}{\ref{line:propagate-knownvc-ts}}.
      By lines~\code{\ref{alg:unistore-replication}}{\ref{line:propagate-knownvc-ts}}
      and \code{\ref{alg:unistore-replica}}{\ref{line:preparecausal-ts}},
      \begin{align*}
        &\knownVC^{m}_{d}(\realtime_{1})[d] \\
          &\quad = \min\set{\tsvar \mid \langle \_, \_, \tsvar \rangle \in \preparedcausal^{m}_{d}(\realtime_{1})} - 1 \\
          &\quad \le \min\set{\tsvar \mid \langle \_, \_, \tsvar \rangle \in \preparedcausal^{m}_{d}(\realtime_{2})} - 1 \\
          &\quad = \knownVC^{m}_{d}(\realtime_{2})[d].
      \end{align*}
  \end{itemize}
\end{proof}

\begin{applemma} \label{lemma:knownvc-i-nondecreasing}
  For $i \in \D \setminus \set{d}$,
  $\knownVC^{m}_{d}[i]$ at any replica $p^{m}_{d}$ in data center $d$
  is non-decreasing.
\end{applemma}

\begin{proof} \label{proof:knownvc-i-nondecreasing}
  Note that $\knownVC^{m}_{d}[i]$ ($i \in \D \setminus \set{d}$)
  can be updated only
  at lines~\code{\ref{alg:unistore-replication}}{\ref{line:replicate-knownvc}}
  and \code{\ref{alg:unistore-replication}}{\ref{line:heartbeat-knownvc}}.
  Therefore, this lemma holds due to
  lines~\code{\ref{alg:unistore-replication}}{\ref{line:replicate-precondition}}
  and \code{\ref{alg:unistore-replication}}{\ref{line:heartbeat-precondition}}.
\end{proof}

\begin{applemma} \label{lemma:knownvc-nondecreasing}
  For $i \in \D$, $\knownVC^{m}_{d}[i]$ at any replica $p^{m}_{d}$
  in data center $d$ is non-decreasing.
\end{applemma}

\begin{proof} \label{proof:knownvc-nondecreasing}
  By Lemmas~\ref{lemma:knownvc-d-nondecreasing}
  and \ref{lemma:knownvc-i-nondecreasing}.
\end{proof}

\begin{applemma} \label{lemma:knownvc-d-clock}
  For any replica $p^{m}_{d}$ in data center $d$,
  \[
    \knownVC^{m}_{d}[d] \le \clockVar^{m}_{d}.
  \]
\end{applemma}

\begin{proof} \label{proof:knownvc-d-clock}
  Note that $\knownVC^{m}_{d}[d]$ is updated only
  at lines~\code{\ref{alg:unistore-replication}}{\ref{line:propagate-knownvc-clock}}
  or \code{\ref{alg:unistore-replication}}{\ref{line:propagate-knownvc-ts}}.
  \begin{itemize}
    \item $\textsc{Case I}$: $\knownVC^{m}_{d}[d]$ is updated
      at line~\code{\ref{alg:unistore-replication}}{\ref{line:propagate-knownvc-clock}}.
      By Assumption~\ref{assumption:clock},
      \[
        \knownVC^{m}_{d}[d] \le \clockVar^{m}_{d}.
      \]
    \item $\textsc{Case II}$: $\knownVC^{m}_{d}[d]$ is updated
      at line~\code{\ref{alg:unistore-replication}}{\ref{line:propagate-knownvc-ts}}.
      By line~\code{\ref{alg:unistore-replica}}{\ref{line:preparecausal-ts}},
      immediately after this update,
      \[
        \knownVC^{m}_{d}[d] < \clockVar^{m}_{d}.
      \]
  \end{itemize}
\end{proof}

\begin{applemma} \label{lemma:knownvc-commitvc-d}
  Let $p^{m}_{d}$ be a replica in data center $d$.
  Consider $\knownVC^{m}_{d}(\realtime)[d]$ at time $\realtime$
  and transaction $\tvar \in \causaltxs$ committed at $p^{m}_{d}$
  after time $\realtime$. Then
  \[
    \commitVC(\tvar)[d] > \knownVC^{m}_{d}(\realtime)[d].
  \]
\end{applemma}

\begin{proof} \label{proof:knownvc-commitvc-d}
  Suppose that before time $\realtime$,
  $\knownVC^{m}_{d}[d]$ is last updated at time $\realtime' < \realtime$.
  Therefore,
  \[
    \knownVC^{m}_{d}(\realtime)[d] = \knownVC^{m}_{d}(\realtime')[d].
  \]
  We distinguish between two cases according to whether
  \[
    \preparedcausal^{m}_{d}(\realtime') = \emptyset
  \]
  when $\knownVC^{m}_{d}[d]$ is updated at time $\realtime'$.
  \begin{itemize}
    \item $\textsc{Case I}$: $\preparedcausal^{m}_{d}(\realtime') = \emptyset$.
      By line~\code{\ref{alg:unistore-replication}}{\ref{line:propagate-knownvc-clock}},
      \[
        \knownVC^{m}_{d}(\realtime')[d] = \clockVar^{m}_{d}(\realtime').
      \]
      By line~\code{\ref{alg:unistore-replica}}{\ref{line:preparecausal-ts}},
      line~\code{\ref{alg:unistore-coord}}{\ref{line:commitcausal-commitvc-d}},
      and Assumption~\ref{assumption:clock},
      \[
        \commitVC(\tvar)[d] > \clockVar^{m}_{d}(\realtime').
      \]
      Therefore,
      \begin{align*}
        \commitVC(\tvar)[d] &> \knownVC^{m}_{d}(\realtime')[d] \\
                          &= \knownVC^{m}_{d}(\realtime)[d].
      \end{align*}
    \item $\textsc{Case II}$: $\preparedcausal^{m}_{d}(\realtime') \neq \emptyset$.
      We further distinguish between two cases according to whether
      \[
        \langle \tidselector(\tvar), \_, \_ \rangle \in \preparedcausal^{m}_{d}(\realtime').
      \]
      \begin{itemize}
        \item $\textsc{Case II-1}$:
          $\langle \tidselector(\tvar), \_, \tsvar \rangle \in \preparedcausal^{m}_{d}(\realtime')$.
          By lines~\code{\ref{alg:unistore-replication}}{\ref{line:propagate-knownvc-ts}}
          and \code{\ref{alg:unistore-coord}}{\ref{line:commitcausal-commitvc-d}},
          \begin{align*}
            \commitVC(\tvar)[d] &\ge \tsvar \\
                              &> \knownVC^{m}_{d}(\realtime')[d] \\
                              &= \knownVC^{m}_{d}(\realtime)[d].
          \end{align*}
        \item $\textsc{Case II-2}$:
          $\langle \tidselector(\tvar), \_, \_ \rangle \notin \preparedcausal^{m}_{d}(\realtime')$.
          By Lemma~\ref{lemma:knownvc-d-clock},
          Assumption~\ref{assumption:clock},
          line~\code{\ref{alg:unistore-replica}}{\ref{line:preparecausal-ts}},
          and line~\code{\ref{alg:unistore-coord}}{\ref{line:commitcausal-commitvc-d}},
          \begin{align*}
            \commitVC(\tvar)[d] &> \knownVC^{m}_{d}(\realtime')[d] \\
                              &= \knownVC^{m}_{d}(\realtime)[d].
          \end{align*}
      \end{itemize}
  \end{itemize}
\end{proof}

\begin{applemma} \label{lemma:knownvc-local-d}
  Let $\tvar \in \causaltxs$ be a causal transaction
  that originates at data center $d$ and accesses partition $m$.
  If
  \[
    \commitVC(\tvar)[d] \le \knownVC^{m}_{d}[d],
  \]
  then
  \[
    \log(\tvar)[m] \subseteq \oplog^{m}_{d}.
  \]
\end{applemma}

\begin{proof} \label{proof:knownvc-local-d}
  Suppose that the value $\knownVC^{m}_{d}[d]$ is set at time $\realtime$.
  By Lemma~\ref{lemma:knownvc-commitvc-d},
  $\tvar$ is committed at $p^{m}_{d}$ before time $\realtime$.
  Therefore, by line~\code{\ref{alg:unistore-replica}}{\ref{line:commit-oplog}},
  \[
    \log(\tvar)[m] \subseteq \oplog^{m}_{d}.
  \]
\end{proof}

The following lemmas consider the replication and forwarding
of causal transactions.

\begin{applemma} \label{lemma:replication-order}
  Let $p^{m}_{d}$ be a replica in data center $d$.
  Let $\tvar_{1}$ and $\tvar_{2}$ be two transactions
  replicated by $p^{m}_{d}$ to sibling replicas
  at time $\realtime_{1}$ and $\realtime_{2}$
  (line~\code{\ref{alg:unistore-replication}}{\ref{line:propagate-call-replicate}}),
  respectively. Then
  \[
    \realtime_{1} < \realtime_{2} \implies \commitVC(\tvar_{1})[d] < \commitVC(\tvar_{2})[d].
  \]
\end{applemma}

\begin{proof} \label{proof:replication-order}
  Since $\tvar_{1}$ is replicated at time $\realtime_{1}$,
  by line~\code{\ref{alg:unistore-replication}}{\ref{line:propagate-txs}},
  \[
    \commitVC(\tvar_{1})[d] \le \knownVC^{m}_{d}(\realtime_{1})[d].
  \]
  Assume that $\realtime_{1} < \realtime_{2}$.
  We distinguish between two cases according to whether
  \[
    \langle \tidselector(\tvar_{2}), \_, \_, \_ \rangle \in \committedcausal^{m}_{d}(\realtime_{1})[d].
  \]
  \begin{itemize}
    \item $\textsc{Case I}$:
      $\langle \tidselector(\tvar_{2}), \_, \_, \_ \rangle \in \committedcausal^{m}_{d}(\realtime_{1})[d]$.
      Since $\tvar_{2}$ is not replicated at time $\realtime_{1}$,
      by line~\code{\ref{alg:unistore-replication}}{\ref{line:propagate-txs}},
      \[
        \commitVC(\tvar_{2})[d] > \knownVC^{m}_{d}(\realtime_{1})[d].
      \]
    \item $\textsc{Case II}$:
      $\langle \tidselector(\tvar_{2}), \_, \_, \_, \rangle \notin \committedcausal^{m}_{d}(\realtime_{1})[d]$.
      Thus, $\tvar_{2}$ is committed at $p^{m}_{d}$ after time $\realtime_{1}$.
      By Lemma~\ref{lemma:knownvc-commitvc-d},
      \[
        \commitVC(\tvar_{2})[d] > \knownVC^{m}_{d}(\realtime_{1})[d].
      \]
  \end{itemize}
  Therefore, in either case,
  \[
    \commitVC(\tvar_{1})[d] < \commitVC(\tvar_{2})[d].
  \]
\end{proof}

\begin{applemma} \label{lemma:heartbeat-replication-order}
  Let $p^{m}_{d}$ be a replica in data center $d$.
  Consider a heartbeat $\knownVC^{m}_{d}(\realtime_{1})[d]$
  sent by $p^{m}_{d}$ at time $\realtime_{1}$
  (line~\code{\ref{alg:unistore-replication}}{\ref{line:propagate-call-heartbeat}}).
  Let $\tvar$ be a transaction replicated by $p^{m}_{d}$ at time $\realtime_{2}$
  (line~\code{\ref{alg:unistore-replication}}{\ref{line:propagate-call-replicate}}).
  Then
  \[
    \realtime_{1} < \realtime_{2} \iff \knownVC^{m}_{d}(\realtime_{1})[d] < \commitVC(\tvar)[d].
  \]
\end{applemma}

\begin{proof} \label{proof:heartbeat-replication-order}
  We first show that
  \[
    \realtime_{1} < \realtime_{2} \implies \knownVC^{m}_{d}(\realtime_{1})[d] < \commitVC(\tvar)[d].
  \]
  Assume that $\realtime_{1} < \realtime_{2}$.
  We distinguish between two cases according to whether
  \[
    \langle \tidselector(\tvar), \_, \_, \_ \rangle \in \committedcausal^{m}_{d}(\realtime_{1})[d].
  \]
  \begin{itemize}
    \item $\textsc{Case I}$:
      $\langle \tidselector(\tvar), \_, \_, \_ \rangle \in \committedcausal^{m}_{d}(\realtime_{1})[d]$.
      By line~\code{\ref{alg:unistore-replication}}{\ref{line:propagate-txs}},
      \[
        \commitVC(\tvar)[d] > \knownVC^{m}_{d}(\realtime_{1})[d].
      \]
    \item $\textsc{Case II}$:
      $\langle \tidselector(\tvar), \_, \_, \_ \rangle \notin \committedcausal^{m}_{d}(\realtime_{1})[d]$.
      Thus, $\tvar$ is committed at $p^{m}_{d}$ after time $\realtime_{1}$.
      By Lemma~\ref{lemma:knownvc-commitvc-d},
      \[
        \commitVC(\tvar)[d] > \knownVC^{m}_{d}(\realtime_{1})[d].
      \]
  \end{itemize}

  Next we show that (note that $\realtime_{1} \neq \realtime_{2}$)
  \[
    \realtime_{2} < \realtime_{1} \implies \commitVC(\tvar)[d] \le \knownVC^{m}_{d}(\realtime_{1})[d].
  \]
  Since $\tvar$ is replicated by $p^{m}_{d}$ at time $\realtime_{2}$,
  by line~\code{\ref{alg:unistore-replication}}{\ref{line:propagate-txs}},
  \[
    \commitVC(\tvar)[d] \le \knownVC^{m}_{d}(\realtime_{2})[d].
  \]
  Assume that $\realtime_{2} < \realtime_{1}$.
  By Lemma~\ref{lemma:knownvc-d-nondecreasing},
  \[
    \knownVC^{m}_{d}(\realtime_{2})[d] \le \knownVC^{m}_{d}(\realtime_{1})[d].
  \]
  Putting it together yields
  \[
    \commitVC(\tvar)[d] \le \knownVC^{m}_{d}(\realtime_{1})[d].
  \]
\end{proof}

\begin{applemma} \label{lemma:committedcausal-i}
  Let $p^{m}_{d}$ be a replica in data center $d$. Then
  \begin{align*}
    &\forall i \neq d.\;
      \forall \langle \tidselector(\tvar), \_, \_, \_ \rangle \in \committedcausal^{m}_{d}[i]. \\
        &\quad \commitVC(\tvar)[i] \le \knownVC^{m}_{d}[i].
  \end{align*}
\end{applemma}

\begin{proof} \label{proof:committedcausal-i}
  By lines~\code{\ref{alg:unistore-replication}}{\ref{line:replicate-committedcausal}}
  and \code{\ref{alg:unistore-replication}}{\ref{line:replicate-knownvc}}
  and Lemma~\ref{lemma:knownvc-i-nondecreasing}.
\end{proof}

\begin{applemma} \label{lemma:globalmatrix-nondecreasing}
  For $j \neq d$ and $i \notin \set{d, j}$,
  $\globalmatrix^{m}_{d}[i][j]$ at any replica $p^{m}_{d}$
  in data center $d$ is non-decreasing.
\end{applemma}

\begin{proof} \label{proof:globalmatrix-nondecreasing}
  Note that $\globalmatrix^{m}_{d}[i][j]$ can be updated only
  at line~\code{\ref{alg:unistore-clock}}{\ref{line:knownvcglobal-globalmatrix}}.
  Therefore, by Lemma~\ref{lemma:knownvc-i-nondecreasing},
  it is non-decreasing.
\end{proof}

\begin{applemma} \label{lemma:forwarding-order}
  Let $p^{m}_{d}$ be a replica in data center $d$.
  Let $\tvar_{1}$ and $\tvar_{2}$ be two transactions
  that originate at data center $j \neq d$
  and are forwarded by $p^{m}_{d}$ to
  sibling replica $p^{m}_{i}$ in data center $i \notin \set{d, j}$
  at time $\realtime_{1}$ and $\realtime_{2}$
  (line~\code{\ref{alg:unistore-replication}}{\ref{line:forward-call-replicate}}),
  respectively. Then
  \[
    \realtime_{1} < \realtime_{2} \implies \commitVC(\tvar_{1})[j] < \commitVC(\tvar_{2})[j].
  \]
\end{applemma}

\begin{proof} \label{proof:forwarding-order}
  Since $\tvar_{1}$ is forwarded by $p^{m}_{d}$ at time $\realtime_{1}$,
  by line~\code{\ref{alg:unistore-replication}}{\ref{line:forward-txs}},
  \[
    \langle \tidselector(\tvar_{1}), \_, \_, \_ \rangle \in \committedcausal^{m}_{d}(\realtime_{1})[j].
  \]
  By Lemmas~\ref{lemma:committedcausal-i} and \ref{lemma:knownvc-i-nondecreasing},
  \begin{align}
    \commitVC(\tvar_{1})[j] \le \knownVC^{m}_{d}(\realtime_{1})[j].
    \label{eqn:tid1-knownvc}
  \end{align}
  Assume that $\realtime_{1} < \realtime_{2}$.
  We first argue that
  \begin{align}
    \langle \tidselector(\tvar_{2}), \_, \_, \_ \rangle \notin \committedcausal^{m}_{d}(\realtime_{1})[j].
    \label{eqn:tid2-committedcausal-t1}
  \end{align}
  Otherwise, by line~\code{\ref{alg:unistore-replication}}{\ref{line:forward-txs}},
  \[
    \commitVC(\tvar_{2})[j] \le \globalmatrix^{m}_{d}(\realtime_{1})[i][j].
  \]
  By Lemma~\ref{lemma:globalmatrix-nondecreasing},
  \[
    \commitVC(\tvar_{2})[j] \le \globalmatrix^{m}_{d}(\realtime_{2})[i][j].
  \]
  Therefore, by line~\code{\ref{alg:unistore-replication}}{\ref{line:forward-txs}},
  $\tvar_{2}$ would not be forwarded by $p^{m}_{d}$ to $p^{m}_{i}$ at time $\realtime_{2}$.
  Thus, (\ref{eqn:tid2-committedcausal-t1}) holds.
  Since $\tvar_{2}$ is forwarded by $p^{m}_{d}$ to $p^{m}_{i}$ at time $\realtime_{2}$,
  \[
    \langle \tidselector(\tvar_{2}), \_, _, \_ \rangle \in \committedcausal^{m}_{d}(\realtime_{2})[j].
  \]
  By Lemma~\ref{lemma:knownvc-i-nondecreasing}
  and line~\code{\ref{alg:unistore-replication}}{\ref{line:replicate-precondition}},
  \begin{align}
    \commitVC(\tvar_{2})[j] > \knownVC^{m}_{d}(\realtime_{1})[j].
    \label{eqn:tid2-knownvc}
  \end{align}
  Putting (\ref{eqn:tid1-knownvc}) and (\ref{eqn:tid2-knownvc}) together yields
  \[
    \commitVC(\tvar_{1})[j] < \commitVC(\tvar_{2})[j].
  \]
\end{proof}

\begin{applemma} \label{lemma:heartbeat-forwarding-order}
  Let $p^{m}_{d}$ be a replica in data center $d$.
  Consider a heartbeat $\knownVC^{m}_{d}(\realtime_{1})[j]$ ($j \neq d$)
  sent by $p^{m}_{d}$ to sibling replica $p^{m}_{i}$
  in data center $i \notin \set{d, j}$ at time $\realtime_{1}$
  (line~\code{\ref{alg:unistore-replication}}{\ref{line:forward-call-heartbeat}}).
  Let $\tvar$ be a transaction that originates at data center $j$
  and is forwarded by $p^{m}_{d}$ to $p^{m}_{i}$ at time $\realtime_{2}$
  (line~\code{\ref{alg:unistore-replication}}{\ref{line:forward-call-replicate}}).
  Then
  \[
    \realtime_{1} < \realtime_{2} \iff \knownVC^{m}_{d}(\realtime_{1})[j] < \commitVC(\tvar)[j].
  \]
\end{applemma}

\begin{proof} \label{proof:heartbeat-forwarding-order}
  We first show that
  \[
    \realtime_{1} < \realtime_{2} \implies \knownVC^{m}_{d}(\realtime_{1})[j] < \commitVC(\tvar)[j].
  \]
  Assume that $\realtime_{1} < \realtime_{2}$.
  We first argue that
  \begin{align}
    \langle \tidselector(\tvar), \_, \_, \_ \rangle \notin \committedcausal^{m}_{d}(\realtime_{1})[j].
    \label{eqn:tid-notin-committedcausal-t1}
  \end{align}
  Otherwise, since $\tvar$ is not forwarded at time $\realtime_{1}$,
  by line~\code{\ref{alg:unistore-replication}}{\ref{line:forward-txs}},
  \[
    \commitVC(\tvar)[j] \le \globalmatrix^{m}_{d}(\realtime_{1})[i][j].
  \]
  By Lemma~\ref{lemma:globalmatrix-nondecreasing},
  \[
    \commitVC(\tvar)[j] \le \globalmatrix^{m}_{d}(\realtime_{2})[i][j].
  \]
  Therefore, by line~\code{\ref{alg:unistore-replication}}{\ref{line:forward-txs}},
  $\tvar$ would not be forwarded by $p^{m}_{d}$ to $p^{m}_{i}$ at time $\realtime_{2}$.
  Thus, (\ref{eqn:tid-notin-committedcausal-t1}) holds.
  Since $\tvar$ is forwarded by $p^{m}_{d}$ to $p^{m}_{i}$ at time $\realtime_{2}$,
  \[
    \langle \tidselector(\tvar), \_, \_, \_ \rangle \in \committedcausal^{m}_{d}(\realtime_{2})[j].
  \]
  By Lemma~\ref{lemma:knownvc-i-nondecreasing}
  and line~\code{\ref{alg:unistore-replication}}{\ref{line:replicate-precondition}},
  \[
    \knownVC^{m}_{d}(\realtime_{1})[j] < \commitVC(\tvar)[j].
  \]

  Next we show that (note that $\realtime_{1} \neq \realtime_{2}$)
  \[
    \realtime_{2} < \realtime_{1} \implies \commitVC(\tvar)[j] \le \knownVC^{m}_{d}(\realtime_{1})[j].
  \]
  Since $\tvar$ is forwarded by $p^{m}_{d}$ to $p^{m}_{i}$ at time $\realtime_{2}$,
  by line~\code{\ref{alg:unistore-replication}}{\ref{line:forward-txs}},
  \[
    \langle \tidselector(\tvar), \_, \_, \_ \rangle \in \committedcausal^{m}_{d}(\realtime_{2})[j].
  \]
  By Lemmas~\ref{lemma:committedcausal-i} and \ref{lemma:knownvc-i-nondecreasing},
  \[
    \commitVC(\tvar)[j] \le \knownVC^{m}_{d}(\realtime_{2})[j].
  \]
  Assume that $\realtime_{2} < \realtime_{1}$.
  By Lemma~\ref{lemma:knownvc-i-nondecreasing},
  \[
    \knownVC^{m}_{d}(\realtime_{2})[j] \le \knownVC^{m}_{d}(\realtime_{1})[j].
  \]
  Putting it together yields
  \[
    \commitVC(\tvar)[j] \le \knownVC^{m}_{d}(\realtime_{1})[j].
  \]
\end{proof}

\begin{applemma} \label{lemma:replication-knownvc}
  Let $\tvar \in \causaltxs$ be a causal transaction
  that originates at data center $i$ and accesses partition $m$.
  If
  \[
    \commitVC(\tvar)[i] \le \knownVC^{m}_{d}[i]
  \]
  for replica $p^{m}_{d}$ in data center $d \neq i$,
  then
  \[
    \log(\tvar)[m] \subseteq \oplog^{m}_{d}.
  \]
\end{applemma}

\begin{proof} \label{proof:replication-knownvc}
  Note that for $i \in \D \setminus \set{d}$,
  $\knownVC^{m}_{d}[i]$ can be updated
  only at lines~\code{\ref{alg:unistore-replication}}{\ref{line:replicate-knownvc}}
  or \code{\ref{alg:unistore-replication}}{\ref{line:heartbeat-knownvc}}
  due to replication of transactions or heartbeats respectively,
  either directly from data center $i$
  (line~\code{\ref{alg:unistore-replication}}{\ref{line:function-propagate}})
  or indirectly from a third data center $j \neq i$
  (line~\code{\ref{alg:unistore-replication}}{\ref{line:function-forward}}).

  We proceed by induction on the length of the execution.
  In the following, for replica $p^{m}_{d}$ in data center $d \in \D$,
  we denote the value of $\knownVC^{m}_{d}$ (resp. $\oplog^{m}_{d}$)
  after $k$ steps in an execution
  by $\knownVC^{m}_{d}(k)$ (resp. $\oplog^{m}_{d}(k)$).
  \begin{itemize}
    \item \emph{Base Case.} $k = 0$.
      It holds trivially,
      since for replica $p^{m}_{d}$ in any data center $d \neq i$,
      \[
        \knownVC^{m}_{d}(0)[i] = 0.
      \]
    \item \emph{Induction Hypothesis.}
      Suppose that for any execution of length $k$, we have
      \begin{align*}
        &\forall d \in \D \setminus \set{i}.\;
          \forall \tvar \in \causaltxs.\; \\
            &\quad \big(\commitVC(\tvar)[i] \le \knownVC^{m}_{d}(k)[i] \\
            &\quad\quad \implies \log(\tvar)[m] \subseteq \oplog^{m}_{d}(k)\big).
      \end{align*}
    \item \emph{Induction Step.}
      Consider an execution of length $k + 1$.
      If the $(k+1)$-st step of this execution does not update
      $\knownVC^{m}_{d}[i]$ for replica $p^{m}_{d}$ in any data center $d \neq i$,
      then by the induction hypothesis,
      \begin{align*}
        &\forall d \in \D \setminus \set{i}.\;
          \forall \tvar \in \causaltxs.\; \\
            &\quad \big(\commitVC(\tvar)[i] \le \knownVC^{m}_{d}(k+1)[i] \\
            &\quad\quad \implies \log(\tvar)[m] \subseteq \oplog^{m}_{d}(k+1)\big).
      \end{align*}
      Otherwise, we perform a case analysis according to
      how $\knownVC^{m}_{d}[i]$ of replica $p^{m}_{d}$ in data center $d \neq i$
      is updated in the $(k+1)$-st step.
      \begin{itemize}
        \item \textsc{Case I:}
          $\knownVC^{m}_{d}[i]$ is updated due to delivery of a message
          from data center $i$.
          By Lemmas~\ref{lemma:knownvc-d-nondecreasing},
          \ref{lemma:replication-order}, and \ref{lemma:heartbeat-replication-order},
          local transactions and heartbeats are propagated by $p^{m}_{i}$ to sibling replicas
          in increasing order of their local timestamps $\commitvc[i]$
          and $\knownVC^{m}_{i}[i]$ values.
          Therefore, by Assumption~\ref{assumption:message}
          and the induction hypothesis,
          \begin{align*}
            &\forall \tvar \in \causaltxs.\; \\
              &\quad \big(\commitVC(\tvar)[i] \le \knownVC^{m}_{d}(k+1)[i] \\
              &\quad\quad \implies \log(\tvar)[m] \subseteq \oplog^{m}_{d}(k+1)\big).
          \end{align*}
        \item \textsc{Case II:}
          $\knownVC^{m}_{d}[i]$ is updated due to delivery of a message
          from a third data center $j \neq i$.
          By Lemmas~\ref{lemma:knownvc-i-nondecreasing},
          \ref{lemma:forwarding-order}, and \ref{lemma:heartbeat-forwarding-order},
          transactions originating at data center $i$ and heartbeats are forwarded
          by some replica, say $p^{m}_{j} (j \neq i)$,
          to sibling replicas in increasing order of
          their local timestamps $\commitvc[i]$ and $\knownVC^{m}_{j}[i]$ values.
          Therefore, by Assumption~\ref{assumption:message}
          and the induction hypothesis,
          \begin{align*}
            &\forall \tvar \in \causaltxs.\; \\
              &\quad \big(\commitVC(\tvar)[i] \le \knownVC^{m}_{d}(k+1)[i] \\
                &\quad\quad \implies \log(\tvar)[m] \subseteq \oplog^{m}_{d}(k+1)\big).
          \end{align*}
      \end{itemize}
  \end{itemize}
\end{proof}

\begin{applemma}[\prop{1}] \label{lemma:knownvc-causal}
  Let $\tvar \in \causaltxs$ be a causal transaction
  that originates at data center $i$
  and accesses partition $m$.
  If
  \[
    \commitVC(\tvar)[i] \le \knownVC^{m}_{d}[i]
  \]
  for replica $p^{m}_{d}$ in data center $d$,
  then
  \[
    \log(\tvar)[m] \subseteq \oplog^{m}_{d}.
  \]
\end{applemma}

\begin{proof} \label{proof:knownvc-causal}
  By Lemmas~\ref{lemma:knownvc-local-d} and \ref{lemma:replication-knownvc}.
\end{proof}
\subsubsection{Properties of $\stableVC$}
\label{sss:stablevc}

\begin{applemma} \label{lemma:stablevc-nondecreasing}
  For $i \in \D$, $\stableVC^{m}_{d}[i]$
  at any replica $p^{m}_{d}$ in data center $d$ is non-decreasing.
\end{applemma}

\begin{proof} \label{proof:stablevc-nondecreasing}
  Note that $\stableVC^{m}_{d}[i]$ ($i \in \D$) can be updated
  only at line~\code{\ref{alg:unistore-clock}}{\ref{line:knownvclocal-stablevc-causal}}.
  By Lemma~\ref{lemma:knownvc-nondecreasing}
  and Assumption~\ref{assumption:message},
  $\stableVC^{m}_{d}[i]$ is non-decreasing.
\end{proof}

\begin{applemma}[\prop{2}] \label{lemma:stablevc-knownvc}
  For any replica $p^{m}_{d}$ in data center $d$,
  \[
    \forall i \in \D.\; \forall n \in \P.\;
      \stableVC^{m}_{d}[i] \le \knownVC^{n}_{d}[i].
  \]
\end{applemma}

\begin{proof} \label{proof:stablevc-knownvc}
  Note that $\stableVC^{m}_{d}[i]$ ($i \in \D$)
  can be updated only
  at line~\code{\ref{alg:unistore-clock}}{\ref{line:knownvclocal-stablevc-causal}}.
  By the way $\stableVC^{m}_{d}[i]$ is updated
  and Lemmas~\ref{lemma:knownvc-d-nondecreasing}
  and \ref{lemma:knownvc-i-nondecreasing},
  \[
    \forall n \in \P.\; \stableVC^{m}_{d}[i] \le \knownVC^{n}_{d}[i].
  \]
\end{proof}

\begin{applemma} \label{lemma:replication-stablevc}
  Let $\tvar \in \causaltxs$ be a causal transaction
  that originates at data center $i$
  and accesses partition $n$.
  If
  \[
    \commitVC(\tvar)[i] \le \stableVC^{m}_{d}[i]
  \]
  for some replica $p^{m}_{d}$ in data center $d$,
  then
  \[
    \log(\tvar)[n] \subseteq \oplog^{n}_{d}.
  \]
\end{applemma}

\begin{proof} \label{proof:replication-stablevc}
  By Lemma~\ref{lemma:stablevc-knownvc},
  \[
    \stableVC^{m}_{d}[i] \le \knownVC^{n}_{d}[i].
  \]
  Therefore,
  \[
    \commitVC(\tvar)[i] \le \knownVC^{n}_{d}[i].
  \]
  By Lemma~\ref{lemma:replication-knownvc},
  \[
    \log(\tvar)[n] \subseteq \oplog^{n}_{d}.
  \]
\end{proof}
\subsubsection{Properties of $\uniformVC$}
\label{sss:uniformvc}

\begin{applemma} \label{lemma:uniformvc-nondecreasing}
  For $i \in \D$, $\uniformVC^{m}_{d}[i]$
  at any replica $p^{m}_{d}$ in data center $d$ is non-decreasing.
\end{applemma}

\begin{proof} \label{proof:uniformvc-nondecreasing}
  Note that whenever $\uniformVC^{m}_{d}[i]$ is updated
  at lines~\code{\ref{alg:unistore-coord}}{\ref{line:start-uniformvc}},
  \code{\ref{alg:unistore-replica}}{\ref{line:readkey-uniformvc}},
  \code{\ref{alg:unistore-replica}}{\ref{line:preparecausal-uniformvc}},
  or \code{\ref{alg:unistore-clock}}{\ref{line:stablevc-uniformvc}},
  we take the maximum of it and some scalar value.
\end{proof}

\begin{applemma} \label{lemma:pastvc-uniformvc-except-d}
  Let $e \in E$ be an event issued by client $\cl$
  to replica $p^{m}_{d}$ in data center $d$. Then
  \begin{align*}
    &e \in E \setminus \Fence \implies \\
      &\quad \forall i \in \D \setminus \set{d}.\;
      (\pastVC_{\cl})_{e}[i] \le (\uniformVC^{m}_{d})_{e}[i],
  \end{align*}
  and
  \[
    e \in \Fence \implies (\pastVC_{\cl})_{e}[d] \le (\uniformVC^{m}_{d})_{e}[d].
  \]
\end{applemma}

\begin{proof} \label{proof:pastvc-uniformvc-except-d}
  We perform a case analysis according to the type of event $e$.
  \begin{itemize}
    \item $\textsc{Case I}$: $e \in S$.
      By line~\code{\ref{alg:unistore-coord}}{\ref{line:start-uniformvc}},
      \[
        \forall i \in \D \setminus \set{d}.\;
          (\pastVC_{\cl})_{e}[i] \le (\uniformVC^{m}_{d})_{e}[i].
      \]
    \item $\textsc{Case II}$: $e \in R \cup U$.
      In this case,
      \[
        (\pastVC_{\cl})_{e} = (\pastVC_{\cl})_{\startoftx(e)}.
      \]
      By \textsc{Case I},
      \[
        \forall i \in \D \setminus \set{d}.\;
          (\pastVC_{\cl})_{\startoftx(e)}[i] \le (\uniformVC^{m}_{d})_{\startoftx(e)}[i].
      \]
      By Lemma~\ref{lemma:uniformvc-nondecreasing},
      \[
        \forall i \in \D \setminus \set{d}. \\
          (\uniformVC^{m}_{d})_{\startoftx(e)} \le (\uniformVC^{m}_{d})_{e}.
      \]
      Putting it together yields
      \[
        \forall i \in \D \setminus \set{d}.\;
          (\pastVC_{\cl})_{e}[i] \le (\uniformVC^{m}_{d})_{e}[i].
      \]
    \item $\textsc{Case III}$: $e \in C_{\causalentry}$.
      By line~\code{\ref{alg:unistore-client}}{\ref{line:commitcausaltx-pastvc}},
      \[
        (\pastVC_{\cl})_{e} = vc_{(\commitcausaltx, e)}.
      \]
      By lines~\code{\ref{alg:unistore-coord}}{\ref{line:commitcausal-return-ro}},
      \code{\ref{alg:unistore-coord}}{\ref{line:commitcausal-commitvc}},
      and \code{\ref{alg:unistore-coord}}{\ref{line:commitcausal-return}},
      \begin{align*}
        &\forall i \in \D \setminus \set{d}. \\
          &\quad vc_{(\commitcausaltx, e)}[i] = (\snapVC^{m}_{d})_{e}[\txfunc(e)][i].
      \end{align*}
      By line~\code{\ref{alg:unistore-coord}}{\ref{line:start-snapvc}},
      \begin{align*}
        &\forall i \in \D \setminus \set{d}. \\
          &\quad (\snapVC^{m}_{d})_{e}[\txfunc(e)][i] = (\uniformVC^{m}_{d})_{\startoftx(e)}[i].
      \end{align*}
      By Lemma~\ref{lemma:uniformvc-nondecreasing},
      \[
        \forall i \in \D \setminus \set{d}. \\
          (\uniformVC^{m}_{d})_{\startoftx(e)} \le (\uniformVC^{m}_{d})_{e}.
      \]
      Putting it together yields
      \[
        \forall i \in \D \setminus \set{d}.\;
          (\pastVC_{\cl})_{e}[i] \le (\uniformVC^{m}_{d})_{e}[i].
      \]
    \item $\textsc{Case IV}$: $e \in C_{\strongentry}$.
      By line~\code{\ref{alg:unistore-client}}{\ref{line:commitstrongtx-pastvc}},
      \[
        (\pastVC_{\cl})_{e} = vc_{(\commitstrongtx, e)}.
      \]
      By (\ref{eqn:gcf-commitvc}),
      \begin{align*}
        &\forall i \in \D \setminus \set{d}. \\
          &\quad vc_{(\commitstrongtx, e)}[i] = (\snapVC^{m}_{d})_{e}[\txfunc(e)][i].
      \end{align*}
      Therefore, similar to \textsc{Case III}, we have
      \[
        \forall i \in \D \setminus \set{d}.\;
          (\pastVC_{\cl})_{e}[i] \le (\uniformVC^{m}_{d})_{e}[i].
      \]
    \item $\textsc{Case V}$: $e \in \Fence$.
      By line~\code{\ref{alg:unistore-replica}}{\ref{line:uniformbarrier-wait-uniformvc-d}},
      \[
        (\pastVC_{\cl})_{e}[d] \le (\uniformVC^{m}_{d})_{e}[d].
      \]
    \item $\textsc{Case VI}$: $e \in \Attach$.
      By line~\code{\ref{alg:unistore-replica}}{\ref{line:attach-wait-condition}},
      \[
        \forall i \in \D \setminus \set{d}.\;
          (\pastVC_{\cl})_{e}[i] \le (\uniformVC^{m}_{d})_{e}[i].
      \]
  \end{itemize}
\end{proof}

\begin{applemma} \label{lemma:pastvc-uniformvc}
  Let $\cl$ be a client and $d \triangleq \cldc(\cl)$.
  At any time,
  \[
    \forall i \in \D \setminus \set{d}.\;
      \pastVC_{\cl}[i] \le \uniformVC^{m}_{d}[i]
  \]
  for some replica $p^{m}_{d}$ in data center $d$.
\end{applemma}

\begin{proof} \label{proof:pastvc-uniformvc}
  By a simple induction on the number of events that $\cl$ issues
  and Lemmas~\ref{lemma:pastvc-uniformvc-except-d} and
  \ref{lemma:uniformvc-nondecreasing}.
\end{proof}

\begin{applemma} \label{lemma:snapshotvc-uniformvc}
  Let $\tvar$ be a transaction that originates at data center $d$.
  At any time,
  \[
    \forall i \in \D \setminus \set{d}.\;
      \snapshotVC(\tvar)[i] \le \uniformVC^{m}_{d}[i]
  \]
  for some replica $p^{m}_{d}$ in data center $d$.
\end{applemma}

\begin{proof} \label{proof:snapshotvc-uniformvc}
  By line~\code{\ref{alg:unistore-coord}}{\ref{line:start-snapvc}}
  and Lemma~\ref{lemma:uniformvc-nondecreasing}.
\end{proof}

\begin{applemma}[\prop{3}] \label{lemma:uniformvc-knownvc-f+1}
  For any replica $p^{m}_{d}$ in data center $d$,
  \begin{align*}
    &\forall i \in \D.\; \exists g \subseteq \D.\; \Big(
       |g| \ge f + 1 \land d \in g\; \land \\
      &\; \big(\forall j \in g.\; \forall n \in \P.\;
        \uniformVC^{m}_{d}[i] \le \knownVC^{n}_{j}[i] \big)\Big).
  \end{align*}
\end{applemma}

\begin{proof} \label{proof:uniformvc-knownvc-f+1}
  Fix $i \in \D$.
  We proceed by induction on the length of the execution.
  In the following, we denote the value of
  $\knownVC^{m}_{d}$, $\stableVC^{m}_{d}$, $\uniformVC^{m}_{d}$,
  $\stablematrix^{m}_{d}$, and $\pastVC_{\cl}$ (for some client $\cl$)
  after $k$ steps of an execution by
  $\knownVC^{m}_{d}(k)$, $\stableVC^{m}_{d}(k)$, $\uniformVC^{m}_{d}(k)$,
  $\stablematrix^{m}_{d}(k)$, and $\pastVC_{\cl}(k)$, respectively.
  \begin{itemize}
    \item {\it Base Case.} $k = 0$. It holds trivially since
      \[
        \uniformVC^{m}_{d}(0)[i] = 0.
      \]
    \item {\it Induction Hypothesis.}
      Suppose that for any execution of length $k$,
      for any replica $p^{m}_{d}$ in data center $d$,
      \begin{align*}
        &\exists g \subseteq \D.\; |g| \ge f + 1 \land d \in g\; \land \\
          &\quad \big(\forall j \in g.\; \forall 1 \le n \le N.\; \\
            &\qquad \uniformVC^{m}_{d}(k)[i] \le \knownVC^{n}_{j}(k)[i] \big).
      \end{align*}
    \item {\it Induction Step.}
      Consider an execution of length $k + 1$.
      If the $(k+1)$-st step of this execution does not update
      $\uniformVC^{m}_{d}[i]$ for any replica $p^{m}_{d}$ in data center $d$,
      then by the induction hypothesis and Lemma~\ref{lemma:knownvc-nondecreasing},
      \begin{align*}
        &\exists g \subseteq \D.\; |g| \ge f + 1 \land d \in g\; \land \\
          &\quad \big(\forall j \in g.\; \forall 1 \le n \le N.\; \\
            &\qquad \uniformVC^{m}_{d}(k + 1)[i] = \uniformVC^{m}_{d}(k)[i] \\
            &\phantom{\qquad \uniformVC^{m}_{d}(k + 1)[i]}
              \le \knownVC^{n}_{j}(k)[i] \\
            &\phantom{\qquad \uniformVC^{m}_{d}(k + 1)[i]}
              \le \knownVC^{n}_{j}(k + 1)[i] \big).
      \end{align*}
      Otherwise, we perform a case analysis
      according to how $\uniformVC^{m}_{d}[i]$ is updated.
      \begin{itemize}
        \item $\textsc{Case I}$: $\uniformVC^{m}_{d}[i]$ is updated
          at line~\code{\ref{alg:unistore-clock}}{\ref{line:stablevc-uniformvc}}.
          By line~\code{\ref{alg:unistore-clock}}{\ref{line:stablevc-g}},
          \begin{align}
            &\exists g' \subseteq \D.\; |g'| \ge f + 1 \land d \in g'\; \land
              \label{eqn:g-prime} \\
              &\quad \uniformVC^{m}_{d}(k+1)[i] = \nonumber \\
                &\qquad \max\big\{\uniformVC^{m}_{d}(k)[i], \nonumber \\
                  &\qquad\qquad\;\; \min_{j \in g'} \stablematrix^{m}_{d}(k+1)[j][i]\big\}. \nonumber
          \end{align}
          By the induction hypothesis and Lemma~\ref{lemma:knownvc-i-nondecreasing},
          \begin{align}
            &\exists g'' \subseteq \D.\; |g''| \ge f + 1 \land d \in g''\; \land
              \label{eqn:g-prime-prime} \\
              &\quad \big(\forall j \in g''.\; \forall n \in \P. \nonumber \\
                &\qquad \uniformVC^{m}_{d}(k)[i] \le \knownVC^{n}_{j}(k)[i] \nonumber \\
                &\phantom{\qquad \uniformVC^{m}_{d}(k)[i] }
                  \le \knownVC^{n}_{j}(k+1)[i]\big). \nonumber
          \end{align}
          By Lemma~\ref{lemma:stablevc-nondecreasing},
          for the particular $g' \subseteq \D$ in (\ref{eqn:g-prime}),
          \begin{align}
            \forall j \in g'.\; &\stablematrix^{m}_{d}(k+1)[j][i]
              \label{eqn:stablematrix-j-i} \\
              &\le \stableVC^{m}_{j}(k+1)[i].
              \nonumber
          \end{align}
          By Lemmas~\ref{lemma:stablevc-knownvc} and \ref{lemma:knownvc-i-nondecreasing},
          for any replica $p^{m}_{j}$ in data center $j$,
          \begin{align}
            &\forall n \in \P.\; \stableVC^{m}_{j}(k+1)[i]
              \label{eqn:stablevc-j-i}\\
              &\qquad\;\;\; \le \knownVC^{n}_{j}(k+1)[i]. \nonumber
          \end{align}
          Therefore, for the particular $g' \subseteq \D$ in $(\ref{eqn:g-prime})$,
          \begin{align}
            &\forall j' \in g'.\; \forall n \in \P.\;
              \min_{j \in g'} \stablematrix^{m}_{d}(k+1)[j][i]
              \nonumber\\
              &\qquad\qquad\qquad\;\; \le \knownVC^{n}_{j'}(k+1)[i]. \label{eqn:j-prime}
          \end{align}
          By (\ref{eqn:g-prime}), (\ref{eqn:g-prime-prime}),
          and (\ref{eqn:j-prime}),
          we can either take $g = g'$ in (\ref{eqn:g-prime})
          or $g = g''$ in (\ref{eqn:g-prime-prime}) such that
          \begin{align*}
            &\forall j \in g.\; \forall n \in \P. \\
              &\;\; \uniformVC^{m}_{d}(k+1)[i] \le \knownVC^{n}_{j}(k+1)[i].
          \end{align*}
          Therefore,
          \begin{align*}
            &\exists g \subseteq \D.\; |g| \ge f + 1 \land d \in g\; \land \\
              &\; \big(\forall j \in g.\; \forall n \in \P.\; \\
                &\quad \uniformVC^{m}_{d}(k+1)[i] \le \knownVC^{n}_{j}(k+1)[i] \big).
          \end{align*}
        \item $\textsc{Case II}$: $\uniformVC^{m}_{d}[i]$
          ($i \in \D \setminus \set{d}$) is updated
          at line~\code{\ref{alg:unistore-coord}}{\ref{line:start-uniformvc}}.
          Then there exists some client $\cl$ with $d = \cldc(\cl)$ such that
          \begin{align}
            &\uniformVC^{m}_{d}(k+1)[i] = \label{eqn:uniformvc-pastvc} \\
              &\quad \max\big\{\pastVC_{\cl}(k)[i], \uniformVC^{m}_{d}(k)[i]\big\}.
              \nonumber
          \end{align}
          By the induction hypothesis and Lemma~\ref{lemma:knownvc-i-nondecreasing},
          \begin{align}
            &\exists g' \subseteq \D.\; |g'| \ge f + 1 \land d \in g'\; \land
              \label{eqn:g-prime-2} \\
              &\quad \big(\forall j \in g'.\; \forall n \in \P. \nonumber \\
                &\qquad \uniformVC^{m}_{d}(k)[i] \le \knownVC^{n}_{j}(k)[i] \nonumber \\
                &\phantom{\qquad \uniformVC^{m}_{d}(k)[i] }
                  \le \knownVC^{n}_{j}(k+1)[i]\big). \nonumber
          \end{align}
          By Lemma~\ref{lemma:pastvc-uniformvc}, the induction hypothesis,
          and Lemma~\ref{lemma:knownvc-nondecreasing},
          \begin{align}
            &\exists g'' \subseteq \D.\; |g''| \ge f + 1 \land d \in g''\; \land
              \label{eqn:g-prime-prime-2} \\
              &\quad \big(\forall j \in g''.\; \forall n \in \P. \nonumber \\
                &\qquad \pastVC_{\cl}(k)[i] \le \knownVC^{n}_{j}(k+1)[i]. \nonumber
          \end{align}
          By (\ref{eqn:uniformvc-pastvc}), (\ref{eqn:g-prime-2}),
          and (\ref{eqn:g-prime-prime-2}),
          we can take $g = g'$ in (\ref{eqn:g-prime-2})
          or $g = g''$ in (\ref{eqn:g-prime-prime-2}) such that
          \begin{align*}
            &\forall j \in g.\; \forall n \in \P.\; \\
              &\;\; \uniformVC^{m}_{d}(k+1)[i] \le \knownVC^{n}_{j}(k+1)[i].
          \end{align*}
          Therefore,
          \begin{align*}
            &\exists g \subseteq \D.\; |g| \ge f + 1 \land d \in g\; \land \\
              &\; \big(\forall j \in g.\; \forall n \in \P.\; \\
                &\;\; \uniformVC^{m}_{d}(k+1)[i] \le \knownVC^{n}_{j}(k+1)[i] \big).
          \end{align*}
        \item $\textsc{Case III}$: $\uniformVC^{m}_{d}[i]$
          ($i \in \D \setminus \set{d}$) is updated
          at lines~\code{\ref{alg:unistore-replica}}{\ref{line:readkey-uniformvc}}
          or \code{\ref{alg:unistore-replica}}{\ref{line:preparecausal-uniformvc}}.
          Therefore, there exists some transaction $\tvar$
          originating at data center $d$ such that
          \begin{align}
            &\uniformVC^{m}_{d}(k+1)[i] = \label{eqn:uniformvc-snapshotvc} \\
              &\quad \max\big\{\snapshotVC(\tvar)[i], \uniformVC^{m}_{d}(k)[i]\big\}.
              \nonumber
          \end{align}
          By the induction hypothesis and Lemma~\ref{lemma:knownvc-nondecreasing},
          \begin{align}
            &\exists g' \subseteq \D.\; |g'| \ge f + 1 \land d \in g'\; \land
              \label{eqn:g-prime-3} \\
              &\quad \big(\forall j \in g'.\; \forall n \in \P. \nonumber \\
                &\qquad \uniformVC^{m}_{d}(k)[i] \le \knownVC^{n}_{j}(k)[i] \nonumber \\
                &\phantom{\qquad \uniformVC^{m}_{d}(k)[i] }
                  \le \knownVC^{n}_{j}(k+1)[i]\big). \nonumber
          \end{align}
          By Lemma~\ref{lemma:snapshotvc-uniformvc}, the induction hypothesis,
          and Lemma~\ref{lemma:knownvc-nondecreasing},
          \begin{align}
            &\exists g'' \subseteq \D.\; |g''| \ge f + 1 \land d \in g''\; \land
              \label{eqn:g-prime-prime-3} \\
              &\quad \big(\forall j \in g''.\; \forall n \in \P. \nonumber \\
                &\qquad \snapshotVC(\tvar)[i] \le \knownVC^{n}_{j}(k+1)[i]. \nonumber
          \end{align}
          By (\ref{eqn:uniformvc-snapshotvc}), (\ref{eqn:g-prime-3}),
          and (\ref{eqn:g-prime-prime-3}),
          we can take $g = g'$ in (\ref{eqn:g-prime-3})
          or $g = g''$ in (\ref{eqn:g-prime-prime-3}) such that
          \begin{align*}
            &\forall j \in g.\; \forall n \in \P.\; \\
              &\;\; \uniformVC^{m}_{d}(k+1)[i] \le \knownVC^{n}_{j}(k+1)[i].
          \end{align*}
          Therefore,
          \begin{align*}
            &\exists g \subseteq \D.\; |g| \ge f + 1 \land d \in g\; \land \\
              &\; \big(\forall j \in g.\; \forall n \in \P.\; \\
                &\quad \uniformVC^{m}_{d}(k+1)[i] \le \knownVC^{n}_{j}(k+1)[i] \big).
          \end{align*}
      \end{itemize}
  \end{itemize}
\end{proof}

\begin{applemma} \label{lemma:uniformvc-knownvc}
  For any replica $p^{m}_{d}$ in data center $d$,
  \[
    \forall i \in \D.\; \forall n \in \P.\;
      \uniformVC^{m}_{d}[i] \le \knownVC^{n}_{d}[i].
  \]
\end{applemma}

\begin{proof} \label{proof:uniformvc-knownvc}
  By Lemma~\ref{lemma:uniformvc-knownvc-f+1}.
\end{proof}

\begin{applemma} \label{lemma:replication-uniformvc}
  Let $\tvar \in \causaltxs$ be a causal transaction
  that originates at data center $i$
  and accesses partition $n$.
  If
  \[
    \commitVC(\tvar)[i] \le \uniformVC^{m}_{d}[i]
  \]
  for some replica $p^{m}_{d}$ in data center $d$,
  then
  \[
    \log(\tvar)[n] \subseteq \oplog^{n}_{d}.
  \]
\end{applemma}

\begin{proof} \label{proof:replication-uniformvc}
  By Lemma~\ref{lemma:uniformvc-knownvc},
  \[
    \uniformVC^{m}_{d}[i] \le \knownVC^{n}_{d}[i].
  \]
  Therefore,
  \[
    \commitVC(\tvar)[i] \le \knownVC^{n}_{d}[i].
  \]
  By Lemma~\ref{lemma:replication-knownvc},
  \[
    \log(\tvar)[n] \subseteq \oplog^{n}_{d}.
  \]
\end{proof}

\begin{applemma} \label{lemma:uniformvc-clock}
  For any replica $p^{m}_{d}$ in data center $d$,
  \[
    \uniformVC^{m}_{d}[d] \le \clockVar^{m}_{d}.
  \]
\end{applemma}

\begin{proof} \label{proof:uniformvc-clock}
  By Lemma~\ref{lemma:uniformvc-knownvc},
  \[
    \uniformVC^{m}_{d}[d] \le \knownVC^{m}_{d}[d].
  \]
  By Lemma~\ref{lemma:knownvc-d-clock},
  \[
    \knownVC^{m}_{d}[d] \le \clockVar^{m}_{d}.
  \]
  Putting it together yields
  \[
    \uniformVC^{m}_{d}[d] \le \clockVar^{m}_{d}.
  \]
\end{proof}
\subsubsection{Properties of $\pastVC$}
\label{sss:cvc}

\begin{applemma} \label{lemma:start-pastvc-snapshotvc}
  Let $e \in S$ be a \start{} event of transaction $\tvar$
  issued by client $\cl$. Then
  \[
    (\pastVC_{\cl})_{e} \le \snapshotVC(\tvar).
  \]
\end{applemma}

\begin{proof} \label{proof:start-pastvc-snapshotvc}
  By Definition~\ref{def:snapshotvc} of $\snapshotVC(\tvar)$
  and lines~\code{\ref{alg:unistore-coord}}{\ref{line:start-uniformvc-index}}--
  \code{\ref{alg:unistore-coord}}{\ref{line:start-snapvc-strong}}.
\end{proof}

\begin{applemma} \label{lemma:pastvc-nondecreasing}
  For $i \in \D$, $\pastVC_{\cl}[i]$ at any client $\cl$
  is non-decreasing.
\end{applemma}

\begin{proof} \label{proof:pastvc-nondecreasing}
  Note that $\pastVC_{\cl}[i]$ ($i \in \D$) is updated only
  at lines~\code{\ref{alg:unistore-client}}{\ref{line:commitcausaltx-pastvc}}
  or \code{\ref{alg:unistore-client}}{\ref{line:commitstrongtx-pastvc}}
  when some transaction is committed.
  Therefore, the lemma holds due to Lemmas~\ref{lemma:snapshotvc-commitvc}
  and \ref{lemma:start-pastvc-snapshotvc}.
\end{proof}

\subsection{Metadata for Strong Transactions} \label{ss:metadata-strong}

\begin{applemma} \label{lemma:knownvc-strong-nondecreasing}
  For any replica $p^{m}_{d}$ in any data center $d$,
  $\knownVC^{m}_{d}[\strongentry]$ is non-decreasing.
\end{applemma}

\begin{proof} \label{proof:knownvc-strong-nondecreasing}
  By \ref{tcs-requirement:deliver-order}
  and line~\code{\ref{alg:unistore-strong-commit}}{\ref{line:deliverupdates-foreach-wbuff}}.
\end{proof}

\begin{applemma} \label{lemma:stablevc-strong-nondecreasing}
  For any replica $p^{m}_{d}$ in any data center $d$,
  $\stableVC^{m}_{d}[\strongentry]$ is non-decreasing.
\end{applemma}

\begin{proof} \label{proof:stablevc-strong-nondecreasing}
  By Lemma~\ref{lemma:knownvc-strong-nondecreasing},
  Assumption~\ref{assumption:message},
  and line~\code{\ref{alg:unistore-clock}}{\ref{line:knownvclocal-stablevc-strong}}.
\end{proof}

\begin{applemma}[\prop{6}] \label{lemma:knownvc-strong}
  Let $\tvar \in \strongtxs$ be a strong transaction
  that originates at data center $i$
  and accesses partition $m$.
  If
  \[
    \commitVC(\tvar)[\strongentry] \le \knownVC^{m}_{d}[\strongentry]
  \]
  for some replica $p^{m}_{d}$ in data center $d$,
  then
  \[
    \log(\tvar)[m] \subseteq \oplog^{m}_{d}.
  \]
\end{applemma}

\begin{proof} \label{proof:knownvc-strong}
  Note that $\knownVC^{m}_{d}[\strongentry]$ can be updated
  only at line~\code{\ref{alg:unistore-strong-commit}}{\ref{line:deliverupdates-knownvc-strongentry}}.
  By \ref{tcs-requirement:deliver-order},
  all committed strong transactions with strong timestamps
  less than or equal to $\knownVC^{m}_{d}[\strongentry]$
  have been delivered to $p^{m}_{d}$.
  By line~\code{\ref{alg:unistore-strong-commit}}{\ref{line:deliverupdates-oplog}},
  \[
    \log(\tvar)[m] \subseteq \oplog^{m}_{d}.
  \]
\end{proof}

\subsection{Timestamps} \label{ss:ts}

\begin{appdefinition}[Timestamps of Events] \label{def:ts-op}
  Let $e \in C_{\causalentry} \cup C_{\strongentry} \cup \Fence \cup \Attach$
  be an event issued by client $\cl$.
  We define its timestamp $\timestamp(e)$ as
  \[
    \tsfunc(e) \triangleq (\pastVC_{\cl})_{e}.
  \]
  Let $e \in S$ be a \start{} event of transaction $\tvar$.
  Let $d \triangleq \dc(\tvar)$ and $m \triangleq \coord(\tvar)$.
  We define its timestamp $\timestamp(e)$ as
  \[
    \tsfunc(e) \triangleq (\snapVC^{m}_{d})_{e}[\tvar].
  \]
  See lines~\code{\ref{alg:unistore-client}}{\ref{line:starttx-ts}},
  \code{\ref{alg:unistore-client}}{\ref{line:commitcausaltx-ts}},
  \code{\ref{alg:unistore-client}}{\ref{line:commitstrongtx-ts}},
  \code{\ref{alg:unistore-client}}{\ref{line:fence-ts}},
  and \code{\ref{alg:unistore-client}}{\ref{line:clattach-ts}}
  for \start, \commitcausaltx, \commitstrongtx, \fence,
  and \clattach{} events, respectively.
\end{appdefinition}

\begin{appdefinition}[Timestamps of Transactions] \label{def:ts-tx}
  The timestamp $\timestamp(\tvar)$ of a transaction $\tvar$
  is that of its commit event, i.e.,
  \[
    \forall \tvar \in \txs.\;
      \timestamp(\tvar) \triangleq \timestamp(\commitoftx(\tvar)).
  \]
\end{appdefinition}

\begin{applemma} \label{lemma:ts-start}
  Let $e \in S$ be a \start{} event.
  Let $d \triangleq \dc(\txfunc(e))$ and $m \triangleq \coord(\txfunc(e))$.
  Then
  \begin{align*}
    (\forall i \in \D.\; \timestamp(e)[i] \ge (\uniformVC^{m}_{d})_{e}[i])
    \;\land \\
    \timestamp(e)[\strongentry] \ge (\stableVC^{m}_{d})_{e}[\strongentry].
  \end{align*}
\end{applemma}

\begin{proof} \label{proof:ts-start}
  By lines~\code{\ref{alg:unistore-coord}}{\ref{line:start-snapvc}},
  \code{\ref{alg:unistore-coord}}{\ref{line:start-snapvc-d}},
  and \code{\ref{alg:unistore-coord}}{\ref{line:start-snapvc-strong}}.
\end{proof}

\begin{applemma} \label{lemma:ts-extread}
  Let $e \in \extread$ be an external read event.
  Let $d \triangleq \dc(\txfunc(e))$ and $m \triangleq \coord(\txfunc(e))$.
  Then
  \begin{align*}
    \timestamp(\startoftx(e)) &= \snapshotVC(\txfunc(e)) \\
                  &= \snapvc_{(\readkey, e)} \\
                  &= (\snapVC^{m}_{d})_{\startoftx(e)}[\txfunc(e)].
  \end{align*}
\end{applemma}

\begin{proof} \label{proof:ts-extread}
  By Definition~\ref{def:ts-op} of timestamps,
  line~\code{\ref{alg:unistore-coord}}{\ref{line:doread-from-snapshot}},
  and line~\code{\ref{alg:unistore-coord}}{\ref{line:start-snapvc-d}}.
\end{proof}

\begin{applemma} \label{lemma:rf-ts}
  Let $e \in \extread$ be an external read event
  which reads from transaction $\tvar$. Then
  \[
    \timestamp(\tvar) \le \timestamp(\startoftx(e)).
  \]
\end{applemma}

\begin{proof} \label{proof:rf-ts}
  By line~\code{\ref{alg:unistore-replica}}{\ref{line:readkey-read}},
  \[
    \commitvc_{(\readkey, e)} \le \snapvc_{(\readkey, e)}.
  \]
  Since $e$ reads from $\tvar$,
  \[
    \timestamp(\tvar) = \commitvc_{(\readkey, e)}.
  \]
  By Lemma~\ref{lemma:ts-extread},
  \[
    \timestamp(\startoftx(e)) = \snapvc_{(\readkey, e)}.
  \]
  Therefore,
  \[
    \timestamp(\tvar) \le \timestamp(\startoftx(e)).
  \]
\end{proof}

\begin{applemma} \label{lemma:ts-commit}
  \begin{align*}
    &\forall e \in (\updatecommit \cap C_{\causalentry})
      \cup C_{\strongentry}.\; \\
      &\quad \timestamp(e) = \timestamp(\txfunc(e))
                           = \commitVC(\txfunc(e)).
  \end{align*}
\end{applemma}

\begin{proof} \label{proof:ts-commit}
  By Definition~\ref{def:ts-op} of timestamps
  and Definition~\ref{def:commitvc} of $\commitVC(\txfunc(e))$.
\end{proof}

\begin{applemma} \label{lemma:ts-tid-st-tid}
  \[
    \forall \tvar \in \txs.\; \timestamp(\tvar) \ge \timestamp(\startoftx(\tvar)).
  \]
\end{applemma}

\begin{proof} \label{proof:ts-tid-st-tid}
  By Lemmas~\ref{lemma:snapshotvc-commitvc}, \ref{lemma:ts-extread},
  and \ref{lemma:ts-commit}.
\end{proof}

\subsection{Session Order}  \label{ss:so}

\begin{applemma} \label{lemma:so-ts}
  \begin{align*}
    &\forall \txevents_{1}, \txevents_{2} \in X.\;
      \txevents_{1} \rel{\so} \txevents_{2} \implies
      \big(\timestamp(\txevents_{1}) \le \timestamp(\txevents_{2}) \\
      &\quad \land \big(\txevents_{2} \in \txs
        \implies \timestamp(\txevents_{1})
          \le \timestamp(\startoftx(\txevents_{2}))
          \le \timestamp(\txevents_{2})\big)\big).
  \end{align*}
\end{applemma}

\begin{proof} \label{proof:so-ts}
  By Definitions~\ref{def:ts-op} and \ref{def:ts-tx} of timestamps
  and Lemma~\ref{lemma:pastvc-nondecreasing},
  \[
    \timestamp(\txevents_{1}) \le \timestamp(\txevents_{2}),
  \]
  and
  \[
    \txevents_{2} \in \txs \implies
      \timestamp(\txevents_{1}) \le \timestamp(\startoftx(\txevents_{2})).
  \]
  Besides, by Lemma~\ref{lemma:snapshotvc-commitvc},
  \[
    \txevents_{2} \in \txs \implies
      \timestamp(\startoftx(\txevents_{2})) \le \timestamp(\txevents_{2}).
  \]
  Therefore,
  \[
    \txevents_{2} \in \txs \implies
      \timestamp(\txevents_{1}) \le \timestamp(\startoftx(\txevents_{2}))
      \le \timestamp(\txevents_{2}).
  \]
\end{proof}

\subsection{Lamport Clocks} \label{ss:lc}

\begin{appdefinition}[Lamport Clocks of Events] \label{def:lc-op}
  Let $e \in C_{\causalentry} \cup C_{\strongentry} \cup \Fence \cup \Attach$
  be an event issued by client $\cl$.
  We define its Lamport clock $\lclock(e)$ as
  \[
    \lclock(e) \triangleq (\lc_{\cl})_{e}.
  \]
  See lines~\code{\ref{alg:unistore-client}}{\ref{line:commitcausaltx-lc}},
  \code{\ref{alg:unistore-client}}{\ref{line:commitstrongtx-lc}},
  \code{\ref{alg:unistore-client}}{\ref{line:fence-lc}},
  and \code{\ref{alg:unistore-client}}{\ref{line:clattach-lc}}
  for \commitcausaltx, \commitstrongtx, \fence, and \clattach{} events, respectively.
\end{appdefinition}

\begin{appdefinition}[Lamport Clocks of Transactions] \label{def:lc-tx}
  The Lamport clock $\lclock(\tvar)$ of a transaction $\tvar$
  is that of its commit event, i.e.,
  \[
    \forall \tvar \in \txs.\;
      \lclock(\tvar) \triangleq \lclock(\commitoftx(\tvar)).
  \]
\end{appdefinition}

\begin{applemma} \label{lemma:lc-extread-commit}
  Let $e \in \extread$ be an external read event issued by client $\cl$.
  Then
  \[
    (\lc_{\cl})_{e} < \lclock(\txfunc(e)).
  \]
\end{applemma}

\begin{proof} \label{proof:lc-extread-commit}
  If $\txfunc(e)$ is a causal transaction,
  by line~\code{\ref{alg:unistore-client}}{\ref{line:commitcausaltx-lc}},
  \[
    \lclock(e) < \lclock(\commitoftx(e)) = \lclock(\txfunc(e)).
  \]
  If $\txfunc(e)$ is a strong transaction,
  by line~\code{\ref{alg:unistore-client}}{\ref{line:commitstrongtx-lc-so}}
  and (\ref{eqn:gcf-lc}),
  \[
    \lclock(e) < \lclock(\commitoftx(e)) = \lclock(\txfunc(e)).
  \]
\end{proof}

\begin{appdefinition}[Lamport Clock Order] \label{def:lco}
  The Lamport clock order $\lcorder$ on $X$
  is the total order defined by their Lamport clocks,
  with their client identifiers for tie-breaking.
\end{appdefinition}

\begin{applemma} \label{lemma:so-lc}
  \[
    \so \subseteq \lcorder.
  \]
\end{applemma}

\begin{proof} \label{proof:so-lc}
  By lines~\code{\ref{alg:unistore-client}}{\ref{line:commitcausaltx-lc}},
  \code{\ref{alg:unistore-client}}{\ref{line:commitstrongtx-lc-so}},
  (\ref{eqn:gcf-lc}),
  \code{\ref{alg:unistore-client}}{\ref{line:commitstrongtx-lc}},
  \code{\ref{alg:unistore-client}}{\ref{line:fence-lc}},
  and \code{\ref{alg:unistore-client}}{\ref{line:clattach-lc}}.
\end{proof}

\begin{applemma} \label{lemma:rf-lc}
  Let $e \in \extread$ be an external read event
  which reads from transaction $\tvar$. Then
  \[
    \tvar \rel{\lcorder} \txfunc(e).
  \]
\end{applemma}

\begin{proof} \label{proof:rf-lc}
  Suppose that $e$ is issued by client $\cl$.
  By line~\code{\ref{alg:unistore-client}}{\ref{line:read-lc}},
  \[
    \lclock(\tvar) \le (\lc_{\cl})_{e}.
  \]
  By Lemma~\ref{lemma:lc-extread-commit},
  \[
    (\lc_{\cl})_{e} < \lclock(\txfunc(e)).
  \]
  Therefore,
  \[
    \lclock(\tvar) < \lclock(\txfunc(e)).
  \]
  By Definition~\ref{def:lco} of $\lcorder$,
  \[
    \tvar \rel{\lcorder} \txfunc(e).
  \]
\end{proof}

\subsection{Visibility Relation}  \label{ss:vis}

\begin{appdefinition}[Visibility Relation] \label{def:vis-tx}
  \begin{align*}
    &\forall \txevents_{1}, \txevents_{2} \in X.\;
      \txevents_{1} \rel{\vis} \txevents_{2} \iff \\
      &\quad \big((\txevents_{2} \in \txs \implies
        \timestamp(\txevents_{1}) \le \timestamp(\startoftx(\txevents_{2}))) \\
        &\quad\quad \land (\txevents_{2} \in \Fence \cup \Attach \implies \timestamp(\txevents_{1}) \le \timestamp(\txevents_{2}))\big)
          \land \txevents_{1} \rel{\lcorder} \txevents_{2}.
  \end{align*}
\end{appdefinition}

\begin{apptheorem} \label{thm:conflictaxiom}
  \[
    A \models \conflictaxiom.
  \]
\end{apptheorem}

\begin{proof} \label{proof:conflictaxiom}
  We need to show that
  \[
    \forall \tvar_{1}, \tvar_{2} \in \strongtxs.\;
      \tvar_{1} \conflict \tvar_{2} \implies \tvar_{1} \rel{\vis} \tvar_{2} \lor \tvar_{2} \rel{\vis} \tvar_{1}.
  \]
  Consider the history $h$ of TCS.
  By Theorem~\ref{thm:tcs-correctness},
  $h \mid \committedVar(h)$ has a legal permutation $\pi$.
  Suppose that
  \[
    \intcertify(\tidselector(\tvar_{1}), \_, \_, \_, \_) \prec_{\pi}
      \intcertify(\tidselector(\tvar_{2}), \_, \_, \_, \_).
  \]
  Since $\tvar_{1} \conflict \tvar_{2}$ and $\tvar_{2}$ is committed,
  by (\ref{eqn:gcf-decision}),
  \[
    \commitVC(\tvar_{1}) \le \snapshotVC(\tvar_{2}).
  \]
  By Lemmas~\ref{lemma:ts-extread} and \ref{lemma:ts-commit},
  \[
    \timestamp(\tvar_{1}) \le \timestamp(\startoftx(\tvar_{2})).
  \]
  On the other hand, by (\ref{eqn:gcf-lc}),
  \[
    \lclock(\tvar_{1}) < \lclock(\tvar_{2}).
  \]
  By Definition~\ref{def:lco} of $\lcorder$,
  \[
    \tvar_{1} \rel{\lcorder} \tvar_{2}.
  \]
  Therefore, by Definition~\ref{def:vis-tx} of $\vis$,
  \[
    \tvar_{1} \rel{\vis} \tvar_{2}.
  \]
\end{proof}

\begin{applemma} \label{lemma:so-vis}
  \[
    \so \subseteq \vis.
  \]
\end{applemma}

\begin{proof} \label{proof:so-vis}
  By Lemmas~\ref{lemma:so-ts} and \ref{lemma:so-lc}.
\end{proof}

\begin{applemma} \label{lemma:vis-partial}
  The visibility relation $\vis$ is a partial order.
\end{applemma}

\begin{proof} \label{proof:vis-partial}
  We need to show that
  \begin{itemize}
    \item $\vis$ is irreflexive.
      This holds because $\lcorder$ is irreflexive.
    \item $\vis$ is transitive.
      Suppose that $\txevents_{1} \rel{\vis} \txevents_{2} \rel{\vis} \txevents_{3}$.
      By Definition~\ref{def:vis-tx} of $\vis$,
      \[
        \txevents_{1} \rel{\lcorder} \txevents_{2} \rel{\lcorder} \txevents_{3}.
      \]
      By Definition~\ref{def:lco} of $\lcorder$,
      \[
        \txevents_{1} \rel{\lcorder} \txevents_{3}.
      \]
      Regarding timestamps,
      we distinguish between the following four cases
      and use Lemma~\ref{lemma:ts-tid-st-tid}.
      \begin{itemize}
        \item $\txevents_{2} \in \Fence \cup \Attach
          \land \txevents_{3} \in \Fence \cup \Attach$.
          \[
            \timestamp(\txevents_{1}) \le \timestamp(\txevents_{2})
              \le \timestamp(\txevents_{3}).
          \]
        \item $\txevents_{2} \in \Fence \cup \Attach
          \land \txevents_{3} \in \txs$.
          \[
            \timestamp(\txevents_{1}) \le \timestamp(\txevents_{2})
              \le \timestamp(\startoftx(\txevents_{3})).
          \]
        \item $\txevents_{2} \in \txs \land \txevents_{3} \in \Fence \cup \Attach$.
          \[
            \timestamp(\txevents_{1}) \le \timestamp(\startoftx(\txevents_{2}))
              \le \timestamp(\txevents_{2}) \le \timestamp(\txevents_{3}).
          \]
        \item $\txevents_{2} \in \txs \land \txevents_{3} \in \txs$.
          \[
            \timestamp(\txevents_{1}) \le \timestamp(\startoftx(\txevents_{2}))
              \le \timestamp(\txevents_{2}) \le \timestamp(\startoftx(\txevents_{3})).
          \]
      \end{itemize}
      By Definition~\ref{def:vis-tx} of $\vis$,
      \[
        \txevents_{1} \rel{\vis} \txevents_{3}.
      \]
  \end{itemize}
\end{proof}

\begin{apptheorem} \label{thm:cv}
  \[
    A \models \cv.
  \]
\end{apptheorem}

\begin{proof} \label{proof:cv}
  By Lemmas~\ref{lemma:so-vis} and \ref{lemma:vis-partial},
  \[
    (\so \cup \vis)^{+} = \vis^{+} = \vis.
  \]
\end{proof}

\begin{applemma}[\prop{5}] \label{lemma:conflict-strong-ts}
  For any two conflicting transactions $\tvar_{1}$ and $\tvar_{2}$,
  \begin{align*}
    &\tvar_{1} \rel{\vis} \tvar_{2} \iff \\
      &\quad \commitVC(\tvar_{1})[\strongentry] < \commitVC(\tvar_{2})[\strongentry].
  \end{align*}
\end{applemma}

\begin{proof} \label{proof:conflict-strong-ts}
  We first show that
  \begin{align}
    &\tvar_{1} \rel{\vis} \tvar_{2} \implies \label{eq:vis-strong-ts} \\
      &\quad \commitVC(\tvar_{1})[\strongentry] < \commitVC(\tvar_{2})[\strongentry].
      \nonumber
  \end{align}
  Assume that $\tvar_{1} \rel{\vis} \tvar_{2}$.
  By Definition~\ref{def:vis-tx} of $\vis$,
  \[
    \tsfunc(\tvar_{1}) \le \tsfunc(\startoftx(\tvar_{2})).
  \]
  By Lemmas~\ref{lemma:ts-extread} and \ref{lemma:ts-commit},
  \[
    \commitVC(\tvar_{1}) \le \snapshotVC(\tvar_{2}).
  \]
  Therefore,
  \[
    \commitVC(\tvar_{1})[\strongentry] \le \snapshotVC(\tvar_{2})[\strongentry].
  \]
  By (\ref{eqn:gcf-commitvc}),
  \[
    \commitVC(\tvar_{2})[\strongentry] > \snapshotVC(\tvar_{2})[\strongentry].
  \]
  Putting it together yields
  \[
    \commitVC(\tvar_{1})[\strongentry] < \commitVC(\tvar_{2})[\strongentry].
  \]
  Next we show that
  \begin{align*}
    &\tvar_{1} \rel{\vis} \tvar_{2} \impliedby \\
      &\quad \commitVC(\tvar_{1})[\strongentry] < \commitVC(\tvar_{2})[\strongentry].
  \end{align*}
  Assume that
  \begin{align}
    \commitVC(\tvar_{1})[\strongentry] < \commitVC(\tvar_{2})[\strongentry].
    \label{eq:tid1-less-tid2-strong-ts}
  \end{align}
  Since $\tvar_{1} \conflict \tvar_{2}$, by Theorem~\ref{thm:conflictaxiom},
  \[
    \tvar_{1} \rel{\vis} \tvar_{2} \lor \tvar_{2} \rel{\vis} \tvar_{1}.
  \]
  By (\ref{eq:vis-strong-ts}) and (\ref{eq:tid1-less-tid2-strong-ts}),
  \[
    \lnot(\tvar_{2} \rel{\vis} \tvar_{1}).
  \]
  Therefore,
  \[
    \tvar_{1} \rel{\vis} \tvar_{2}.
  \]
\end{proof}

\subsection{Execution Order} \label{ss:eo}

\begin{appdefinition}[Execution Points] \label{def:ep}
  Let $k$ be a key.
  The ``execution point'' $\ep(e, k)$ of event
  $e \in (\extread \cap R_{k}) \cup C_{k}$ is defined as follows:

  \begin{itemize}
    \item If $e \in \extread \cap R_{k}$,
      then $\ep(e, k)$ is at
      line~\code{\ref{alg:unistore-replica}}{\ref{line:readkey-read}};
    \item If $e \in C_{k} \cap C_{\causalentry}$,
      then $\ep(e, k)$ is at
      line~\code{\ref{alg:unistore-replica}}{\ref{line:commit-oplog}}
      for this particular key $k$;
    \item If $e \in C_{k} \cap C_{\strongentry}$
      then $\ep(e, k)$ is at
      line~\code{\ref{alg:unistore-strong-commit}}{\ref{line:deliverupdates-oplog}}
      for delivery of the update of $\txfunc(e)$ on this particular key $k$.
      Note that \deliver{} is asynchronous with the commit event $e$.
  \end{itemize}
\end{appdefinition}

\begin{appdefinition}[Per-key Execution Order] \label{def:perkey-eo}
  Let $k$ be a key.
  Suppose that $\set{e_{1}, e_{2}} \subseteq (\extread \cap R_{k}) \cup C_{k}$.
  Event $e_1$ is executed before event $e_2$, denoted $e_1 \rel{\eok} e_2$,
  if $\ep(e_{1}, k)$ is executed before $\ep(e_{2}, k)$ in real time.
\end{appdefinition}

\begin{applemma} \label{lemma:vis-perkey-eo}
  Let $k \in \Key$ be a key, $\tvar \in \txs_{k}$ be a transaction,
  and $e \in \extread \cap R_{k}$ be an external read event.
  Suppose that $d \triangleq \dc(\tvar) = \dc(\txfunc(e))$.
  Then
  \[
    \tvar \rel{\vis} \txfunc(e) \implies \commitoftx(\tvar) \rel{\eok} e.
  \]
\end{applemma}

\begin{proof} \label{proof:vis-perkey-eo}
  By Definition~\ref{def:vis-tx} of $\vis$,
  \[
    \timestamp(\tvar) \le \timestamp(\startoftx(e)).
  \]
  Since $e \in \extread$, by Lemma~\ref{lemma:ts-extread},
  \[
    \timestamp(\tvar) \le \snapvc_{(\readkey, e)}.
  \]
  In the following, we distinguish between two cases
  according to whether $\tvar \in \causaltxs$
  or $\tvar \in \strongtxs$.
  Let $m \triangleq \partitionofproc(k)$.
  \begin{itemize}
    \item $\textsc{Case I}$: $\tvar \in \causaltxs$.
      By Lemma~\ref{lemma:ts-commit},
      \[
        \tsfunc(\tvar) = \commitVC(\tvar) \le \snapvc_{(\readkey, e)}.
      \]
      Therefore, after line~\code{\ref{alg:unistore-replica}}{\ref{line:readkey-wait-util-knownvc}}
      for $e$,
      \begin{align}
        (\knownVC^{m}_{d})_{e}[d]
        &\ge \snapvc_{(\readkey, e)}[d] \notag \\
        &\ge \commitVC(\tvar)[d].
        \label{eqn:vis-perkey-eo-knownvc-commitvc-causal}
      \end{align}
      By Lemma~\ref{lemma:knownvc-local-d},
      \commit{} of Algorithm~\ref{alg:unistore-replica}
      for $\ws(\tvar)[m] \ni \langle k, \_ \rangle$
      finishes before $e$ starts at replica $p^{m}_{d}$.
      By Definition~\ref{def:perkey-eo} of $\eok$,
      \[
        \commitoftx(\tvar) \rel{\eok} e.
      \]
    \item $\textsc{Case II}$: $\tvar \in \strongtxs$.
      By Lemma~\ref{lemma:ts-commit},
      \[
        \tsfunc(\tvar) = \commitVC(\tvar) \le \snapvc_{(\readkey, e)}.
      \]
      Therefore, after line~\code{\ref{alg:unistore-replica}}{\ref{line:readkey-wait-util-knownvc}}
      for $e$,
      \begin{align}
        (\knownVC^{m}_{d})_{e}[\strongentry]
        &\ge \snapvc_{(\readkey, e)}[\strongentry] \notag \\
        &\ge \commitVC(\tvar)[\strongentry].
        \label{eqn:vis-perkey-eo-knownvc-commitvc-strong}
      \end{align}
      By Lemma~\ref{lemma:knownvc-strong},
      \deliver{} of Algorithm~\ref{alg:unistore-strong-commit}
      for $\ws(\tvar)[m] \ni \langle k, \_ \rangle$
      finishes before $e$ starts at replica $p^{m}_{d}$.
      By Definition~\ref{def:perkey-eo} of $\eok$,
      \[
        \commitoftx(\tvar) \rel{\eok} e.
      \]
  \end{itemize}
\end{proof}

\subsection{Arbitration Relation}  \label{ss:ar}

\begin{appdefinition}[Arbitration Relation] \label{def:ar}
  We define the arbitration relation $\ar$ on $X$
  as the Lamport clock order between them, i.e.,
  \[
    \ar = \lcorder.
  \]
\end{appdefinition}

\begin{apptheorem} \label{thm:ca}
  \[
    A \models \ca.
  \]
\end{apptheorem}

\begin{proof} \label{proof:ca}
  By Definition~\ref{def:vis-tx} of $\vis$ and Definition~\ref{def:ar} of $\ar$,
  \begin{align*}
    \vis \subseteq \lcorder = \ar.
  \end{align*}
\end{proof}

\subsection{Return Values} \label{ss:rval}

It is straightforward to show that $\intretval$ holds
for \emph{internal} read events.
\begin{apptheorem} \label{thm:intretval}
  \[
    A \models \intretval.
  \]
\end{apptheorem}

\begin{proof} \label{proof:intretval}
  Let $e \in \intread \cap R_{k}$ be an internal read event.
  The transaction $\txfunc(e)$ contains update events on $k$.
  By line~\code{\ref{alg:unistore-coord}}{\ref{line:doread-return-from-buffer}},
  $e$ reads from the last update event on $k$
  preceding $e$ in $\txfunc(e)$.
\end{proof}

Now let $e$ be an \emph{external} read event.
For notational convenience,
we define $V_{e}$ to be the set of update transactions on $k$
that are visible to $\txfunc(e)$,
and $S_{e}$ the set of update transactions on $k$ that are safe to read
at line~\code{\ref{alg:unistore-replica}}{\ref{line:readkey-read}}.
By Assumption~\ref{assumption:complete-execution},
\ref{tcs-requirement:certify-before-deliver}, and
\ref{tcs-requirement:abort-cannot-deliver},
all transactions in $S_{e}$ are committed.
Formally,

\begin{appdefinition}[Visibility Set]
  \label{def:visible-set-tx}
  Let $e \in \extread \cap R_{k}$ be an external read event on key $k$.
  \[
    V_{e} \triangleq \vis^{-1}(\txfunc(e)) \cap \txs_{k}.
  \]
\end{appdefinition}

\begin{appdefinition}[Safe Set]
  \label{def:safe-set-tx}
  Let $e \in \extread \cap R_{k}$ be an external read event on key $k$.
  Suppose that $e$ is issued to replica $p^{m}_{d}$ in data center $d$.
  \begin{align*}
    S_{e} \triangleq \set{\tvar \in \txs_{k}:
      &\; \timestamp(\tvar) \le \snapvc_{(\readkey, e)}\; \land \\
      &\; \log[\tvar][k] \in (\oplog^{m}_{d})_{e}[k]}.
  \end{align*}
\end{appdefinition}

\begin{applemma} \label{lemma:visible-safe-tx}
  Let $e \in \extread \cap R_{k}$ be an external read event on key $k$.
  Suppose that $e$ is issued to replica $p^{m}_{d}$ in data center $d$.
  When $e$ returns at $p^{m}_{d}$
  (line~\code{\ref{alg:unistore-replica}}{\ref{line:readkey-read}}),
  we have
  \[
    V_{e} \subseteq S_{e}.
  \]
\end{applemma}

\begin{proof}  \label{proof:visible-safe-tx}
  For each $\tvar \in V_{e}$,
  we need to show that $\tvar \in S_{e}$.
  That is,
  \begin{gather}
    \timestamp(\tvar) \le \snapvc_{(\readkey, e)}
    \label{eqn:visible-safe-tx-snapshotvc}
  \end{gather}
  and
  \begin{gather}
    \log[\tvar][k] \in (\oplog^{m}_{d})_{e}[k].
    \label{eqn:visible-safe-tx-oplog}
  \end{gather}

  We first show that (\ref{eqn:visible-safe-tx-snapshotvc}) holds.
  Since $\tvar \in V_{e}$,
  \[
    \tvar \rel{\vis} \txfunc(e).
  \]
  By Definition~\ref{def:vis-tx} of $\vis$,
  \[
    \timestamp(\tvar) \le \timestamp(\startoftx(e)).
  \]
  By Lemma~\ref{lemma:ts-extread},
  \[
    \timestamp(\tvar) \le \snapvc_{(\readkey, e)}.
  \]

  To show that (\ref{eqn:visible-safe-tx-oplog}) holds,
  we perform a case analysis according to
  whether $\tvar$ is a local transaction in data center $d$
  or a remote one in data center $i \neq d$.

  \begin{itemize}
    \item $\textsc{Case I}$: $\tvar$ is a local transaction in data center $d$.
      Since $\tvar \rel{\vis} \txfunc(e)$,
      by Lemma~\ref{lemma:vis-perkey-eo},
      \[
        \commitoftx(\tvar) \rel{\eok} e.
      \]
      Therefore,
      \[
        \log[\tvar][k] \in (\oplog^{m}_{d})_{e}[k].
      \]
    \item $\textsc{Case II}$: $\tvar$ is a remote transaction
      in data center $i \neq d$.
      We distinguish between two cases
      according to whether $\tvar \in \causaltxs$ or $\tvar \in \strongtxs$.
      \begin{itemize}
        \item $\textsc{Case I}$: $\tvar \in \causaltxs$.
          Since $i \neq d$,
          by line~\code{\ref{alg:unistore-replica}}{\ref{line:readkey-uniformvc}},
          \[
            \snapvc_{(\readkey, e)}[i] \le (\uniformVC^{m}_{d})_{e}[i].
          \]
          By (\ref{eqn:visible-safe-tx-snapshotvc}),
          \[
            \timestamp(\tvar)[i] \le (\uniformVC^{m}_{d})_{e}[i].
          \]
          By Lemma~\ref{lemma:replication-uniformvc},
          \[
            \log[\tvar][k] \in (\oplog^{m}_{d})_{e}[k].
          \]
        \item $\textsc{Case II}$: $\tvar \in \strongtxs$.
          By line~\code{\ref{alg:unistore-replica}}{\ref{line:readkey-wait-util-knownvc}},
          \begin{align*}
            &\snapvc_{(\readkey, e)}[\strongentry] \\
              &\qquad \le (\knownVC^{m}_{d})_{e}[\strongentry].
          \end{align*}
          By (\ref{eqn:visible-safe-tx-snapshotvc}),
          \[
            \timestamp(\tvar)[\strongentry] \le (\knownVC^{m}_{d})_{e}[\strongentry].
          \]
          By Lemma~\ref{lemma:knownvc-strong},
          \[
            \log[\tvar][k] \in (\oplog^{m}_{d})_{e}[k].
          \]
      \end{itemize}
  \end{itemize}
\end{proof}

\begin{apptheorem} \label{thm:extretval}
  \[
    A \models \extretval.
  \]
\end{apptheorem}

\begin{proof} \label{proof:extretval}
  Let $e \in \extread \cap R_{k}$
  be an external read event on key $k$.
  Suppose that $e$ reads from transaction $\tvar$ in $S_{e}$.
  Since all transactions in $S_{e}$ are committed,
  $\tvar$ is committed.
  By Lemma~\ref{lemma:rf-ts},
  \[
    \timestamp(\tvar) \le \timestamp(\startoftx(e)).
  \]
  By Lemma~\ref{lemma:rf-lc},
  \[
    \tvar \rel{\lcorder} \txfunc(e).
  \]
  By Definition~\ref{def:vis-tx} of $\vis$,
  \[
    \tvar \rel{\vis} \txfunc(e).
  \]
  By Definition~\ref{def:visible-set-tx} of $V_{e}$,
  \[
    \tvar \in V_{e}.
  \]

  Both $V_{e}$ and $S_{e}$
  are totally ordered by $\lcorder$.
  Since $\tvar$ is the latest one in $S_{e}$
  and $V_{e} \subseteq S_{e}$
  (Theorem~\ref{lemma:visible-safe-tx}),
  $\tvar$ is also the latest one in $V_{e}$.
  Thus, $e$ reads from $\tvar$ in $V_{e}$.
  That is, $e$ reads from the update event
  $\ud(\tvar, k)$ of $V_{e}$.
\end{proof}

\begin{apptheorem} \label{thm:retval}
  \[
    A \models \retval.
  \]
\end{apptheorem}

\begin{proof} \label{proof:retval}
 By Theorems~\ref{thm:intretval} and \ref{thm:extretval}.
\end{proof}

\subsection{Uniformity} \label{ss:uniformity}

\subsubsection{Uniformity of Causal Transactions Originating at Correct Data Centers}
\label{sss:uniformity-correct-dc}

\begin{applemma} \label{lemma:knownvc-d-no-bound}
  For any replica $p^{m}_{d}$ in any correct data center $d \in \C$,
  $\knownVC^{m}_{d}[d]$ grows without bound.
\end{applemma}

\begin{proof} \label{proof:knownvc-d-no-bound}
  Since data center $d$ is correct, by Assumption~\ref{assumption:fairness},
  \propagate{} of Algorithm~\ref{alg:unistore-replication}
  will be executed infinitely often.
  \begin{itemize}
    \item $\textsc{Case I}$:
      Line~\code{\ref{alg:unistore-replication}}{\ref{line:propagate-knownvc-clock}}
      is executed infinitely often.
      By Assumption~\ref{assumption:clock},
      $\knownVC^{m}_{d}[d]$ grows without bound.
    \item $\textsc{Case II}$:
      Line~\code{\ref{alg:unistore-replication}}{\ref{line:propagate-knownvc-ts}}
      is executed infinitely often.
      That is, it is infinitely often that
      \[
        \preparedcausal^{m}_{d} \neq \emptyset.
      \]
      By Assumption~\ref{assumption:fairness},
      causal transactions in $\preparedcausal^{m}_{d}$
      will eventually be committed and removed from $\preparedcausal^{m}_{d}$
      (line~\code{\ref{alg:unistore-replica}}{\ref{line:commit-preparedcausal}}).
      Thus, it is infinitely often that new causal transactions
      are prepared and added into $\preparedcausal^{m}_{d}$
      (line~\code{\ref{alg:unistore-replica}}{\ref{line:preparecausal-preparedcausal}})
      with larger and larger prepare timestamps
      (line~\code{\ref{alg:unistore-replica}}{\ref{line:preparecausal-ts}}).
      Therefore,
      \[
        \min\set{\tsvar \mid \langle \_, \_, \tsvar \rangle \in \preparedcausal^{m}_{d}}
      \]
      and
      \begin{align*}
        &\knownVC^{m}_{d}[d] \\
          &\quad = \min\set{\tsvar \mid \langle \_, \_, \tsvar \rangle \in \preparedcausal^{m}_{d}} - 1
      \end{align*}
      grow without bound.
    \end{itemize}
\end{proof}

\begin{applemma} \label{lemma:knownvc-x}
  Let $p^{m}_{d}$ be a replica in a correct data center $d \in \C$.
  If for some $j \in \D$ and some value $x \in \mathbb{N}$
  \[
    \knownVC^{m}_{d}[j] \ge x,
  \]
  then eventually
  \[
    \forall \cdrange \in \C.\; \knownVC^{m}_{\cdrange}[j] \ge x.
  \]
\end{applemma}

\begin{proof} \label{proof:knownvc-x}
  Since data center $d$ is correct, by Assumption~\ref{assumption:fairness},
  for each other data center $i \neq d$, replica $p^{m}_{d}$ will keep
  \begin{itemize}
    \item \emph{replicating} to data center $i$ the write sets
      \[
        \langle \_, \wbuffvar, \commitvc, \_ \rangle \in \committedcausal^{m}_{d}[j]
      \]
      that have not been received by $i$ from the perspective of $d$
      ($\commitvc[d] \le \knownVC^{m}_{d}[d]$
      at line~\code{\ref{alg:unistore-replication}}{\ref{line:propagate-txs}}
      and $\commitvc[j] > \globalmatrix^{m}_{d}[i][j]$ at
      line~\code{\ref{alg:unistore-replication}}{\ref{line:forward-txs}});
    \item or sending \emph{heartbeats} with up-to-date
      $\knownVC^{m}_{d}[j]$ to data center $i$
      (lines~\code{\ref{alg:unistore-replication}}{\ref{line:propagate-call-heartbeat}}
      and \code{\ref{alg:unistore-replication}}{\ref{line:forward-call-heartbeat}}).
  \end{itemize}
  By Assumption~\ref{assumption:message},
  $\knownVC^{m}_{\cdrange}[j]$ at replica $p^{m}_{\cdrange}$ of
  each correct data center $\cdrange \in \C$ will eventually be updated
  (lines~\code{\ref{alg:unistore-replication}}{\ref{line:replicate-knownvc}}
  and \code{\ref{alg:unistore-replication}}{\ref{line:heartbeat-knownvc}})
  such that
  \[
    \knownVC^{m}_{\cdrange}[j] \ge \knownVC^{m}_{d}[j] \ge x.
  \]
\end{proof}

\begin{applemma} \label{lemma:uniformvc-c-no-bound}
  Let $d \in \C$ be a correct data center.
  For any replica $p^{m}_{\cdrange}$
  in any correct data center $\cdrange \in \C$,
  $\uniformVC^{m}_{\cdrange}[d]$ grows without bound.
\end{applemma}

\begin{proof} \label{proof:uniformvc-c-no-bound}
  By Lemmas~\ref{lemma:knownvc-d-no-bound} and \ref{lemma:knownvc-x},
  for any replica $p^{m}_{\cdrange}$
  in any correct data center $\cdrange \in \C$,
  $\knownVC^{m}_{\cdrange}[d]$ grows without bound.
  By lines~\code{\ref{alg:unistore-clock}}{\ref{line:bcast-call-knownvclocal}}
  and \code{\ref{alg:unistore-clock}}{\ref{line:knownvclocal-stablevc-causal}},
  for any replica $p^{m}_{\cdrange}$
  in any correct data center $\cdrange \in \C$,
  $\stableVC^{m}_{\cdrange}[d]$ grows without bound.
  By line~\code{\ref{alg:unistore-clock}}{\ref{line:bcast-call-stablevc}},
  Assumptions~\ref{assumption:message} and \ref{assumption:failure-model},
  and lines~\code{\ref{alg:unistore-clock}}{\ref{line:stablevc-g}}--\code{
    \ref{alg:unistore-clock}}{\ref{line:stablevc-uniformvc}},
  for any replica $p^{m}_{\cdrange}$
  in any correct data center $\cdrange \in \C$,
  $\uniformVC^{m}_{\cdrange}[d]$ grows without bound.
\end{proof}

\begin{applemma}[\prop{4}] \label{lemma:uniformvc-x}
  Let $p^{m}_{d}$ be any replica in any data center $d$.
  For any time $\realtime$, there exists some time $\realtime'$ such that
  \begin{align*}
    &\forall i \in \D.\; \forall \cdrange \in \C.\; \forall n \in \P. \\
      &\quad \uniformVC^{n}_{\cdrange}(\realtime')[i] \ge \uniformVC^{m}_{d}(\realtime)[i].
  \end{align*}
\end{applemma}

\begin{proof} \label{proof:uniformvc-x}
  By Lemma~\ref{lemma:uniformvc-knownvc-f+1},
  Assumption~\ref{assumption:failure-model}, and the fact that at most
  $f$ data centers may fail,
  \begin{align*}
    &\forall i \in \D.\; \exists \cdrange \in \C.\; \forall n \in \P. \\
      &\quad \uniformVC^{m}_{d}(\realtime)[i] \le \knownVC^{n}_{\cdrange}(\realtime)[i].
  \end{align*}
  By Lemma~\ref{lemma:knownvc-x},
  there exists some time $\realtime''$ such that
  \begin{align*}
    &\forall i \in \D.\; \forall \cdrange \in \C.\; \forall n \in \P. \\
      &\quad \uniformVC^{m}_{d}(\realtime)[i] \le \knownVC^{n}_{\cdrange}(\realtime'')[i].
  \end{align*}
  By Algorithm~\ref{alg:unistore-clock} and Assumption~\ref{assumption:failure-model},
  there exists some time $\realtime'$ such that
  \begin{align*}
    &\forall i \in \D.\; \forall \cdrange \in \C.\; \forall n \le \P. \\
      &\quad \uniformVC^{n}_{\cdrange}(\realtime')[i] \ge \uniformVC^{m}_{d}(\realtime)[i].
  \end{align*}
\end{proof}

\begin{applemma} \label{lemma:uniformity-correct-dc}
  Let $d \in \C$ be a correct data center
  and $\tvar \in \causaltxs$ be a causal transaction
  that originates at $d$.
  Then for any replica $p^{m}_{\cdrange}$
  in any correct data center $\cdrange \in \C$, eventually
  \[
    \forall i \in \D.\; \timestamp(\tvar)[i] \le \uniformVC^{m}_{\cdrange}[i].
  \]
\end{applemma}

\begin{proof} \label{proof:uniformity-correct-dc}
  Since $d$ is correct, by Lemma~\ref{lemma:uniformvc-c-no-bound},
  there exists some time $\realtime'$ such that
  \[
    \timestamp(\tvar)[d] \le \uniformVC^{m}_{\cdrange}(\realtime')[d].
  \]
  On the other hand,
  by Definition~\ref{def:ts-op} of timestamps
  and Lemma~\ref{lemma:pastvc-uniformvc-except-d}
  (let $n \triangleq \coord(\tvar)$ and $cl \triangleq \client(\tvar)$),
  \begin{align*}
    \forall i \in \D \setminus \set{d}.\;
      \timestamp(\tvar)[i] &= (\pastVC_{\cl})_{\commitoftx(\tvar)}[i] \\
      &\le (\uniformVC^{n}_{d})_{\commitoftx(\tvar)}[i].
  \end{align*}
  By Lemma~\ref{lemma:uniformvc-x},
  there exists some time $\realtime''$ such that
  \[
    \forall i \in \D \setminus \set{d}.\;
      \timestamp(\tvar)[i] \le (\uniformVC^{m}_{\cdrange})(\realtime'')[i].
  \]
  Let
  \[
    \realtime \triangleq \max\set{\realtime', \realtime''}.
  \]
  By Lemma~\ref{lemma:uniformvc-nondecreasing},
  \[
    \forall i \in \D.\; \timestamp(\tvar)[i] \le \uniformVC^{m}_{\cdrange}(\realtime)[i].
  \]
\end{proof}
\subsubsection{Uniformity of Causal Transactions Visible to \fence{} events}
\label{sss:uniformity-fences}

\begin{applemma} \label{lemma:uniformity-causal-fence}
  Let $\tvar \in \causaltxs$ be a causal transaction
  and $\fencerange \in \Fence$ be a \fence{} event.
  If $\tvar \rel{\vis} \fencerange$,
  then for any replica $p^{m}_{\cdrange}$
  in any correct data center $\cdrange \in \C$, eventually
  \[
    \forall i \in \D.\; \timestamp(\tvar)[i] \le \uniformVC^{m}_{\cdrange}[i].
  \]
\end{applemma}

\begin{proof} \label{proof:uniformity-causal-fence}
  Since $\tvar \rel{\vis} \fencerange$, by Definition~\ref{def:vis-tx} of $\vis$,
  \[
    \timestamp(\tvar) \le \timestamp(\fencerange).
  \]
  Suppose that $\fencerange$ is issued by client $\cl$
  to replica $p^{n}_{d}$ in data center $d$
  and is returned at time $\realtime_{\fencerange}$.
  By Definition~\ref{def:ts-op} of timestamps
  and Lemma~\ref{lemma:pastvc-uniformvc-except-d},
  \[
    \timestamp(\tvar)[d] \le \timestamp(\fencerange)[d] \le (\uniformVC^{n}_{d})_{\fencerange}[d].
  \]
  By Lemma~\ref{lemma:uniformvc-x},
  there exists some time $\realtime'$ such that
  \[
    \timestamp(\tvar)[d] \le (\uniformVC^{m}_{\cdrange})(\realtime')[d].
  \]
  On the other hand, by Definition~\ref{def:ts-op} of timestamps
  and Lemma~\ref{lemma:pastvc-uniformvc},
  \begin{align*}
    \forall i \in \D \setminus \set{d}.\;
      \timestamp(\tvar)[i] &\le \timestamp(\fencerange)[i] \\
      &= (\pastVC_{\cl})_{\fencerange}[i] \\
      &\le \uniformVC^{l}_{d}(\realtime_{\fencerange})[i]
  \end{align*}
  for some replica $p^{l}_{d}$ in data center $d$.
  By Lemma~\ref{lemma:uniformvc-x},
  there exists some time $\realtime''$ such that
  \[
    \forall i \in \D \setminus \set{d}.\;
      \timestamp(\tvar)[i] \le (\uniformVC^{m}_{\cdrange})(\realtime'')[i].
  \]
  Let
  \[
    \realtime \triangleq \max\set{\realtime', \realtime''}.
  \]
  By Lemma~\ref{lemma:uniformvc-nondecreasing},
  \[
    \forall i \in \D.\; \timestamp(\tvar)[i] \le \uniformVC^{m}_{\cdrange}(\realtime)[i].
  \]
\end{proof}
\subsubsection{Uniformity of Strong Transactions}
\label{sss:uniformity-strong}

\begin{applemma} \label{lemma:knownvc-strong-no-bound}
  For any replica $p^{m}_{\cdrange}$
  in any correct data center $\cdrange \in \C$,
  $\knownVC^{m}_{\cdrange}[\strongentry]$ grows without bound.
\end{applemma}

\begin{proof} \label{proof:knownvc-strong-no-bound}
  By Assumption~\ref{assumption:fairness}
  and Theorem~\ref{thm:tcs-correctness},
  replica $p^{m}_{\cdrange}$ will either deliver committed strong transactions
  infinitely often (\deliver{} of Algorithm~\ref{alg:unistore-strong-commit})
  or submit dummy strong transactions infinitely often
  (\heartbeatstrong{} of Algorithm~\ref{alg:unistore-strong-commit}).
  Thus, $\knownVC^{m}_{\cdrange}[\strongentry]$ grows without bound.
\end{proof}

\begin{applemma} \label{lemma:uniformity-strong}
  Let $\tvar \in \txs$ be a transaction.
  Then for any replica $p^{m}_{\cdrange}$
  in any correct data center $\cdrange \in \C$, eventually
  \[
    \timestamp(\tvar)[\strongentry] \le \stableVC^{m}_{\cdrange}[\strongentry].
  \]
\end{applemma}

\begin{proof} \label{proof:uniformity-strong}
  By Lemma~\ref{lemma:knownvc-strong-no-bound}
  and lines~\code{\ref{alg:unistore-clock}}{\ref{line:bcast-call-knownvclocal}}
  and \code{\ref{alg:unistore-clock}}{\ref{line:knownvclocal-stablevc-strong}},
  $\stableVC^{m}_{\cdrange}[\strongentry]$ grows without bound.
  Therefore, there exists some time $\realtime$ such that
  \[
    \timestamp(\tvar)[\strongentry] \le \stableVC^{m}_{\cdrange}(\realtime)[\strongentry].
  \]
\end{proof}

\begin{applemma} \label{lemma:uniformity-strongtx-i}
  Let $\tvar \in \strongtxs$ be a strong transaction.
  Then for any replica $p^{m}_{\cdrange}$
  in any correct data center $\cdrange \in \C$, eventually
  \[
    \forall i \in \D.\; \timestamp(\tvar)[i] \le \uniformVC^{m}_{\cdrange}[i].
  \]
\end{applemma}

\begin{proof} \label{proof:uniformity-strongtx-i}
  Let $d \triangleq \dc(\tvar)$,
  $n \triangleq \coord(\tvar)$, and $cl \triangleq \client(\tvar)$.
  By Definition~\ref{def:ts-op} of timestamps,
  \[
    \timestamp(\tvar) = (\pastVC_{\cl})_{\commitoftx(\tvar)}.
  \]
  On the one hand, by Lemma~\ref{lemma:pastvc-uniformvc-except-d},
  \begin{align*}
    \forall i \in \D \setminus \set{d}.\;
      \timestamp(\tvar)[i] &= (\pastVC_{\cl})_{\commitoftx(\tvar)}[i] \\
      &\le (\uniformVC^{n}_{d})_{\commitoftx(\tvar)}[i].
  \end{align*}
  By Lemma~\ref{lemma:uniformvc-x},
  there exists some time $\realtime'$ such that
  \[
    \forall i \in \D \setminus \set{d}.\;
      \timestamp(\tvar)[i] \le (\uniformVC^{m}_{\cdrange})(\realtime')[i].
  \]
  On the other hand, by Lemma~\ref{lemma:ts-commit},
  \[
    \timestamp(\tvar)[d] = \commitVC(\tvar)[d].
  \]
  By (\ref{eqn:gcf-commitvc}),
  \[
    \commitVC(\tvar)[d] = \snapshotVC(\tvar)[d].
  \]
  By lines~\code{\ref{alg:unistore-strong-commit}}{\ref{line:commitstrong-call-uniformbarrier}}
  and \code{\ref{alg:unistore-replica}}{\ref{line:uniformbarrier-wait-uniformvc-d}},
  \[
    \snapshotVC(\tvar)[d] \le (\uniformVC^{n}_{d})_{\commitoftx(\tvar)}[d].
  \]
  Putting it together yields,
  \[
    \timestamp(\tvar)[d] \le (\uniformVC^{n}_{d})_{\commitoftx(\tvar)}[d].
  \]
  By Lemma~\ref{lemma:uniformvc-x},
  there exists some time $\realtime''$ such that
  \[
    \timestamp(\tvar)[d] \le (\uniformVC^{m}_{\cdrange})(\realtime'')[d].
  \]
  Let
  \[
    \realtime \triangleq \max\set{\realtime', \realtime''}.
  \]
  By Lemma~\ref{lemma:uniformvc-nondecreasing},
  \[
    \forall i \in \D.\; \timestamp(\tvar)[i] \le \uniformVC^{m}_{\cdrange}(\realtime)[i].
  \]
\end{proof}

\subsection{Eventual Visibility} \label{ss:ev}

\begin{apptheorem} \label{thm:ev}
  \[
    A \models \ev.
  \]
\end{apptheorem}

\begin{proof} \label{proof:ev}
  Consider a transaction $\tvar \in \txs$ such that
  \begin{align*}
    \dc(\tvar) \in \C \lor (\exists \fencerange \in \Fence.\; \tvar \rel{\vis} \fencerange)
    \lor \tvar \in \strongtxs.
  \end{align*}
  By Lemma~\ref{lemma:so-vis},
  it suffices to show that for any client $\cl$,
  \[
    \Big\lvert \txs|_{\cl} \Big\rvert = \infty
      \implies \exists \tvar' \in \txs|_{\cl}.\; \tvar \rel{\vis} \tvar'.
  \]
  By Lemmas~\ref{lemma:uniformity-correct-dc},
  \ref{lemma:uniformity-causal-fence}, \ref{lemma:uniformity-strong},
  and \ref{lemma:uniformity-strongtx-i},
  there exists some time $\realtime$ such that
  \begin{align}
    &\forall \cdrange \in \C.\; \forall n \in \P. \nonumber\\
      &\quad (\forall i \in \D.\; \timestamp(\tvar)[i]
        \le \uniformVC^{n}_{\cdrange}(\realtime)[i])
      \;\land \nonumber\\
      &\quad\;\; \timestamp(\tvar)[\strongentry]
        \le \stableVC^{n}_{\cdrange}(\realtime)[\strongentry].
    \label{eqn:time-t}
  \end{align}
  Since $\Big\lvert \txs|_{\cl} \Big\rvert = \infty$,
  there exists some correct data center $d \in \C$
  to which $\cl$ issues an infinite number of transactions.
  Let $\tvar' \in \txs$ be the first transaction
  issued by client $\cl$ to data center $d$
  which starts after time $\realtime$ such that
  \[
    \lclock(\tvar) < \lclock(\tvar').
  \]
  Thus, by Definition~\ref{def:lco} of $\lcorder$,
  \[
    \tvar \rel{\lcorder} \tvar'.
  \]
  Let $m \triangleq \coord(\tvar')$.
  Since $d$ is correct, by (\ref{eqn:time-t}),
  \begin{align*}
    &(\forall i \in \D.\; \timestamp(\tvar)[i] \le \uniformVC^{m}_{d}(\realtime)[i])
    \;\land \\
    &\;\;\timestamp(\tvar)[\strongentry] \le \stableVC^{m}_{d}(\realtime)[\strongentry].
  \end{align*}
  By Lemma~\ref{lemma:uniformvc-nondecreasing},
  \[
    \forall i \in \D.\;
      \uniformVC^{m}_{d}(\realtime)[i] \le (\uniformVC^{m}_{d})_{\startoftx(\tvar')}[i].
  \]
  By Lemma~\ref{lemma:stablevc-strong-nondecreasing},
  \[
    \stableVC^{m}_{d}(\realtime)[\strongentry] \le
      (\stableVC^{m}_{d})_{\startoftx(\tvar')}[\strongentry].
  \]
  By Lemma~\ref{lemma:ts-start},
  \begin{gather*}
    (\forall i \in \D.\; (\uniformVC^{m}_{d})_{\startoftx(\tvar')}[i]
      \le \timestamp(\startoftx(\tvar'))[i])
    \;\land \\
    (\stableVC^{m}_{d})_{\startoftx(\tvar')}[\strongentry]
      \le \timestamp(\startoftx(\tvar'))[\strongentry].
  \end{gather*}
  Putting it together yields
  \[
    \timestamp(\tvar) \le \timestamp(\startoftx(\tvar')).
  \]
  By Definition~\ref{def:vis-tx} of $\vis$,
  \[
    \tvar \rel{\vis} \tvar'.
  \]
\end{proof}

\subsection{\unistore{} Correctness} \label{ss:unistore-correctness}

\begin{apptheorem} \label{thm:unistore-por}
  \[
    \unistore \models \por.
  \]
\end{apptheorem}

\begin{proof} \label{proof:unistore-por}
  By Theorems~\ref{thm:conflictaxiom}, \ref{thm:cv},
  \ref{thm:ca}, \ref{thm:retval}, and \ref{thm:ev}.
\end{proof}

\subsection{\unistore{} Liveness} \label{ss:unistore-liveness}

\begin{applemma} \label{lemma:migrate-terminate}
  Each client migration \clattach{} to a correct data center
  will eventually terminate, provided that the client managed to
  complete its \fence{} call at its original data center
  just before \clattach.
\end{applemma}

\begin{proof} \label{proof:migrate-terminate}
  Suppose that client $\cl$ in its current data center $d \triangleq \cldc(\cl)$
  issues a \clattach{} event $\attachrange$
  to replica $p^{m}_{\cdrange}$ in correct data center $\cdrange \in \C$.

  Let $e \in C$ be the last commit event
  issued by client $\cl$ before $\attachrange$.
  If $e$ does not exist, then
  \[
    \forall i \in \D.\; \pastVC_{\cl}[i] = 0.
  \]
  Therefore, the wait condition
  \[
    \forall i \in \D \setminus \set{c}.\; \uniformVC^{m}_{c}[i] \ge 0
  \]
  at line~\code{\ref{alg:unistore-replica}}{\ref{line:attach-wait-condition}}
  at replica $p^{m}_{c}$ will eventually hold.
  Thus, the \clattach{} event $\attachrange$ eventually terminates.

  Otherwise, suppose that $e$ is issued to replica $p^{n'}_{d'}$
  in data center $d'$.
  By Lemma~\ref{lemma:pastvc-uniformvc-except-d},
  \begin{align}
    \forall i \in \D \setminus \set{d'}.\;
      (\uniformVC^{n'}_{d'})_{e}[i] \ge (\pastVC_{\cl})_{e}[i].
      \label{eq:attach-dprime}
  \end{align}
  Note that it is possible that $d' \neq d$,
  since there may exist other \clattach{} events between $e$ and $\attachrange$.
  Therefore, we distinguish between the following two cases:
  \begin{itemize}
    \item \textsc{Case I}: $d' = d$. By (\ref{eq:attach-dprime}),
      \begin{align}
        \forall i \in \D \setminus \set{d}.\;
          (\uniformVC^{n'}_{d})_{e}[i] \ge (\pastVC_{\cl})_{e}[i].
          \label{eq:attach-dprime-d}
      \end{align}
    \item \textsc{Case II}: $d' \neq d$.
      Consider the last \clattach{} event, denoted $a'$, before $\attachrange$.
      Suppose that $a'$ is issued to replica $p^{n}_{d}$ in data center $d$.
      By Lemma~\ref{lemma:pastvc-uniformvc-except-d}, when $a'$ terminates,
      \begin{align}
        \forall i \in \D \setminus \set{d}.\;
          &(\uniformVC^{n}_{d})_{a'}[i] \ge (\pastVC_{\cl})_{a'}[i] \nonumber \\
          &\quad = (\pastVC_{\cl})_{e}[i].
          \label{eq:attach-dprime-not-d}
      \end{align}
  \end{itemize}

  Let $\fencerange \in \Fence$ be the \fence{} event issued by client $\cl$
  just before $\attachrange$.
  Suppose that $\fencerange$ is issued to replica $p^{l}_{d}$ in data center $d$.
  By Lemma~\ref{lemma:pastvc-uniformvc-except-d},
  \begin{align}
    (\uniformVC^{l}_{d})_{\fencerange}[d] &\ge (\pastVC_{\cl})_{\fencerange}[d]
      \nonumber\\
      &= (\pastVC_{\cl})_{e}[d].
      \label{eq:barrier-just-before-attach}
  \end{align}

  By (\ref{eq:attach-dprime-d}), (\ref{eq:attach-dprime-not-d}),
  (\ref{eq:barrier-just-before-attach}), and Lemma~\ref{lemma:uniformvc-x},
  eventually for the correct data center $\cdrange$,
  \[
    \forall i \in \D.\;
      \uniformVC^{m}_{\cdrange}[i] \ge (\pastVC_{\cl})_{e}[i].
  \]
  Therefore, the wait condition
  \[
    \forall i \in \D \setminus \set{c}.\;
      \uniformVC^{m}_{c}[i] \ge (\pastVC_{\cl})_{e}[i]
  \]
  at line~\code{\ref{alg:unistore-replica}}{\ref{line:attach-wait-condition}}
  at replica $p^{m}_{c}$ will eventually hold.
  Thus, the \clattach{} event $\attachrange$ eventually terminates.
\end{proof}

\begin{apptheorem} \label{thm:termination}
  Any client event issued at a correct data center
  will eventually terminate.
\end{apptheorem}

\begin{proof} \label{proof:termination}
  Consider any event $e$ issued by client $\cl$
  at a correct data center $c \in \C$.
  It suffices to show that each wait condition in the execution of $e$,
  if any, will eventually hold.
  In the following, we perform a case analysis
  according to the type of event $e$.
  \begin{itemize}
    \item $\textsc{Case I}$: $e \in S$.
      The theorem holds trivially.
    \item $\textsc{Case II}$: $e \in R$.
      If $e \in \intread$, the theorem holds trivially.
      Otherwise, $e \in \extread$.
      By Lemmas~\ref{lemma:knownvc-d-no-bound}
      and \ref{lemma:knownvc-strong-no-bound},
      the wait condition at
      line~\code{\ref{alg:unistore-replica}}{\ref{line:readkey-wait-util-knownvc}}
      for $e$ will eventually hold.
    \item $\textsc{Case III}$: $e \in U$.
      The theorem holds trivially.
    \item $\textsc{Case IV}$: $e \in C_{\causalentry}$.
      Since data center $c$ is correct,
      the wait condition at
      line~\code{\ref{alg:unistore-coord}}{\ref{line:commitcausal-wait-prepareack}}
      will eventually hold.
    \item $\textsc{Case V}$: $e \in C_{\strongentry}$.
      By Lemma~\ref{lemma:uniformvc-c-no-bound},
      the wait condition at
      line~\code{\ref{alg:unistore-replica}}{\ref{line:uniformbarrier-wait-uniformvc-d}}
      will eventually hold.
      Thus, line~\code{\ref{alg:unistore-strong-commit}}{\ref{line:commitstrong-call-uniformbarrier}}
      will eventually terminate.
      Then, by Assumption~\ref{assumption:complete-execution},
      the procedure $\certify$ and thus the event $e$ will eventually terminate.
    \item $\textsc{Case VI}$: $e \in \fence$.
      By Lemma~\ref{lemma:uniformvc-c-no-bound},
      the wait condition at
      line~\code{\ref{alg:unistore-replica}}{\ref{line:uniformbarrier-wait-uniformvc-d}}
      will eventually hold.
    \item $\textsc{Case VII}$: $e \in \clattach$.
      The theorem holds due to \textsc{Case VI} and Lemma~\ref{lemma:migrate-terminate}.
  \end{itemize}
\end{proof}

\fi

\end{document}